\documentclass[11pt, a4paper,reqno]{amsart}

\usepackage{color} \definecolor{bleu_sombre}{rgb}{0,0,0.6}  \definecolor{rouge_sombre}{rgb}{0.8,0,0}\definecolor{vert_sombre}{rgb}{0,0.6,0}

\usepackage[english]{babel}

\usepackage{amsmath,amssymb,amsthm,graphicx,amsfonts,url,color,enumerate,dsfont,stmaryrd,mathabx,mathrsfs}

\usepackage[plainpages=false,colorlinks,linkcolor=bleu_sombre,citecolor=rouge_sombre,urlcolor=vert_sombre,breaklinks]{hyperref}

\theoremstyle{plain}
\newtheorem{theorem}{{Theorem}}[section] 
\newtheorem*{theorem*}{{Theorem}}
\newtheorem{proposition}[theorem]{Proposition}
\newtheorem*{proposition*}{Proposition}
\newtheorem{corollary}[theorem]{Corollary}
\newtheorem*{corollary*}{Corollary}
\newtheorem{lemma}[theorem]{Lemma}
\newtheorem*{lemma*}{Lemma}

\theoremstyle{definition}
\newtheorem{definition}[theorem]{Definition}
\newtheorem*{definition*}{Definition}

\theoremstyle{remark}
\newtheorem{remark}[theorem]{Remark}
\newtheorem*{remark*}{Remark}

\makeatletter

\@addtoreset{equation}{section}  
\makeatother

\renewcommand{\leq}{\leqslant}	\renewcommand{\geq}{\geqslant}
\renewcommand{\bar}[1]{\overline{#1}}
\renewcommand\over[2]{{\,\buildrel #1\over#2\,}}

\newcommand{\ie}{{\it{ie. }}}

\newcommand{\inv}{^{-1}}

\newcommand {\limt}[2]{\xrightarrow[#1 \to #2]{}}

\newcommand{\fonc}[4] { \left\{ \begin{array}{ccc} #1 & \to & #2 \\ #3 & \mapsto & #4 \end{array} \right. }

\newcommand{\abs}[1]{\left\vert #1\right\vert}        
\newcommand{\nr}[1]{\left\Vert #1\right\Vert}         
\newcommand{\innp}[2]{\left< #1 , #2 \right>}         

\newcommand{\Dom}{\Dc}			

\newcommand{\Opwm}{{\mathop{\rm{Op}}}_{h_m}^w}

\newcommand{\pppg}[1] {\left< #1 \right>} 	
\newcommand{\symb} {\Sc}		

\newcommand{\bigo}[2]{\mathop{O}\limits_{#1 \to #2}}
\newcommand{\littleo}[2]{\mathop{o}\limits_{#1 \to #2}}
\newcommand{\simm}[2]{\mathop{\mbox{$\sim$}}\limits_{#1 \to #2}}

\newcommand{\singl}[1]{\left\{ #1 \right\}}		
\newcommand{\Ii}[2] {\llbracket #1,#2 \rrbracket}	
\newcommand{\R}{\mathbb{R}}		\newcommand{\C}{\mathbb{C}}
\newcommand{\N}{\mathbb{N}}

\newcommand{\1}{\mathds 1}

\newcommand{\st}{\,:\,}					

\newcommand{\seq}[2]{\left({#1}_{#2}\right)_{#2 \in\N}} 

\newcommand{\restr}[2]{\left.#1\right|_{#2}}         
\renewcommand{\Re}{\mathop{\rm{Re}}\nolimits}        
\renewcommand{\Im}{\mathop{\rm{Im}}\nolimits}        
\DeclareMathOperator{\Ran}{Ran}	

\DeclareMathOperator{\diag}{diag}                    

\DeclareMathOperator{\Id}{Id}                        
 
\DeclareMathOperator{\supp}{supp}                    

\renewcommand{\a}{\alpha}\renewcommand{\b}{\beta}\newcommand{\g}{\gamma}\newcommand{\G}{\Gamma}\renewcommand{\d}{\delta}\newcommand{\D}{\Delta}\newcommand{\e}{\varepsilon}\newcommand{\z}{\zeta} \newcommand{\y}{\eta}\renewcommand{\th}{\theta}\newcommand{\Th}{\Theta}\renewcommand{\k}{\kappa}\renewcommand{\l}{\lambda}\renewcommand{\L}{\Lambda}\newcommand{\m}{\mu}\newcommand{\n}{\nu}\newcommand{\x}{\xi}\newcommand{\s}{\sigma}\renewcommand{\t}{\tau}\newcommand{\f}{\varphi}\newcommand{\vf}{\phi}\newcommand{\h}{\chi}\newcommand{\p}{\psi}\renewcommand{\o}{\omega}\renewcommand{\O}{\Omega}

\newcommand{\Ac}{{\mathcal A}}\newcommand{\Bc}{{\mathcal B}}\newcommand{\Dc}{{\mathcal D}}\newcommand{\Ec}{{\mathcal E}}\newcommand{\Gc}{{\mathcal G}}\newcommand{\Hc}{{\mathcal H}}\newcommand{\Kc}{{\mathcal K}}\newcommand{\Lc}{{\mathcal L}}\newcommand{\Pc}{{\mathcal P}}\newcommand{\Rc}{{\mathcal R}}\newcommand{\Sc}{{\mathcal S}}\newcommand{\Tc}{{\mathcal T}}\newcommand{\Uc}{{\mathcal U}}\newcommand{\Vc}{{\mathcal V}}\newcommand{\Wc}{{\mathcal W}}\newcommand{\Xc}{{\mathcal X}}

\newcommand{\loc}{{\rm{loc}}}

\newcommand{\qandq}{\quad \text{and} \quad}
\newcommand{\qqandqq}{\qquad \text{and} \qquad}

\newcounter{stepproof}
\newcommand{\stepp}{\stepcounter{stepproof} \noindent {\bf $\bullet$}\quad }

\newcommand{\detail}[1]
{
}

\title{Local energy decay and diffusive phenomenon in a dissipative wave guide}
\author{Julien Royer}
\address{Institut de Math\'ematiques de Toulouse \\ 118, route de Narbonne \\ 31062 Toulouse C\'edex 09 \\ France}
\email{julien.royer@math.univ-toulouse.fr}

\begin{document}


\begin{abstract}
We prove the local energy decay for the wave equation in a wave guide with dissipation at the boundary. It appears that for large times the dissipated wave behaves like a solution of a heat equation in the unbounded directions. The proof is based on resolvent estimates. Since the eigenvectors for the transverse operator do not form a Riesz basis, the spectral analysis does not trivially reduce to separate analyses on compact and Euclidean domains.
\end{abstract}

\maketitle

\tableofcontents

\section{Introduction and statement of the main results}

\newcommand{\EE}{\mathscr{E}}
\newcommand{\HH}{\mathcal{H}}

\newcommand{\aaa}{\a}
\newcommand{\Ha}{H_\aaa}
\newcommand{\hatHa}{\hat H_\a}
\newcommand{\Ho}{H_0}

\newcommand{\HuO}{H^1(\O)}
\newcommand{\HuOp}{H^1(\O)'}
\newcommand{\Huo}{H^1(\o)}
\newcommand{\Huop}{H^1(\o)'}

\newcommand{\tRaz}{\tilde R_a(z)}
\newcommand{\tRat}{\tilde R_a(\t)}
\newcommand{\tRa}{\tilde R_a}
\newcommand{\Raz}{R_a(z)}
\newcommand{\Rat}{R_a(\t)}
\newcommand{\Taz} {T_{az}}
\newcommand{\Tat} {T_{a\t}}
\newcommand{\TatO} {T_{a\t}^\O}
\newcommand{\Haz} {H_{az}}
\renewcommand{\Hat} {H_{a\t}}
\newcommand{\tHaz} {\tilde H_{az}}
\newcommand{\Hah} {H_{h}}
\newcommand{\Ta} {T_{\a}} \newcommand{\TaO} {T_{\a}}  \newcommand{\tTa} {\tilde T_{\a}}
\newcommand{\To} {T_{0}}
\newcommand{\Tam} {T_{a_m}}

\newcommand{\Nder}{m} 
\newcommand{\Ndev}{M} 
\newcommand{\Nfact}{\s}

\newcommand{\Ups}{\Upsilon}
\newcommand{\Rest}{\tilde R_{a,\Ndev}}
\newcommand{\Restu}{\tilde R_{a,1}}

\newcommand{\nn}{n} 
\newcommand{\NN}{N} 
\newcommand{\dd}{d} 

\newcommand{\Lo}{\D}
\newcommand{\LD}{\L}
\newcommand{\LDh}{\L_h}

\newcommand{\sO}{\O}
\newcommand{\OO}{\O}
\newcommand{\OP}{\O_0}

\newcommand{\Csr}{\tilde \C}
\newcommand{\Ucont}{\Uc}

\newcommand{\heat}{{\rm{heat}}}
\newcommand{\Rheatzk}[2]{(\LD-ia\Ups #1)^{-#2}}\newcommand{\Rheat}{\Rheatzk{z}{1}}
\newcommand{\Rcheat}{\Rc_\heat}
\newcommand{\rest}{{\rm{rest}}}
\newcommand{\Rcrest}{\Rc_\rest}


\newcommand{\Rh}{R_{h}}
\newcommand{\Rcha}{\Rc_{h}}
\newcommand{\zoneS}{\mathcal Z}
\newcommand{\Opwx}{{\mathop{\rm{Op}}}_{x,h}^w}

\newcommand{\Tthh}{T_{\th,h}}
\newcommand{\Tah}{T_{\a,h}}
\newcommand{\Tahm}{T_{\a_m,h_m}}
\newcommand{\tTthh}{\tilde T_{\th,h}}
\newcommand{\tTah}{\tilde T_{\a,h}}
\newcommand{\tvm}{\widetilde{v_m}}
\newcommand{\tvmj}[1]{\widetilde{v_m^{#1}}}

\newcommand{\Dd}{D_\nn}
\newcommand{\Ddd}{D_m}
\newcommand{\Dddp}{D_m}
\newcommand{\Dp}{D_m'}
\newcommand{\Dpp}{(D_m')}

Let $\dd,\nn \in \N^*$. We consider a smooth, connected, open and bounded subset $\o$ of $\R^{\nn}$ and denote by $\O$ the straight wave guide $\R^{\dd} \times \o \subset \R^{\dd+\nn}$. Let $a > 0$. For $(u_0,u_1) \in \HuO \times L^2(\O)$ we consider the wave equation with dissipative boundary condition
\begin{equation} \label{wave}
\begin{cases}
\partial_t ^2 u  - \D u  = 0 & \text{on } \R_+ \times  \O,\\
\partial_\n u + a \partial_t u = 0 & \text{on } \R_+ \times  \partial \O,\\
\restr{(u,\partial_t u)}{t=0} = (u_0,u_1) & \text{on } \O.
\end{cases}
\end{equation}

There is already a huge litterature about wave guides, which are of great interest for physical applications. For the spectral point of view we refer for instance to \cite{duclose95,krejcirikk05,borisovk08,bonnetgh11,rabinovichcu13,krejcirikr14} and references therein.

Our purpose in this paper is to study some large time properties for the solution of \eqref{wave}. The analysis will be mostly based on resolvent estimates for the corresponding stationary problem.

\subsection{Local energy decay}

If $u$ is a solution of \eqref{wave} then its energy at time $t$ is defined by 
\begin{equation} \label{def-E}
E(t) =  \int_\OO  \abs{\nabla u (t)}^2  + \int_\OO \abs{\partial_t u (t)}^2.
\end{equation}
It is standard computation to check that this energy is non-increasing, and that the decay is due to the dissipation at the boundary:
\detail
{
\begin{align*}
E'(t) 
& = 2 \int_\OO \partial_t^2 u(t,x) \partial_t \bar u(t,x) \, dx + 2 \int_\OO \nabla  u(t,x) \cdot \nabla \partial_t \bar u(t,x) \, dx \\
& = 2 \int_\OO \partial_t^2 u(t,x) \partial_t \bar u(t,x) \, dx - 2 \int_\OO \D  u(t,x) \partial_t \bar u(t,x) \, dx + 2 \int_{\partial \OO} \partial_\n u(t,x) \partial_t \bar u(t,x) \, d\s(x)  \\
& = - 2a \int_{\partial \OO} \abs{\partial_t  u(t,x)}^2 \, d\s(x) .
\end{align*}
}
\[
E(t_2) - E(t_1) = -2 \int_{t_1}^{t_2} \int_{\partial \O} a \abs {\partial_t u(t)}^2 \, d\s \, dt.
\]

There are many papers dealing with the energy decay for the damped wave equation in various settings. For the wave equation on a compact manifold (with dissipation by a potential or at the boundary), it is now well-known that we have uniform exponential decay under the so-called geometric control condition. See \cite{raucht74,bardoslr92}. Roughly speaking, the assumption is that any trajectory for the underlying classical problem should meet the damping region (for the free wave equation on a subset of $\R^\nn$, the spatial projections of these bicharacteristics are straight lines, reflected at the boundary according to the classical laws of geometrical optics). 
\\

For the undamped wave equation, the energy is conserved. However, on an unbounded domain it is useful to study the decay of the energy on any compact for localized initial conditions. This is equivalent to the fact that the energy escapes at infinity for large times.

The local energy decay for the undamped wave equation has been widely inverstigated on perturbations of the Euclidean space, under the assumption that all classical trajectories escape to infinity (this is the so-called non-trapping condition). For a compact perturbation of the model case we obtain an exponential decay for the energy on any compact in odd dimensions, and a decay at rate $t^{-2d}$ if the dimension $d$ is even. We refer to \cite{laxmp63} for the free wave equation outside some star-shapped obstacle, \cite{morawetzrs77} and \cite{melrose79} for a non-trapping obstacle, \cite{ralston69} for the necessity of the non-trapping condition and \cite{burq98} for a logarithmic decay with loss of regularity but without any geometric assumption. In \cite{bonyh12} and \cite{bouclet11} the problem is given by long-range perturbation of the free wave equation. The local energy (defined with a polynomially decaying weight) decays at rate $O(t^{-2d + \e})$ for any $\e > 0$.\\

Here we are interested in the local energy decay for the damped wave equation on an unbounded domain. Closely related results have been obtained in \cite{alouik02,khenissi03} for the dissipative wave equation outside a compact obstacle of the Euclidean space (with dissipation at the boundary or in the interior of the domain) and \cite{boucletr14,art-dld-energy-space} for the asymptotically free model. The decay rates are the same as for the corresponding undamped problems, but the non-trapping condition can be replaced by the geometric control condition: all the bounded classical trajectories go through the region where the damping is effective.\\

Under a stronger damping assumption (all the classical trajectories go through the damping region, and not only the bounded ones), it is possible to study the decay of the total energy \eqref{def-E}. We mention for instance \cite{burq-joly}, where exponential decay is proved for the total energy of the damped Klein-Gordon equation with periodic damping on $\R^\dd$. This stronger damping condition is not satisfied in our setting, since the classical trajectories parallel to the boundary never meet the damping region.\\

Compared to all these results, our domain $\O$ is neither bounded nor close to the Euclidean space at infinity. In particular the boundary $\partial \O$ itself is unbounded. Our main theorem gives local energy decay in this setting:

\begin{theorem}[Local energy decay] \label{th-loc-dec}
Let $\d > \frac \dd 2 + 1$. Then there exists $C \geq 0$ such that for $u_0 \in H^{1,\d}(\OO)$, $u_1 \in L^{2,\d}(\OO)$ and $t \geq 0$ we have
\[
\nr{ \pppg x^{-\d} \nabla u(t)}_{L^2(\O)} + \nr {\pppg x^{-\d} \partial_t u(t)}_{L^2(\O)} \leq  C \pppg t ^{-\frac {\dd} 2-1} \left( \nr{\pppg x^\d \nabla u_0}_{L^2(\O)} + \nr {\pppg x^\d u_1}_{L^2(\O)} \right),
\]
where $u$ is the solution of the problem \eqref{wave}.
\end{theorem}

Everywhere in the paper we denote by $(x,y)$ a general point in $\O$, with $x \in \R^\dd$ and $y \in \o$. Moreover we have denoted by $L^{2,\d}(\O)$ the weighted space $L^2(\pppg x^{2\d} \, dx \, dy)$ and by $H^{1,\d}(\O)$ the corresponding Sobolev space, where $\pppg \cdot$ stands for $(1+\abs \cdot^2)^{\frac 12}$.\\

We first remark that the power of $t$ in the rate of decay only depends on $\dd$ and not on $\nn$. This is coherent with the fact that the energy has only $\dd$ directions to escape. Although the energy is dissipated in the bounded directions, the result does not depend on their number (nonetheless, we will see that the constant $C$ depends on the shape of the section $\o$).

However, we observe that the local energy does not decay as for a wave on $\R^\dd$. In fact, it appears that the rate of dacay is the same as for the heat equation on $\R^\dd$. This phenomenon will be discussed in Theorem \ref{th-heat} below.\\

As usual for a wave equation, we can rewrite \eqref{wave} as a first order equation on the so-called energy space. For $\d \in \R$ we denote by $\EE^\d$ the Hilbert completion of $C_0^\infty(\bar \O) \times C_0^\infty(\bar \O)$ for the norm
\[
\nr{(u,v)}_{\EE^\d} ^2 = \nr{\pppg x^\d \nabla u}^2_{L^2(\OO)} + \nr {\pppg x^\d v}_{L^2(\OO)}^2.
\]
When $\d = 0$ we simply write $\EE$ instead of $\EE^0$. We consider on $\EE$ the operator
\begin{equation} \label{def-Ac}
\Ac = \begin{pmatrix} 0 & 1 \\ -\D & 0 \end{pmatrix}
\end{equation}
with domain
\begin{equation} \label{dom-Ac}
\Dom(\Ac) = \singl{ (u,v) \in \EE \st (v,-\D u) \in \EE \text{ and } \partial_\n u = ia v \text { on } \partial \OO}.
\end{equation}
Let $u_0,u_1$ be such that $U_0 = (u_0,iu_1) \in \Dom(\Ac)$. Then $u$ is a solution of \eqref{wave} if and only if $U : t \mapsto \big(u(t),i\partial_t u(t)\big)$ is a solution for the problem
\begin{equation} \label{wave-Ac}
\begin{cases}
\partial_t U(t) + i\Ac U(t) = 0,\\
U(0) = U_0.
\end{cases}
\end{equation}
We are going to prove that $\Ac$ is a maximal dissipative operator on $\EE$ (see Proposition \ref{prop-Ac-diss}), which implies in particular that $-i\Ac$ generates a contractions semigroup. Thus the problem \eqref{wave-Ac} has a unique solution $U : t \mapsto e^{-it\Ac} U_0$ in $C^0(\R_+,\Dom(\Ac)) \cap C^1(\R_+,\EE)$. In this setting the estimate of Theorem \ref{th-loc-dec} simply reads
\begin{equation} \label{eq-loc-decay-Ac}
\forall t \geq 0, \quad \nr{e^{-it\Ac} U_0}_{\EE^{-\d}} \leq C \pppg t^{-\frac \dd 2-1} \nr{U_0}_{\EE^\d}.
\end{equation}

We will see that as usual for the local energy decay under the geometric control condition, the rate of decay is governed by the contribution of low frequencies. With a suitable weight, we obtain a polynomial decay at any order if we only consider the contribution of high frequencies. We refer for instance to the result of \cite{wang87} for the self-adjoint Schr\"odinger equation on the Euclidean space. The difficulty with the damped wave equation is that we do not have a functional calculus to localize on high frequencies. Here on a dissipative wave guide we can at least localize with respect to the Laplacian on $\R^\dd$.\\

We denote by $\LD$ the usual Laplacian $-\D_x$ on $\R^\dd$. We also denote by $\LD$ the operator $-\D_x \otimes \Id_{L^2(\o)}$ on $L^2(\O)$.
Let $\h_1 \in C_0^\infty(\R,[0,1])$ be equal to 1 on a neighborhood of 0. For $z \in \C\setminus \singl 0$ we set $\h_z = \h_1(\cdot / \abs{z}^2)$ and
\begin{equation} \label{def-Xc}
\Xc_z = \begin{pmatrix} \h_z(\LD) & 0 \\ 0 & \h_z(\LD) \end{pmatrix} \in \Lc(\EE)
\end{equation}
(where $\Lc(\EE)$ denotes the space of bounded operators on $\EE$).

\begin{theorem} [High frequency time decay] \label{th-high-freq-loc-decay}
Let $\g \geq 0$ and $\d > \g$. Then there exists $C \geq 0$ such that for $U_0 \in \EE^\d$ we have
\[
\forall t \geq 0, \quad \nr{(1-\Xc_1) e^{-it\Ac} U_0}_{\EE^{-\d}} \leq C \pppg t^{-\g} \nr{U_0}_{\EE^\d}.
\]
\end{theorem}

Notice that in the same spirit we could also state the same kind of result for the damped Klein-Gordon equation.

\subsection{Diffusive properties for the contribution of low frequencies}

In Theorem \ref{th-loc-dec} we have seen that the local energy of the damped wave on $\O = \R^\dd \times \o$ decays like a solution of a heat equation on $\R^d$. This is due to the fact that the damping is effective even at infinity. This phenomenon has already been observed for instance for the damped wave equation
\begin{equation} \label{wave-standard}
\partial_t^2 u -\D u + a(x) \partial_t u = 0
\end{equation}
on the Euclidean space $\R^\dd$ itself. For a constant absorption index ($a\equiv 1$), it has been proved that the solution of the damped wave equation \eqref{wave-standard} behaves like a solution of the heat equation
\[
- \D v + \partial_t v = 0.
\]
Roughly, this is due to the fact that for the contribution of low frequencies (which govern the rate of decay for the local energy decay under G.C.C.) the term $\partial_t^2 u $ becomes small compared to $\partial_t u$. See \cite{nishihara03,marcatin03, hosonoo04, narazaki04}. See also \cite{Ikehata-02,Aloui-Ib-Kh} for the damped wave equation on an exterior domain. 
For a slowly decaying absorption index ($a(x) \sim \pppg x^{-\rho}$ with $\rho \in ]0,1]$), we refer to \cite{ikehataty13, Wakasugi-14} (recall that if the absorption index is of short range ($\rho > 1$), then we recover the properties of the undamped wave equation, see \cite{boucletr14,art-dld-energy-space}). Finally, results on an abstract setting can be found in \cite{Chill-Ha-04,Radu-To-Yo-11,nishiyama,Radu-To-Yo-16}.\\ 

Compared to the results in all these papers, we have a damping which is not effective everywhere at infinity but only at the boundary. In particular, the heat equation to which our damped wave equation reduces for low frequencies is not so obvious.\\

For the next result we need more notation. The boundary $\partial \O$ ($\partial \o$, respectively) is a submanifold of $\R^{\dd+\nn}$ (of $\R^\nn$). It is endowed with the structure given by the restriction of the usual scalar product of $\R^{\dd+\nn}$ (of $\R^\nn$) and with the corresponding measure. This is the usual Lebesgue measure on $\partial \O$ (on $\partial \o$).

For $v \in L^2(\O)$ we define $P_\o  {v} \in L^2(\R^{\dd})$ by setting, for almost all $x \in \R^{\dd}$:
\begin{equation} \label{def-Po}
(P_\o {v}) (x) = \frac 1 {\abs \o} \int_\o v(x,\cdot), \quad \text{where } \abs{\o} = \int_\o 1.
\end{equation}
$P_\o v$ can also be viewed as a function in $L^2(\O)$ by setting $(P_\o v) (x,y) = (P_\o v)(x)$. If $v \in H^1(\O)$ we similarly define 
\begin{equation} \label{def-Pdo}
(P_{\partial \o} {v}) (x) = \frac 1 {\abs {\partial \o}} \int_{\partial \o} v(x,\cdot), \quad \text{where } \abs{\partial \o} = \int_{\partial \o} 1.
\end{equation}
We also set 
\begin{equation} \label{def-Ups}
\Ups =  \frac {\abs{\partial \o}} {\abs \o}.
\end{equation}

The purpose of the following theorem is to show that the solution $u$ of \eqref{wave} behaves like the solution of the heat equation 
\begin{equation} \label{heat}
\begin{cases}
a \Ups \partial_t v + \LD  v = 0 & \text{on } \R_+ \times \R^\dd,\\
\restr{v}{t = 0} =  i P_{\partial \o} u_0 + \frac i {a\Ups} P_\o u_1 & \text{on } \R^\dd.
\end{cases}
\end{equation}
We denote by $u_{\heat}$ or $u_{\heat,0}$ the solution of \eqref{heat}:
\begin{equation} \label{def-u-heat}
u_\heat(t) = u_{\heat,0} (t) = e^{-\frac {t\LD}{a\Ups}} \left(i P_{\partial \o} u_0 + \frac i {a\Ups} P_\o u_1 \right), \quad t \geq 0.
\end{equation}
Finally for $\b_x = (\b_{x,1},\dots,\b_{x,\dd}) \in \N^\dd$ we denote by $\partial_x^{\b_x}$ the differential operator $\partial_{x_1}^{\b_{x,1}} \dots \partial_{x_\dd}^{\b_{x,\dd}}$ on $\R^\dd$. The operator $\partial_y^{\b_y}$ is defined similarly on $\o$.

\begin{theorem}[Comparison with the heat equation] \label{th-heat}
Let $(u_0,iu_1) \in C_0^\infty(\bar{\O})^2 \cap \Dom(\Ac)$.
\begin{enumerate}[(i)]
\item There exists $C \geq 0$ such that for $t \geq 0$ we have 
\begin{equation*}
\nr{ \pppg x^{-\d} \nabla ({u-u_{\heat}})(t)}_{L^2(\O)} + \nr {\pppg x^{-\d} \partial_t ({u-u_\heat})(t)}_{L^2(\O)}\\
 \leq  C {\pppg t ^{-\frac {d} 2- 2}} .
\end{equation*}
\item More precisely for $\Ndev \in \N$ there exist $u_{\heat,1},\dots,u_{\heat,\Ndev},\tilde u_{\Ndev +1}$ such that for $t \geq 0$ we have
\[
u(t) = u_{\heat,0}(t) + \sum_{k=1}^\Ndev u_{\heat,k}(t) + \tilde u _{\Ndev + 1}(t),
\]
and for $\e > 0$, $R > 0$, $t\geq 0$, $k \in \Ii 0 \Ndev$, $\b_t \in \{0,1\}$, $\b_x \in \N^\dd$ and $\b_y \in \N^\nn$ with $\b_t + \abs {\b_x} + \abs {\b_y} \leq 1$ there exists $C \geq 0$ such that if $(u_0,iu_1) \in C_0^\infty(\bar{\O_R})^2$ and $\d > 0$ is large enough we have
\[
\nr{\pppg x^{-\d} \partial_t^{\b_t} \partial_x^{\b_x} \partial_y^{\b_y} u_{\heat,k}(t)}_{L^2(\O)} \leq C \pppg t^{-\frac \dd 2 - k - \b_t - \abs{\b_x}} 
\]
and
\[
\nr{\pppg x^{-\d} \partial_t^{\b_t} \partial_x^{\b_x} \partial_y^{\b_y} \tilde u_{\Ndev + 1} (t)}_{L^2(\O)} \leq C \pppg t^{-\frac \dd 2 - \Ndev -1  - \b_t - \abs{\b_x} + \e}.
\]
\end{enumerate}
\end{theorem}

Notice that if we set $U_\heat (t) = (u_\heat(t) , i\partial_t u_\heat(t))$ then the first statement gives
\begin{equation} \label{estim-Uheat}
\nr{e^{-it\Ac} U_0 - U_\heat(t)}_{\EE^{-\d}} \lesssim \pppg t^{-\frac \dd 2 - 2}.
\end{equation}
Since $U_\heat(t)$ is given by the solution of the standard heat equation on $\R^\dd$, we know that it decays like $t^{-\frac \dd 2 - 1}$ in $\Ec^{-\d}$ (see Remark \ref{rem-heat-kernel}). With \eqref{estim-Uheat}, we deduce that the uniform estimate of Theorem \ref{th-loc-dec} is sharp and could not be improved even with a stronger weight.\\

We also observe that $u_{\heat}$ decays slowly if the coefficient $a\Ups$ is large (formally, $u_\heat$ even becomes constant at the limit $a\Ups = +\infty$). This confirms the general idea that a very strong damping weaken the energy decay. Notice that it is natural that the strength of the damping depends not only on the coefficient $a$ which describes how the wave is damped at the boundary but also on the coefficient $\Ups$ which measures how a general point of $\O$ sees the boundary $\partial \O$.
The expression of $u_\heat$ also confirms that the overdamping phenomenon concerns the contribution of low frequencies.

We notice that in Theorem \ref{th-heat} we not only estimate the derivatives of the solution but also the solution itself. To this purpose we introduce $\HH^\d = H^{1,\d} \times L^{2,\d}$, which can be defined as the Hilbert completion of $C_0^\infty(\bar \O)^2$ for the norm 
\[
\nr{(u,v)}_{\HH^\d}^2 = \nr{\pppg x^\d u}_{L^2}^2 + \nr{\pppg x^\d \nabla  u}_{L^2}^2 + \nr{\pppg x^\d v}_{L^2}^2.
\]
We also write $\HH$ for $\HH^0 = H^1(\O) \times L^2(\O)$.

\begin{remark*}
If $U_0 = (u_0,iu_1) \in \Dom(\Ac) \cap \EE^\d$ is such that 
\begin{equation} \label{cond-init-data}
i P_{\partial \o} u_0 + \frac i {a\Ups} P_\o u_1 = 0,
\end{equation}
then $e^{-it\Ac} U_0$ decays at least like $t^{-\frac \dd 2 - 2}$ in $\EE^{-\d}$. This is in particular (but not only) the case if $u_0 \in C_0^\infty(\O)$ and $u_1 = 0$. Because of the semi-group property, the large time asymptotics should not depend on what is considered as the initial time. And indeed, we can check that 
\[
\frac d {dt} \left( i P_{\partial \o} u(t) + \frac i {a\Ups} P_\o \partial_t u(t) \right) = 0,
\]
so \eqref{cond-init-data} holds at time $t = 0$ if and only if it holds with $(u_0,u_1)$ replaced by $(u(t),\partial_t u(t))$ for any $t \geq 0$.
\end{remark*}

      \detail{
      \begin{align*}
      \frac d {dt} \left( i P_{\partial \o} u(t) + \frac i {a \Ups} \partial_t u(t) \right)
      & = \frac i {\abs{\partial \o}} \int_{\partial \o} \partial_t u(t) + \frac {i \abs \o}{a \abs {\partial \o}} \frac 1 {\abs \o} \int_\o \partial_t^2 u(t) \\
      & = \frac i {\abs{\partial \o}} \int_{\partial \o} \partial_t u(t) + \frac { i}{a \abs {\partial \o}}  \int_\o \D u(t) \\
      & = \frac i {a \abs{\partial \o}} \int_{\partial \o} a \partial_t u(t) + \frac { i}{a \abs {\partial \o}}  \int_{\partial \o} \partial_\nu u(t) \\
      & = 0.
      \end{align*}
      }

\subsection{Resolvent estimates}

We are going to prove the estimates of Theorems \ref{th-loc-dec}, \ref{th-high-freq-loc-decay} and \ref{th-heat} from a spectral point of view. After a Fourier transform, we can write $e^{-it\Ac}$ as the integral over $\t = \Re(z)$ of the resolvent $(\Ac-z)\inv$ or, more precisely, of its limit when $\Im(z) \searrow 0$. As usual we will consider separately the contributions of intermediate frequencies ($\abs \t \sim 1$), high frequencies ($\abs \t \gg 1$) and low frequencies ($\abs \t \ll 1$). And as usual the main difficulties will come from low and high frequencies. We begin with the result about intermediate frequencies:

\begin{theorem}[Intermediate frequency estimates] \label{th-inter-freq}
For any $\t \in \R \setminus \singl 0$ the resolvent ${(\Ac-\t)\inv}$ is well defined in $\Lc(\EE)$. By restriction, it also defines a bounded operator on $\Hc$. 
\end{theorem}

Since the resolvent set of $\Ac$ is open, this result implies that around a non-zero frequency (0 belongs to the spectrum of $\Ac$) we have a spectral gap. Thus the question of the limiting absorption principle is irrelevant, we do not even have to work in weighted spaces, and we have similar estimates for the powers of the resolvent. We also remark that, by continuity, the map $\t \mapsto (\Ac-\t)\inv$ is bounded as a function on $\Lc(\EE)$ or $\Lc(\HH)$ on any compact subset of $\R \setminus \singl 0$.\\

Even if any $\t \in \R \setminus \singl 0$ is in the resolvent set, the size of the resolvent and hence of the spectral gap are not necessarily uniform for high frequencies. 

It is known that for high frequencies the propagation of the wave is well approximated by the flow of the underlying classical problem. For the straight wave guide, the horizontal lines (\ie included in $\R^\dd \times \singl y$ for some $y \in \o$) correspond to (spatial projections of) classical trajectories which never see the damping. Thus, we expect that we neither have a spectral gap for high frequencies nor a uniform exponential decay for the energy of the time-dependant solution. However, the classical trajectories which never meet the boundary escape to infinity, so the damping condition is satisfied by all the {bounded} trajectories. In this setting we expect to recover the usual high-frequency estimates known for the undamped wave on the Euclidean space under the non-trapping condition.

\begin{theorem}[High frequency estimates] \label{th-high-freq}
Let $\Nder \in\N$ and $\d > \Nder + \frac 12$. Then there exist $\t_0 \geq 0$ and $C \geq 0$ such that for $\abs \t \geq \t_0$ we have 
\[
\nr{(\Ac-\t)^{-1- \Nder}}_{\Lc(\EE^{\d},\EE^{-\d})} \leq C.
\]
Moreover there exists $\g > 0$ such that if $\h_1$ is supported in $]-\g,\g[$ then for $\abs \t \geq \t_0$ we have
\[
\nr{\Xc_\t(\Ac-\t)^{-1- \Nder}}_{\Lc(\EE,\EE)} \leq C.
\]
We also have similar estimates in $\Lc(\HH^{\d},\HH^{-\d})$ and $\Lc(\HH,\HH)$, respectively.
\end{theorem}

As already mentioned, the limitation in the rate of decay in Theorem \ref{th-loc-dec} is due to the contribution of low frequencies. From the spectral point of view, this comes from the fact that the derivatives of the resolvent are not uniformly bounded up to any order in a neighborhood of 0. The low frequency resolvent estimates will be given in $L^2(\O)$ in Theorem \ref{th-low-freq-bis} below.\\

Thus this paper is mainly devoted to the proofs of resolvent estimates. For this it is more convenient to go back to the physical space $L^2(\O)$. Therefore we first have to rewrite the resolvent $(\Ac-z)\inv$ in terms of the resolvent of a Laplace operator on $L^2(\O)$.\\

Given $z$ in 
\[
\C_+ := \singl{z \in \C \st \Im(z) > 0}
\]
and $\f$ in the dual space $\HuOp$ of $\HuO$ we denote by $u = \tRaz \f$ the unique solution in $\HuO$ for the variational problem 
\begin{equation} \label{variational-pb}
\forall v \in \HuO, \quad \innp{\nabla u}{\nabla v}_{L^2(\O)} - i z \int_{\partial \O} a u \bar v - z^2 \innp{u}{v}_{L^2(\O)} = \innp{\f}{v}_{\HuOp,\HuO}.
\end{equation}
We will check in Proposition \ref{prop-tRaz} that this defines a map $\tRaz \in \Lc(\HuOp,\HuO)$. Moreover, if $\f \in L^2(\O)$ then
\begin{equation} \label{eq-tRaz-Raz}
\tRaz \f = \big( \Haz - z^2 \big)\inv \f, 
\end{equation}
where for $\a \in \C$ we have set 
\begin{equation} \label{def-Ha}
H_\a = -\D
\end{equation} 
on the domain
\begin{equation} \label{dom-Ha}
\Dom(H_\a) = \singl{u \in H^2(\O) \st \partial_\n u = i \a u \text{ on } \partial \O}.
\end{equation}
In Proposition \ref{prop-tRaz-Raz} we will set for $z \in \C_+$
\[
\Raz = \big(\Haz - z^2 \big) \inv.
\]

We consider in $\Lc(\HuO,\HuOp)$ the operator $\Th_a$ defined as follows:
\begin{equation} \label{def-Th}
\forall \f ,\p \in \HuO, \quad \innp{\Th_a \f}{\p}_{\HuOp,\HuO} = \int_{\partial \O} a \f \bar \p.
\end{equation}
Then the link between $(\Ac-z)\inv$ and $\tRaz$ is the following: we will see in Proposition \ref{prop-Ac-diss} that for all $z \in \C_+$ we have on $\HH$
\begin{equation} \label{eq-res-Ac-tRaz}
(\Ac-z)\inv =
\begin{pmatrix}
\tRaz (i\Th_a + z) & \tRaz \\
1 + \tRaz (iz\Th_a + z^2) & z \tRaz
\end{pmatrix}.
\end{equation}
This is of course of the same form as the equality in \cite[Proposition 3.5]{boucletr14}, taking the limit $a(x) \to a \d_{\partial \O}$. However the damping is no longer a bounded operator on $L^2(\O)$ and can only be seen as a quadratic form on $H^1(\O)$.\\

Our purpose is then to estimate the derivatives of $\tRaz$. As in \cite{boucletr14,art-dld-energy-space}, we have to be careful with the dependance on the spectral parameter. And now the derivatives have to be computed in the sense of forms. For instance for the first derivative we have in $\Lc(\HuOp,\HuO)$
\begin{equation} \label{eq-der-tRaz}
\tRa'(z) = \tRaz (i\Th_a + 2z) \tRaz.
\end{equation}

Let us come back to the low frequency estimates and to the comparison with the heat equation. We first observe that for $z \in \C_+$ small, the absorption coefficient $az$ which appears in \eqref{variational-pb} or in the domain of $\Haz$ becomes small. This explains why there is no spectral gap around 0. More precisely, we said that the contribution of low frequencies for the solution of \eqref{wave} behaves like the solution of \eqref{heat}. In our spectral analysis, this comes from the fact that for $z \in \C_+$ small the resolvent $\tRaz $ is close to $\big( -\D - ia \Ups z \big)\inv P_\o$. More precisely, we will prove the following result:

\begin{theorem} \label{th-low-freq-bis}
Let $\Ndev \in \N$. Then there exists an open neighborhood $\Uc$ of 0 in $\C$ such that for $z \in \Uc \cap \C_+$ we can write 
\begin{equation} \label{dev-res-low-freq}
\tRaz  = \sum_{k=0}^\Ndev \sum_{j=0}^k   z^{j+k}  \big( -\D - ia \Ups z \big)^{-j-1} \Pc_{k,j} +  \Rest(z) 
\end{equation}
where the following properties are satisfied.
\begin{enumerate}[(i)]
\item For $k \in \Ii 0 \Ndev$ and $j \in \Ii 0 k$ the operator $\Pc_{k,j}$ belongs to $\Lc(\Huop,\Huo)$. In particular there exists $\sigma \in\C$ such that $\Pc_{k,k} = \s^{k} P_\o$.
\item Let $m \in \N$, $s \in \big[ 0 , \frac {\dd}2 \big[$, $\d > s$, $\b_x \in \N^d$ and $\b_y \in \N^\nn$ be such that $\abs{\b_x} + \abs{\b_y} \leq 1$. Then there exists $C \geq 0$ such that for $z \in \Uc \cap \C_+$ we have 
\[
\nr{\pppg x^{-\d} \partial^{\b_x}_x \partial^{\b_y}_y  \Rest^{(m)} (z)  \pppg x^{-\d}}_{L^2(\O)} \leq C \left( 1 + \abs z^{\Ndev -\Nder + s + \frac {\abs {\b_x} } 2} \right). 
\]
\end{enumerate}
\end{theorem}

The resolvent $\big( -\D - ia \Ups z \big)^{-1}$ which appears in \eqref{dev-res-low-freq} is the resolvent corresponding to the heat equation \eqref{heat}. Uniform estimates for the powers of this resolvent can be deduced from its explicit kernel for $z \notin (-i\R_+)$. 

\begin{proposition} \label{prop-chaleur-intro}
\begin{enumerate}[(i)] 
\item \label{estim-chaleur-diff}
Let $s_0 > 0$, $j \in \N$, $\d > \frac \dd 2 + j$ and $\b \in \N^\dd$ with $\abs \b \leq 1$. Then there exists $C \geq 0$ such that for $s \in ]0,s_0]$ we have 
\begin{equation*} 
\nr{\lim_{\e \searrow 0} \pppg x^{-\d} \partial^\b \left( \big(\L- (s+i\e) \big)^{-1-j} -   \big(\L- (s-i\e) \big)^{-1-j} \right) \pppg x^{-\d}}_{\Lc(L^2(\R^\dd))} \leq C s^{\frac \dd 2 - j - 1 +  \abs \b}.
\end{equation*}
\item \label{estim-chaleur-zero}
Let $j \in \N$, $\abs \b \in \N^\dd$ and $\e > 0$. Let $\d > \frac \dd 2 - \e$. Then there exists $C \geq 0$ and a neighborhood $\Uc$ of 0 in $\C$ such that for $\z \in \Uc \setminus \R_+$ we have 
\[
\nr{\pppg x^{-\d} \partial^\b (-\LD - \z)^{-1-j} \pppg x^{-\d}} \leq C \left(1+\abs \z^{\dd-\e-1-j}\right).
\]
\end{enumerate}
\end{proposition}

The first statement is sharp. It will be used in particular to obtain the sharp estimates for $u_\heat(t)$ and hence for Theorem \ref{th-loc-dec}. This is not the case for the second estimate. In fact we will only use in Proposition \ref{prop-time-decay-heat} the fact that the estimate is of size $o(\abs \z^{-1-j})$.\\

Theorem \ref{th-low-freq-bis} and Proposition \ref{prop-chaleur-intro} will be used to estimate the contribution of low frequencies in Theorems \ref{th-loc-dec} and \ref{th-heat}. In Theorem \ref{th-high-freq-loc-decay} we localize away from low frequencies with respect to the first $\dd$ variables. As expected, we will see that there is no problem with the contribution of low frequencies in this case.

\begin{proposition} \label{prop-high-freq-low-freq}
The map $z \mapsto (1-\Xc_1) (\Ac-z)\inv \in \Lc(\EE)$ extends to a holomorphic function on a neighborhood of 0. The same holds in $\Lc(\HH)$.
\end{proposition}

\subsection{Separation of variables}

In order to prove resolvent estimates on a straight wave guide, it is natural to write the functions of $L^2(\O) \simeq L^2(\R^\dd, L^2(\o))$ as a series of functions of the form $u_m(x) \otimes \f_m(y)$ where $u_m \in L^2(\R^\dd)$ and $\f_m \in L^2(\o)$ is an eigenfunction for the transverse problem.\\

Given $\aaa \in \C$, we consider on $L^2(\o)$ the operator 
\begin{equation} \label{def-Ta}
\Ta = -\D_\o
\end{equation}
on the domain 
\begin{equation} \label{dom-Ta}
\Dom(\Ta) = \singl{u \in H^2(\o) \st \partial_\n u = i \aaa u \text{ on } \partial \o}.
\end{equation}
We have denoted by $\D_\o$ the Laplace operator on $\o$. We also denote by $\Ta$ the operator $\Id_{L^2(\R^{\dd})} \otimes (-\D_\o)$ on $L^2(\O)$ with boundary condition $\partial_\n u = i\a u$ on $\partial \O$. With $\LD$ defined above, this defines operators on $L^2(\O)$ such that 
\begin{equation} \label{eq-Ha-Ta-L}
\Ha = \LD + \TaO.
\end{equation}

The spectrum of $\Ta$ is given by a sequence $(\l_m(\a))_{m\in \N}$ of isolated eigenvalues with finite multiplicities (see Proposition \ref{prop-Ta}). When $\a = 0$ the operators $H_0$ and $T_0$ are self-adjoint. Then there exists an orthonormal basis $\seq \f m$ of $L^2(\o)$ such that $T_0 \f_m = \l_m(0) \f_m$ for all $m \in \N$. For $u \in L^2(\O)$ and almost all $x \in \R^\dd$ we can write 
\[
u(x,\cdot) = \sum_{m \in \N} u_m(x) \f_m
\]
where $u_m \in L^2(\R^\dd)$ for all $m \in \N$. Then for $z \in \C_+$ we have 
\begin{equation} \label{eq-base}
R_0(z) u = \sum_{m \in \N} \big(\LD - z^2 + \l_m(0) \big)\inv u_m \otimes \f_m,
\end{equation}
and by the Parseval identity:
\begin{equation} \label{parseval}
\nr{R_0(z) u}_{L^2(\O)}^2 = \sum_{m \in \N} \nr{\big( \LD - z^2 + \l_m(0) \big)\inv u_m}_{L^2(\R^d)}^2.
\end{equation}
Thus the estimates on $R_0(z)$ follow from analogous estimates for the family of resolvents $\big(\LD - z^2 + \l_m(0) \big)\inv$ on the Euclidean space $\R^\dd$. The situation is not that simple in our non-selfadjoint setting.

The first remark is that we do not necessarily have a basis of eigenfunctions, since for multiple eigenvalues we may have Jordan blocks. Moreover, even when we have a basis of eigenfunctions, this is not an orthogonal family so \eqref{parseval} does not hold. For the dissipative Schr\"odinger equation on a wave guide with one-dimensional section, we proved in \cite{art-diss-schrodinger-guide} that the eigenvalues are simple and that the corresponding sequence of eigenfunctions forms a Riesz basis (which basically means that the equality in \eqref{parseval} can be replaced by inequalities up to multiplicative constants). Then it was possible to reduce the problem to proving estimates for a family of resolvents on $\R^\dd$ as in the self-adjoint case. Here there are two obstructions which prevent us from following the same strategy.

The Riesz basis property in \cite{art-diss-schrodinger-guide} (and more generally in one-dimensional problems) comes from the fact that eigenfunctions corresponding to large eigenvalues $\l_m(\a)$ are close to the orthonormal family of eigenfunctions for the undamped problem. In higher dimension we have ``more small eigenvalues''. More precisely, even if it does not appear in the litterature (to the best of our knowledge), we can expect that a Weyl law holds for the eigenvalues of an operator like $T_\a$ (we recall that for the Laplace operator on a compact manifold of dimension $\nn$ the number of eigenvalues smaller that $r$ grows like $r^{\nn/2}$, see for instance \cite{strauss,zworski}). Thus, when the dimension $\nn$ grows, there are more and more eigenvalues in a given compact and hence more and more eigenfunctions which are far from being orthogonal to each other. We expect that the Riesz basis property no longer holds when $\nn \geq 2$.

The second point is that even if $\dim(\o) = 1$ we have to be careful with the fact that for the wave equation the absorption coefficient grows with the spectral parameter. In \cite[Proposition 3.2]{art-diss-schrodinger-guide} we proved the Riesz basis property uniformly only for a bounded absorption coefficient. Thus, even when $\nn = 1$ we cannot use the Riesz basis property to prove the uniform high frequency estimates.\\

Here the strategy is the following: for low and intermediate frequencies ($\abs \t \lesssim 1$), we first show that we only have to take into account a finite number of eigenvalues $\l_m(a\t)$ (those for which $\Re(\l_m(a\t)) \lesssim \t^2$). For this we have to separate the contributions of different parts of the spectrum. Without writing a sum like \eqref{eq-base}. There are two common ways to localize a problem with respect to the spectrum of an operator. If the operator is self-adjoint, we can use its spectral projections (or, more generally, the functional calculus). If the spectrum has a bounded part $\Sigma$ separated from the rest of the spectrum, we can use the projection given by the Riesz integral on a curve which surrounds $\Sigma$. One of the keys of our proof is to find a way to use simultaneously the facts that $\LD$ is selfadjoint and that $\Ta$ has a discrete spectrum to obtain spectral localizations for $\Ha$.

Once we have reduced the analysis to a finite number of eigenvalues (each of which being of finite multiplicity), we can deduce properties of our resolvent $\Rat$ from analogous properties of ${\big( \LD - \t^2 +\l_m(a\t) \big)\inv} \in \Lc(L^2(\R^\dd))$ as explained above even without self-adjointness. 

However this strategy cannot give uniform estimates for high frequencies, since then we have to take more and more transverse eigenvalues into account. But we still use the same kind of ideas, together with the standard methods of semiclassical analysis (see for instance \cite{zworski} for a general overview). 
Moreover, we will have to separate again the contributions of the different transverse frenquencies $\l_m(a\t)$. If $\abs {\l_m(a\t)} \ll \abs \t^2$ then the spectral parameter $\t^2 - \l_m(a\t)$ in \eqref{parseval} is large. Even if we cannot use \eqref{parseval} in the dissipative case, this suggests that we should use the same kind of ideas as for high frequency resolvent estimates for the operator $\LD$ on $\R^d$. This is no longer the case for the contribution of large eigenvalues of $\Tat$, for which $\abs{\l_m(a\t)} \simeq \abs{\t}^2$. Then we will use the fact that we have a spectral gap at high frequencies for the transverse operator $\Tat$.\\

We state this result in the semiclassical setting. For $\a \in \C$ and $h \in ]0,1]$ we denote by $\Tah$ the operator $-h^2 \Lo$ we domain 
\begin{equation} \label{dom-Tah}
\Dom(\Tah) = \singl{u \in H^2(\o) \st h \partial _\nu u = i \a u \text{ on } \partial \o}.
\end{equation}
Then we have the following result:

\begin{theorem} \label{th-gap-Tah}
There exist $h_0 \in ]0,1]$, $\g > 0$ and $c \geq 0$ such that for $h \in ]0,h_0]$ and
\[
\a ,\z \in ]1-\g , 1 + \g[ + ih ]1-\g , 1 + \g[
\]
the resolvent $(\Tah -\z)\inv$ is well defined in $\Lc(L^2(\o))$ and we have  
\[
\nr{(\Tah -\z)\inv}_{\Lc(L^2(\o))} \leq \frac c h.
\]
\end{theorem}

It seems that this theorem has never been written from the spectral point of view, but it is very closely related to the stabilisation result of \cite{bardoslr92} in a similar setting. We also refer to \cite{lebeau96} and \cite{lebeaur97} which give stabilisation for the wave equation with dissipation in the interior and at the boundary, respectively, but without the geometric control condition. Notice that we are going to use in this paper the contradiction argument of \cite{lebeau96}. We also refer to \cite{sjostrand00} and \cite{anantharaman10} for more precise results about the damped wave equation on a compact manifold without boundary.

Here we have stated our result with a damping effective everywhere at the boundary, but Theorem \ref{th-gap-Tah} should hold if GCC holds for generalized bicharacteristics (with the additionnal assumption that there is no contact of infinite order, see for instance \cite{burq98}). Our setting allows us to provide a less general but less technical proof.

More generally, for our main results we have only considered the simplest case of a damped wave equation on a wave guide with dissipation at the boundary, which already requires quite a long analysis. But many generalizations of this model case would be of great interest (perturbations of the domain $\O$, of the laplace operator $-\D$ on $\O$, of the absorption index, etc.). They are left as open problems in this work. On the other hand the case of a damping in the interior of the domain is easier than the damping at the boundary and could be added here. However it would make the notation heavier so we content ourselves with a free equation in the interior of the domain.\\

The paper is organized as follows. 
We prove in Section \ref{sec-general-properties} the general properties of the operators $\Ac$, $\Ha$ and $\Ta$ which will be used throughout the paper.
In Section \ref{sec-loc-decay} we use the resolvent estimates of Theorems \ref{th-inter-freq}, \ref{th-high-freq} and \ref{th-low-freq-bis} (and Propositions \ref{prop-chaleur-intro} and \ref{prop-high-freq-low-freq}) to prove Theorems \ref{th-loc-dec}, \ref{th-heat} and \ref{th-high-freq-loc-decay}.
Then the rest of the paper is devoted to the proofs of these spectral results.
In Section 3 we show how we can use the discreteness of the spectrum of $\Ta$ and the selfadjointness of $\LD$ to separate the contributions of the different parts of the spectrum of $\Ha$.
Then we deduce Theorem \ref{th-inter-freq} in Section \ref{sec-inter-freq}.
In Section \ref{sec-low-freq} we study the contribution of low frequencies, and in particular we prove Theorem \ref{th-low-freq-bis}. Section \ref{sec-high-freq} is devoted to Theorem \ref{th-high-freq} concerning high frequencies, and we give a proof of Theorem \ref{th-gap-Tah} in Appendix \ref{sec-gap-Taz}.
Finally we give a quick description of the spectum of $\Ta$ when $\nn = 1$ in Appendix \ref{sec-sec-dim1}.

\section{General properties} \label{sec-general-properties}

In this section we prove the general properties which we need for our analysis. In particular we prove all the basic facts about $\Ac$, $\tRaz$ and $\Taz$ which have been mentioned in the introduction.\\

We first recall that an operator $T$ on a Hilbert space $\Kc$ with domain $\Dom(T)$ is said to be accretive (respectively dissipative) if 
\[
\forall u \in \Dom(T), \quad \Re \innp{Tu}u \geq 0 \quad \big( \text{respectively } \Im \innp {Tu}u \leq 0 \big).
\]
Moreover $T$ is said to be maximal accretive (maximal dissipative) if it has no other accretive (dissipative) extension than itself on $\Kc$. With these conventions, $T$ is (maximal) dissipative if and only if $iT$ is (maximal) accretive. We recall that a dissipative operator $T$ is maximal dissipative if and only if $(T-z)$ has a bounded inverse on $\Kc$ for some (and hence any) $z \in \C_+$. In this case we have 
\[
\forall z \in \C_+, \quad \nr{(T-z)\inv} \leq \frac 1 {\Im(z)}
\]
and hence, by the Hille-Yosida theorem (see for instance \cite{engel2}), the operator $-iT$ generates a contractions semigroup $t \mapsto e^{-itT}$. Then, for $u_0 \in \Dom(T)$, the function $t \mapsto e^{-itT} u_0$ belongs to $C^0(\R_+,\Dom(T)) \cap C^1(\R_+,\Kc)$ and is the unique solution for the Cauchy problem 
\[
\begin{cases}
\partial_t u + iT u = 0, \quad \forall t \geq 0,\\
u(0) = u_0.
\end{cases}
\]

\subsection{General properties of \texorpdfstring{$\tRaz$}{Ra(z)}}

We begin with the general properties of the variational problem \eqref{variational-pb}. For $\a \in \C$ and $u,v \in \HuO$ we set 
\begin{equation} \label{def-qa}
q_\a (u,v) = \int_{\partial \O} \a u\bar v  \qandq Q_\a(u,v) = \int_\O \nabla u \cdot \nabla \bar v - i q_\a (u,v).
\end{equation}
We also denote by $q_\a$ and $Q_\a$ the corresponding quadratic forms on $\HuO$, and by $\tilde \D \in \Lc(\HuO,\HuOp)$ the operator corresponding to $-Q_0$: for $u,v \in \HuO$ we have 
\[
\big< {-\tilde \D u},{v}\big>_{\HuOp,\HuO}  = \innp{\nabla u}{\nabla v}_{L^2(\O)}.
\]

\begin{proposition} \label{prop-tRaz}
 Let $z \in \C_+$. Then for $\f \in \HuOp$ the variational problem \eqref{variational-pb} has a unique solution $\tRaz \f \in \HuO$. Moreover the norm of $\tRaz$ in ${\Lc(\HuOp,\HuO)}$ is bounded on any compact of $\C_+$.
\end{proposition}

\begin{proof}
Let $\th = \frac \pi 2 - \arg(z) \in \big]-\frac \pi 2 , \frac \pi 2 \big[$. Then $u \in \HuO$ is a solution of \eqref{variational-pb} if and only if it is a solution of the problem
\begin{equation} \label{variational-pb-theta}
\forall v \in \HuO, \quad Q_{a,z}^\th (u,v) =  \innp{e^{i\th}\f} v,
\end{equation}
where we have set $Q_{a,z}^\th = e^{i\th} (Q_{az}-z^2)$.
This defines a quadratic form on $\HuO$ and for $v \in \HuO$ we have 
\begin{align*}
\Re \big(Q_{a,z}^\th (v,v) \big)
& = \cos(\th)  \nr{\nabla v}_{L^2(\O)}^2  + \abs z \int_{\partial \O} a \abs v^2 - \cos \big(\th+ \arg(z^2) \big) \abs z^2 \nr{v}_{L^2(\O)}^2\\
& \geq \sin \big(\arg(z) \big) \min \big( 1,\abs z^2 \big) \nr{v}_{\HuO}^2.
\end{align*}
According to the Lax-Milgram Theorem, the problems \eqref{variational-pb-theta} and hence \eqref{variational-pb} have a unique solution $u$. Moreover
\[
\nr {u} _{\HuO} \leq \frac {\nr {\f}_{\HuOp}} {\sin \big(\arg(z) \big) \min \big( 1,\abs z^2 \big)},
\]
and the conclusion follows.
\end{proof}

\begin{remark} \label{rem-tRaz}
For $z \in \C_+$ the operator $\tRaz \in \Lc(\HuOp,\HuO)$ is the inverse of ${(-\tilde \D -iz\Th_a - z^2)} \in \Lc(\HuO,\HuOp)$. Its adjoint $\tRaz^* \in \Lc(\HuOp,\HuO)$ is then the inverse of ${(-\tilde \D + i \bar z \Th_a - \bar z^2)}$. For $\p \in \HuOp$ it gives the solution $v = \tRaz^* \p$ of the variational problem 
\[
\forall u \in \HuO, \quad \innp{\nabla v}{\nabla u}_{L^2(\O)} + i \bar z \int_{\partial \O} a v \bar u - \bar z^2 \innp{v}{u}_{L^2(\O)} = \innp{\p}{u}_{\HuOp,\HuO}.
\]
In particular for $\f,\p \in \HuOp$ and $z \in \C_+$ we have
\begin{equation} \label{eq-tRaz-adjoint}
\innp{\tRaz \f}{\p} = \innp{ \f}{\tilde R_a(-\bar z) \p}.
\end{equation}

\end{remark}

The next result concerns the derivatives of $\tRaz$.

\begin{proposition} \label{prop-der-tRaz}
The map $z \mapsto \tRaz \in \Lc \big( \HuOp,\HuO \big)$ is holomorphic on $\C_+$ and its derivative is given by \eqref{eq-der-tRaz}.
More generally, if we set $\Th_a^1 = \Th_a$ and $\Th_a^0 = \Id_{L^2(\O)}$ then for any $\Nder \in \N$ the derivative $\tilde R_a^{(\Nder)}(z)$ is a linear combination of terms of the form 
\begin{equation} \label{terme-dec-RN}
z^{q} \tRaz \Th_a^{\nu_1} \tRaz  \Th_a^{\nu_2} \dots  \Th_a^{\nu_\Nfact} \tRaz,
\end{equation}
where $\Nfact \in \Ii 0 \Nder$ (there are $\Nfact+1$ factors $\tRaz$), $q \in \N$ and $\nu_1,\dots,\nu_\Nfact \in \{ 0,1\}$ are such that 
\begin{equation} \label{eq-dec-RN}
\Nder = 2\Nfact - q - ( \nu_1 + \dots + \nu_\Nfact).
\end{equation} 
\end{proposition}

\begin{proof}
Let $z \in \C_{+}$. For $\z \in \C_{+}$ we set $T_z(\z) = \tRa (\z) - \tRa(z) \in \Lc(\HuOp,\HuO)$. We can check that for $\f \in \HuOp$ and $v \in \HuO$ we have
\begin{multline*}
\innp{\nabla T_z(\z) \f}{\nabla v} - iz q_a(T_z(\z) \f , v) - z^2 \innp{T_z(\z) \f}{v}\\
= i(\z-z) \innp{\Th_a \tRa(\z) \f}{v}_{\HuOp,\HuO} + (\z^2 - z^2) \innp{\tRa(\z) \f} v.
\end{multline*}
Therefore in $\Lc(\HuOp,\HuO)$ we have
\[
\nr{T_z(\z)} = \nr{\tRaz \left(i(\z-z) \Th_a \tRa(\z) + (\z^2 - z^2) \tRa(\z)   \right)} \limt \z z 0,
\]
and then
\[
\nr{\frac {T_z(\z)}{\z - z} - \tRaz (i\Th_a + 2z) \tRaz } \limt \z z 0.
\]
This proves \eqref{eq-der-tRaz}. The general case follows by induction on $\Nder$.
\end{proof}

In the following proposition we explicit the link between the variational problem \eqref{variational-pb} and the operator $\Ha$ defined by \eqref{def-Ha}-\eqref{dom-Ha}. We first need a lemma about the traces on $\partial \O$.

\begin{lemma} \label{lem-trace-O}
Let $\e > 0$. Then there exists $C \geq 0$ such that for all $u \in C_0^\infty(\bar \O)$ we have 
\[
\nr{u}_{L^2(\partial \O)} \leq \e \nr{u}_{H^1(\O)} + C_\e \nr{u}_{L^2(\O)}.
\]
\end{lemma}

This estimate easily follows from the standard trace and interpolation theorems on a bounded domain (see for instance Theorems 1.5.1.2 and 1.4.3.3 in \cite{grisvard}). The case of a wave guide easily follows:

\begin{proof}
Let $s \in \big] \frac 12 , 1 \big[$. By the trace theorem on the smooth bounded subset $\o$ of $\R^n$ there exists $C \geq 0$ such that for all $x \in \R^\dd$ we have 
\[
\int_{\partial \o} \abs{u(x,\cdot)}^2 \leq C \nr{u(x,\cdot)}_{H^s(\o)}^2.
\]
Then by interpolation there exists $C_\e$ such that
\[
\int_{\partial \o} \abs{u(x,\cdot)}^2 \leq \e \nr{u(x,\cdot)}_{H^1(\o)}^2 + C_\e \nr{u(c,\cdot)}_{L^2(\o)}^2.	
\]
The result follows after integration over $x \in \R^\dd$.
\end{proof}

\begin{proposition} \label{prop-tRaz-Raz}
For $z \in \C_{+}$ the operator $\big( \Haz -z^2 \big)$ has a bounded inverse which we denote by 
\begin{equation} \label{def-Raz}
\Raz = \big( \Haz -z^2 \big) \inv \quad \in L^2(\O).
\end{equation}
Then for any $f \in L^2(\O)$ we have 
\[
\tRaz f = \Raz f.
\]
More generally, for $z \in \C_{+}$, $f \in L^2(\O)$, $g \in \HuO$ then $u = \tRaz (f + \Th_a g)$ is the unique solution in $H^2(\O)$ for the problem
\begin{equation} \label{Pb-H2-general}
\begin{cases}
(-\D - z^2) u = f, & \text{on }\O,\\
\partial_\n u = iaz u + ag , & \text{on } \partial \O.
\end{cases}
\end{equation}
\end{proposition}

\begin{proof}
\stepp We first prove that for $\a \in \C_+$ the operator $\Ha$ is maximal accretive. For this we follow the same ideas as in the proof of Proposition 2.3 in \cite{art-diss-schrodinger-guide}. 
By Lemma \ref{lem-trace-O} and Theorem VI.3.4 in \cite{kato} the form $Q_\a$ is sectorial and closed. By the representation theorem (Theorem VI.2.1 in \cite{kato}), there exists a unique maximal accretive operator $\hatHa$ such that $\Dom(\hatHa) \subset \HuO$ and 
\[
\forall u \in \Dom(\hatHa), \forall v \in \HuO, \quad \big<{\hatHa u},{v}\big> = Q_\a (u,v).
\]
Moreover
\[
\Dom(\hatHa) = \singl{u \in \HuO \st \exists f \in L^2(\O), \forall v \in \HuO, Q_\a(u,v) = \innp{f}{v}},
\]
and for $u \in \Dom(\hatHa)$ the corresponding $f$ is unique and given by $f = \hatHa u$.
It is easy to check that the operator $\Ha$ is accretive and that for all $u \in \Dom(\Ha)$ and $v \in \HuO$ we have $\innp{\Ha u}{v} = Q_\a(u,v)$. Thus $\Dom(\Ha) \subset \Dom(\hatHa)$ and $\Ha = \hatHa$ on $\Dom(\Ha)$. Now let $u \in \Dom(\hatHa)$. There exists $f \in L^2(\O)$ such that for all $v \in \HuO$ we have 
\[
\int_\O \nabla u \cdot \nabla \bar v - i  \int_{\partial \O} \a u \bar v  = \int_\O f \bar v ,
\]
As in the proof of Proposition 2.3 in \cite{art-diss-schrodinger-guide}, we can check that $u \in H^2(\O)$ and $\partial_\nu u = i\a u$ on $\Dom(\Ha)$. We omit the details. This proves that $\Dom(\hatHa) \subset \Dom(\Ha)$. Thus $\Ha = \hatHa$ is maximal accretive. 

\stepp If moreover $\Re(\a) > 0$ then $\Ha$ is also dissipative and hence maximal dissipative. Let $z \in \C_+$. If $\Re(z) > 0$ then $\Haz$ is maximal dissipative and $\Im(z^2) > 0$, so the resolvent $\Raz$ is well defined. This is also the case if $\Re(z) < 0$, since then $\Haz^*$ is maximal dissipative and $\Im(z^2) < 0$. And finally $\Haz$ is non-negative and $z^2 > 0$ when $\Re(z) = 0$, so $\Raz$ is well defined for any $z \in \C_+$. Then it is clear that for $f \in L^2$ then $\Raz f$ satisfies \eqref{variational-pb} where $\innp{\f}{v}$ is replaced by $\int f \bar v$, so that $\Raz f = \tRaz f$.

\stepp Now let $z$, $f$, $g$ and $u$ as in the last statement. Then for all $v \in \HuO$ we have 
\begin{equation} \label{super-variational-pb}
\int_\O \nabla u \cdot \nabla \bar v - i z \int_{\partial \O} a u \bar v - z^2 \int_\O u \bar v = \int_\O f \bar v + \int_{\partial\O} a g \bar v.
\end{equation}
Again, we follow the proof of Proposition 2.3 in \cite{art-diss-schrodinger-guide} to prove that $u$ belongs to $H^2(\O)$. The only difference is that we have to take into account the term $-z^2 \innp u v$. For the boundary condition we have to replace \cite[(2.1)]{art-diss-schrodinger-guide} by $\partial_\n u = ia z u + a g$ (notice that the restriction of $g$ on $\partial \O$ belongs to $H^{1/2}(\partial \O)$). This concludes the proof.
\end{proof}

\subsection{General properties of the wave operator}

Now we turn to the properties of the wave operator $\Ac$ defined by \eqref{def-Ac}-\eqref{dom-Ac}. We have to prove that it is a maximal dissipative operator on $\EE$ (to ensure that the problem \eqref{wave-Ac} is well-posed) and to express its resolvent in terms of $\tRaz$.

\begin{proposition} \label{prop-Ac-diss}
The operator $\Ac$ is maximal dissipative on $\EE$. Moreover for $z \in \C_+$ and $F \in \HH \subset \EE$ we have in $\HH$
\begin{equation} \label{expr-res-Ac}
(\Ac-z)\inv F = \begin{pmatrix}
\tRaz (i\Th_a + z) & \tRaz \\
1 + \tRaz (iz\Th_a + z^2) & z \tRaz
\end{pmatrix} F.
\end{equation}
\end{proposition}

\begin{proof}
\stepp For $U = (u,v) \in \Dom(\Ac)$ we have 
\begin{align*}
\innp{\Ac U}{U}_{\EE}
& =  \innp{ \nabla v}{\nabla u}_{L^2(\OO)} + \innp{-\D u}{v}_{L^2(\OO)} = 2 \Re \innp{ \nabla v}{\nabla u}_{L^2(\OO)} -i  \int_{\partial \OO} a \abs v^2 .
\end{align*}
In particular $\Im \innp{\Ac U}{U} \leq 0$, so $\Ac$ is dissipative on $\EE$.

\stepp Let $z \in \C_+$. We first check that $\Ran(\Ac - z)$ is closed in $\EE$. Let $\seq F m$ be a sequence in $\Ran(\Ac -z)$ which converges to some $F \in \EE$. For all $m\in\N$ we consider $U_{m} \in \Dom(\Ac)$ such that $(\Ac-z)U_{m} = F_{m}$. Then for all $m,\tilde m \in \N$ we have on the one hand
\begin{eqnarray} \label{eq-minor-Ac}
\lefteqn{\nr{(\Ac-z)(U_{m}-U_{\tilde m})}^2}\\
\nonumber
&& \geq \nr{\Ac(U_m - U_{\tilde m})}^2 + \abs z^2 \nr{U_m - U_{\tilde m}}^2 - 2 \Re(z) \innp{\Ac(U_m - U_{\tilde m})}{U_m - U_{\tilde m}}\\
\nonumber
&& \geq \y \left(\nr{\Ac(U_m - U_{\tilde m})}^2 + \abs z^2 \nr{U_m - U_{\tilde m}}^2 \right),
\end{eqnarray}
      \detail{
      \begin{align*}
      \nr{(\Ac-z)(U_{m}-U_{\tilde m})}^2
      & = \nr{\Ac(U_{m}-U_{\tilde m})}^2 + \abs z^2 \nr{U_{m}-U_{\tilde m}}^2 - 2\Re \left(\bar z \innp{\Ac(U_{m}-U_{\tilde m})}{U_{m}-U_{\tilde m}} \right).
      \end{align*}
      \begin{align*}
      2\Re \left(\bar z \innp{\Ac(U_{m}-U_{\tilde m})}{U_{m}-U_{\tilde m}} \right)
      & \geq - 2 \Re(z) \Re(\innp{\Ac(U_{m}-U_{\tilde m})}{U_{m}-U_{\tilde m}})\\
      & \geq - 2 \abs {\Re(z)} \nr{\Ac(U_{m}-U_{\tilde m})} \nr{U_{m}-U_{\tilde m}}\\
      & \geq - (1-\y) \nr{\Ac(U_{m}-U_{\tilde m})}^2 - \frac {\abs{\Re(z)}^2} {1-\y} \nr{U_{m}-U_{\tilde m}}^2\\
      & \geq - (1-\y) \nr{\Ac(U_{m}-U_{\tilde m})}^2 - (1-\y) \abs z^2 \nr{U_{m}-U_{\tilde m}}^2\\
      \end{align*}
      }
where
\[
\y = 1 - \frac {\abs{\Re(z)}}{\abs z} > 0.
\]
And on the other hand:
\[
\nr{(\Ac-z)(U_{m}-U_{\tilde m})}^2 = \nr{F_{m}-F_{\tilde m}}^2 \limt {n,m} {+\infty} 0.
\]
This proves that $\seq U n$ is a Cauchy sequence in $\Dom(\Ac)$, which is complete (as can be seen by routine argument). So this sequence converges in $\Dom(\Ac)$ to some $U$, which means that $(\Ac-z)U_{m} \to (\Ac-z)U$. Since we already know that $(\Ac-z)U_{m} = F_{m} \to F$, we have $F = (\Ac-z) U \in \Ran (\Ac-z)$, and hence $\Ran(\Ac-z)$ is closed. Moreover $(\Ac-z)$ is one-to-one according to \eqref{eq-minor-Ac}.

\stepp 
Now we prove that $\Ran(\Ac-z)$ is dense in $\EE$.
Let $F = (\tilde f, \tilde g)\in \HH$ and define $U = (u,v)$ as the right-hand side of \eqref{expr-res-Ac}. By Proposition \ref{prop-tRaz-Raz} we have $u \in H^2(\O)$ and $v \in \HuO$. Moreover, by the boundary condition in \eqref{Pb-H2-general} and the fact that $\tRaz \tilde g = \Raz \tilde g \in \Dom(\Haz)$ we have on $\partial \O$:
\[
\partial_\n u = iaz \big( \tRaz (i\Th_a + z) \tilde f + \tRaz \tilde g \big) + i a \tilde f = i a v.
\]
This proves that $U \in \Dom(\Ac)$. Then it is not difficult to check that $(\Ac-z) U = F$, which implies that $F \in \Ran(\Ac-z)$. Since $\HH$ is dense in $\EE$, this proves that $(\Ac-z)$ has a bounded inverse in $\Lc(\EE)$. And since we have already checked \eqref{expr-res-Ac}, the proof is complete.
\end{proof}

As already mentioned, Proposition \ref{prop-Ac-diss} implies in particular that $-i\Ac$ generates a contractions semigroup. Thus for $U_0 \in \Dom(\Ac)$ the problem \eqref{wave-Ac} has a unique solution $U : t \mapsto e^{-it\Ac} U_0$ in $C^0(\R_+,\Dom(\Ac)) \cap C^1(\R_+,\EE)$. \\

\subsection{General properties on the section \texorpdfstring{$\o$}{omega}.}

In this paragraph we describe in particular the transverse operator $\Ta$. It is not selfadjoint, but the discreteness of its spectrum will be crucial to localize spectrally with respect to $\Ha \simeq \L \otimes \Ta$.

\begin{proposition} \label{prop-Ta}
Let $\a \in \C$. The spectrum of $T_\a$ is given by a sequence $\big(\l_m(\a) \big)_{m\in\N}$ of eigenvalues with finite multiplicities. Moreover there exist $\g > 0$ and $\th \in \big[0,\frac \pi 2 \big[$ such that all these eigenvalues belong to the sector 
\begin{equation} \label{sector}
\singl{\l \in \C \st \abs{\arg (\l + \g)} \leq \th}.
\end{equation}
In particular $\Re(\l_m(\a)) \limt m \infty +\infty$. If moreover $\Im(\a) \geq 0$ then we can take $\g = 0$ (the eigenvalues have non-negative real parts).
\end{proposition}

\begin{proof}
Since $\o$ is bounded the operator $\Ta$ has a compact resolvent. Therefore its spectrum is given by a discrete set of eigenvalues with finite multiplicities. Since the operator $T_\a$ is maximal sectorial (this is proved exactly as for $\Ha$), the spectrum of $T_\a$ is included is a sector of the form \eqref{sector}. If moreover $\Im(\a) \geq 0$ then it is easy to see that $\Ta$ is accretive, so that we can take $\g = 0$.
\end{proof}

As on $\O$ we can work in the sense of forms. The operator $\Ta$ corresponds to the quadratic form defined as $Q_\a$ in \eqref{def-qa} but on $\o$ instead of $\O$. We still denote by $\Th_\a$ the operator defined as in \eqref{def-Th} but on $\Lc(\Huo,\Huop)$. Then we set 
\begin{equation} \label{def-tTa}
\tTa = -\D_\o -i\Th_\a \in \Lc(\Huo,\Huop).
\end{equation}
At least if $\Re(\z) < -\g$ the operator $(\tTa - \z) \in \Lc(\Huo,\Huop)$ has an inverse $(\tTa - \z)\inv \in \Lc(\Huop,\Huo)$. For $\f \in \Huop$ then $u = (\tTa - \z)\inv\f$ is the unique solution of 
\begin{equation} \label{variational-pb-o}
\forall v \in \Huo, \quad \innp{\nabla u}{\nabla v}_{L^2(\o)} - i \int_{\partial\o} \a u \bar v - \z \innp{u}{v}_{L^2(\o)} = \innp{\f}{v}_{\Huop,\Huo}.
\end{equation}
And for $\f \in L^2(\o)$ we have 
\[
(\Ta-\z)\inv \f = (\tTa - \z)\inv \f.
\]
In the following proposition we denote by $\s(\cdot)$ the spectrum of an operator and write $H^0(\o)$ for $L^2(\o)$.

\begin{lemma} \label{lem-res-Ta-H1}
Let $\a \in \C$ and $\z \in \C \setminus \s(\Ta)$. Then the inverse $(\tTa - \z)\inv$ of $(\tTa - \z)$ is well defined in $\Lc(\Huop,\Huo)$. Moreover there exists $C \geq 0$ such that for $\Re(\z) \leq -C$ and $\b_1,\b_2 \in \{0,1\}$ we have 
\[
\nr{(\tTa - \z)\inv}_{\Lc(H^{\b_1}(\o)',H^{\b_2}(\o))} \leq C \abs{\Re(\z)}^{\frac {\b_1 + \b_2} 2 - 1}.
\]
\end{lemma}

\begin{proof}
Let $\f \in L^2(\o)$ and $ u = (\Ta - \z)\inv\f$. 
We know that $\bar \z$ is not an eigenvalue of $\Ta^*$, so the resolvent $(\Ta^* - \bar \z)\inv$ exists and belongs in particular to $\Lc(L^2(\o),\Huo)$. By duality we obtain that $(\Ta-\z)\inv$ extends to a bounded operator from $\Huop$ to $L^2(\o)$, and hence 
\[
\nr{u}_{L^2(\o)} \lesssim \nr{\f}_{\Huop}.
\]
Let $\b_1 \in \{0,1\}$ and $s \in \big] \frac 12 , 1 \big[$. We can write \eqref{variational-pb-o} with $v = u$. By the trace and interpolation theorems (see the proof of Lemma \ref{lem-trace-O}) there exists $C \geq 0$ (which does not depend on $\f$ or $\z$ but depends on $\a$) such that
\begin{align*}
\nr{\nabla u}_{L^2(\o)}^2 
& \leq \abs {\a} \nr{u}_{L^2(\partial \o)}^2 + \abs \z \nr{u}_{L^2(\o)}^2 + \nr{u}_{H^{\b_1}(\o)} \nr{\f}_{H^{\b_1}(\o)'}\\
& \leq \frac 12 \nr{u}_{\Huo}^2 + (C + \abs \z) \nr{u}_{L^2(\o)}^2 + \nr{u}_{H^{\b_1}(\o)} \nr{\f}_{H^{\b_1}(\o)'},
\end{align*}
and hence
\[
\nr{u}_{\Huo}^2 \leq 2 (C + \abs \z + 1) \nr{u}_{L^2(\o)}^2 + \nr{u}_{H^{\b_1}(\o)} \nr{\f}_{H^{\b_1}(\o)'}.
\]
Applied with $\b_1 = 1$, this proves that $(\Ta-\z)\inv$ extends to a bounded operator in $\Lc(\Huop,\Huo)$. Then we can check that this defines an inverse for $(\tTa-\z)$, which proves the first statement.

When $\b_1 = \b_2 = 0$ the estimate of the lemma follows from the standard resolvent estimate applied to the maximal accretive operator $\Ta + \g$. 
From the above inequality applied with $\b_1 = 0$ we deduce the estimate in $\Lc(L^2(\o),\Huo)$. The estimate in $\Lc(\Huop,L^2(\o))$ follows by duality, and finally we use the above estimate with $\b_1 = 1$ to deduce the estimate in $\Lc(\Huop,\Huo)$.
\end{proof}

We finish this section by recording some basic properties of the projection $P_\o$ defined in \eqref{def-Po}:

\begin{lemma} \label{lem-Po}
\begin{enumerate}[(i)]
\item If $u \in \HuO$ then $P_\o u \in \HuO$. Moreover we have 
\[
\nabla_x P_\o u = P_\o \nabla_x u \qandq \nabla_y P_\o u = 0.
\]
\item For $u \in \HuO$ we have in $L^2(\O)$
\[
P_\o \Th_a u = a\Ups P_{\partial \o} u
\]
\end{enumerate}
\end{lemma}

\begin{proof}
Let $u \in \HuO$. The first statement follows from the theorem of differentiation under the integral sign and the fact that $P_\o u (x,y)$ does not depend on $y \in \o$. 
By duality, $P_\o$ defines a bounded operator on $\Huop$. Then for all $v \in \HuO$ we have 
\begin{align*}
\innp{P_\o \Th_a u}{v}_{\HuOp,\HuO}
& = \innp{\Th_a u}{P_\o v}_{\HuOp,\HuO} = a \int_{\R^\dd} \int_{\partial \o} u \times \left( \frac 1 {\abs \o} \int_\o \bar v \right)\\
& =  a\Ups  \int_{\R^\dd}\int_\o (P_{\partial \o} u) \bar v = \innp{a\Ups P_{\partial \o} u}{v}_{L^2(\O)}.
\end{align*}
In particular $P_\o \Th_a u$ belongs to $L^2(\O)$. This concludes the proof of the lemma.
\end{proof}

\section{Local energy decay and comparison with the heat equation} \label{sec-loc-decay}

In this section we use the resolvent estimates of Theorems \ref{th-inter-freq}, \ref{th-high-freq}, \ref{th-low-freq-bis} and Propositions \ref{prop-chaleur-intro}, \ref{prop-high-freq-low-freq} to prove Theorems \ref{th-loc-dec}, \ref{th-high-freq-loc-decay} and \ref{th-heat}.\\

As in the Euclidean case, the proofs rely on the propagation at finite speed for the wave equation:

\begin{lemma} \label{lem-prop-tps-fini}
Let $\d \geq 0$ and $T > 0$. Then there exists $C_T \geq 0$ such that for $t \in [0,T]$ and $U_0 \in \EE^\d$ we have 
\[
\nr{e^{-it\Ac} U_0}_{\EE^\d} \leq C_T \nr{U_0}_{\EE^\d}.
\]
\end{lemma}

The proof of this lemma is the same as in the Euclidean space (see \cite{art-dld-energy-space}). We recall the idea:

\begin{proof}
For $\d \in \R_+$, $r_1,r_2 \in \R$ and $(u,v) \in \EE^\d$ we set 
\[
\nr{(u,v)}_{\EE^\d(r_1,r_2)}^2 = \int_{r_1 \leq \abs x \leq r_2} \int_{y \in \o} \pppg x^{2\d} \big( \abs{\nabla u (x,y)}^2 + \abs{v(x,y)}^2 \big) \, dy \, dx.
\]
Let $U_0 \in \Dom(\Ac)$ and let $U$ be the solution of \eqref{wave-Ac}. For $r_1,r_2$ with $r_1 \leq r_2$, $t \geq 0$ and $s \in [0,t]$ we can check that 
\[
\frac d {ds} \nr{U(t-s)}_{\EE^0(r_1-s,r_2+s)}^2 \geq 0,
\]
and hence 
\[
\nr{U(t)}_{\EE^0(r_1,r_2)} \leq \nr{U_0}_{\EE^0(r_1-t,r_2+t)}.
\]
Then if $U_0 \in \Dom(\Ac) \cap \EE^\d$ we have for $t \in [0,T]$ 
\begin{align*}
\nr{e^{-it\Ac}U_0}^2_{\EE^\d} 
& \leq  \pppg T ^ {2\d} \nr{e^{-it\Ac}U_0}^2_{\EE^0(0,T)} + \sum_{n \in \N} \pppg {n+T+1}^{2\d} \nr{e^{-it\Ac}U_0}^2_{\EE^0(T+n,T+n+1)}\\
& \leq  \pppg T ^ {2\d} \nr{U_0}^2_{\EE^\d} + \sum_{n \in \N} \frac {\pppg {n+T+1}^{2\d}}{\pppg n^{2\d}} \nr{U_0}^2_{\EE^\d(n,2T+n+1)}\\
& \lesssim \nr{U_0}_{\EE^\d}.
\end{align*}
We conclude the proof by density of $\Dom(\Ac) \cap \EE^\d$ in $\EE^\d$.
\end{proof}

Let $U_0 \in \Dom(\Ac)$. We assume that the two components of $U_0$ are compactly supported (we give the proofs for such initial conditions, and the results of Theorems \ref{th-loc-dec} and \ref{th-high-freq-loc-decay} will follow by density). We denote by $U(t)$ the solution of \eqref{wave-Ac}. Let $\th \in C^\infty(\R, [0,1])$ be equal to 0 on $]-\infty,1[$ and equal to 1 on $]2,+\infty[$. For $\m > 0$ and $t \in \R$ we set 
\[
U_{0,\m}(t) = \1 _{\R_+}(t) e^{-t\m} U(t) \qandq U_{1,\m}(t) = \th(t) e^{-t\m} U(t).
\]

Let $\t \in \R$ and $z = \t +i\m$. We multiply \eqref{wave-Ac} by $e^{itz} \1_{\R_+}(t)$ (or $e^{itz} \th(t)$, respectively) and take the integral over $\t \in \R$. After a partial integration we get for $j \in \{0,1\}$
\begin{equation} \label{eq-Fourier-U-V}
W_j(z) := \int_\R  e^{it\t} U_{j,\m}(t) \, dt =  -i (\Ac-z)\inv V_j(z), 
\end{equation}
where
\[
V_0(z) = U_0  \qandq V_1(z) = \int_1^2 \th'(t) e^{itz} U(t) \, dt.
\]

Notice that $U_{0,\m}(t)$ and $U_{1,\m}(t)$ coincide for $t \geq 2$. The interest of $U_{0,\m}(t)$ is that the source term $V_0(z)$ is exactly given by the initial data $U_0$. This is necessary to obtain the nice expression of $u_{\heat}$ in Theorem \ref{th-heat}. However we use a sharp cut-off in the definition, and the lack of smoothness implies a lack of decay for its Fourier transform. Therefore we will only obtain estimates with a loss of derivative. To obtain uniform estimates as required in Theorem \ref{th-loc-dec} we shall rather use $U_{1,\m}(t)$, defined with a smooth cut-off in time. The difference will appear clearly in Proposition \ref{prop-td-high-freq}.\\

Let $\th_0 \in C_0^\infty(\R,[0,1])$ be supported in [-3,3] and equal to 1 on a neighborhood of [-2,2], and $\th_\infty = 1 - \th_0$. For $R \geq 1$ and $\t \in \R$ we set $\th_{0,R}(\t) = \th_0(\t/ R)$.\\

Let $\m > 0$. The map $t \mapsto U_{j,\m}(t)$ belongs to $L^2(\R,\EE)$ for any $j \in \{0,1\}$. Since $\t \mapsto V_1(\t + i\m)$ decays at least like $\pppg \t \inv$ and $(\Ac-(\t+i\m))\inv$ is uniformly bounded in $\Lc(\EE)$ (its norm is not greater than $1/\m$), this is also the case for the map $\t \mapsto W_1(\t+i\m)$. For $W_0$ we can write for $z \in \C_+$
\begin{equation} \label{eq-loss-derivative}
(\Ac-z)\inv U_0 = \frac 1 {z-i} \left( (\Ac-z)\inv (\Ac-i) - 1 \right) U_0.
\end{equation}
This proves that the map $\t \mapsto W_0(\t+i\m)$ also belongs to $L^2(\R,\EE)$. Thus we can inverse the relations \eqref{eq-Fourier-U-V}: if for $j \in \{0,1\}$ and $R \geq 1$ we set 
\[
U_{j,\m,R}(t)  = \frac 1 {2\pi} \int_{\t \in \R} \th_{0,R}(\t)  e^{-it\t} W_j(\t + i\m) \, d\t,
\]
then we have 
\begin{equation} \label{eq-Uj-UjR}
\nr{U_{j,\m}-  U_{j,\m,R}}_{L^2(\R,\EE)} \limt R {+\infty} 0.
\end{equation}
The same applies in $L^2(\R,\EE^{-\d})$ for any $\d \geq 0$. Moreover these functions are continuous, so if we can prove that for some function $\rho$ and some $\d \geq 0$ we have $\nr{U_{j,\m,R}(t)}_{\EE^{-\d}} \leq \rho (t)$ uniformly in $R \geq 1$ and $\m > 0$, this will imply that $U_j(t)$ satisfies the same estimate for all $t > 0$.\\

We deal separately with the contributions of low and high-frequencies. For $j \in \{0,1\}$, $t \geq 0$ and $R \geq 1$ we write $U_{j,\m,R}(t)$ as the sum of 
\[
U_{j,\m,\infty,R}(t) =  \frac 1 {2\pi} \int_{\t \in \R} e^{-it\t} \th_{0,R}(\t) \th_\infty(\t) W_j(\t+i\m) \, d\t
\]
and 
\[
U_{j,\m,0}(t) =  \frac 1 {2\pi} \int_{\t \in \R} e^{-it\t} \th_0(\t) W_j(\t+i\m) \, d\t.
\]

\begin{proposition}[Contribution of high frequencies] \label{prop-td-high-freq}
Let $\g \geq 0$ and $\d > \g$. Then there exists $C \geq 0$ which does not depend on $U_0$ or $\m > 0$ and such that for $j \in \{ 0,1\}$, $t \geq 0$ and $R > 1$ we have 
\[
\nr{(1-\Xc_1) U_{j,\m}(t)}_{\EE^{-\d}} \leq C \pppg t^{-\g} \nr{(\Ac-i)^{1-j} U_0}_{\EE^\d}
\]
and
\[
\nr{\Xc_1 U_{j,\m,\infty,R}(t)}_{\EE} \leq C \pppg t^{-\g} \nr{(\Ac-i)^{1-j} U_0}_{\EE}.
\]
If moreover $\g \geq \frac 12$ and $\d > \g +1$ then the same estimates hold with $\EE$ replaced by $\HH$ everywhere.
\end{proposition}

We recall that $\Xc_1$ was defined in \eqref{def-Xc}.
Notice that the first statement applied with $j=1$ gives Theorem \ref{th-high-freq-loc-decay}.

\begin{proof}
Let $\Nder \in \N$ and $\d > \Nder + \frac 12$. With partial integrations we see that $(it)^\Nder {(1-\Xc_1)} U_{j,\m,R}(t)$ is a linear combination of terms of the form 
\[
U_{j,\m,R}^{\Nder_0,\Nder_1,\Nder_2}(t) := \int_{\t \in \R} e^{-it\t} \th_{0,R}^{(\Nder_0)}(\t) (1-\Xc_1)(\Ac-(\t+i\m))^{-1-\Nder_1} V_j^{(\Nder_2)}(\t+i\m) \, d\t,
\]
where $\Nder_0,\Nder_1,\Nder_2 \in \N$ are such that $\Nder_0 + \Nder_1 + \Nder_2  = \Nder$. By the Plancherel Theorem, Theorem \ref{th-inter-freq}, Theorem \ref{th-high-freq}, Proposition \ref{prop-high-freq-low-freq} and Lemma \ref{lem-prop-tps-fini} we obtain for $j = 1$:
\begin{align*}
\int_{\R} \nr{U_{1,\m,R}^{\Nder_0,\Nder_1,\Nder_2}(t)}^2_{\EE^{-\d}} \, dt
& \lesssim \int_{\R} \nr{(1-\Xc_1) (\Ac-(\t+i\m))^{-1-\Nder_1} V_1^{(\Nder_2)}(\t+i\m)}_{\EE^{-\d}}^2 \, d\t \\
& \lesssim  \int_{\R} \nr{ V_1^{(\Nder_2)}(\t+i\m)}_{\EE^{\d}}^2 \, d\t
 \lesssim \nr{U_0}_{\EE^{\d}}^2 \, d\t.
\end{align*}
For $j=0$ we use \eqref{eq-loss-derivative}. This costs a derivative but improves the decay of $W_0$, so that we similarly obtain
\[
\int_{\R} \nr{U_{0,\m,R}^{\Nder_0,\Nder_1,\Nder_2}(t)}^2_{\EE^{-\d}} \lesssim \nr{(\Ac-i) U_0}_{\EE^{\d}}^2.
\]
The end of the proof follows the usual strategy. There exists $C \geq 0$ (which does not depend on $U_0$, $\m > 0$ or $R \geq 1$) and $t_0 \in [0,1]$ (which depends on $U_0$) such that 
\[
\nr{U_{j,\m,R}^{\Nder_0,\Nder_1,\Nder_2}(t_0)}_{\EE^{-\d}} \leq C  \nr{(\Ac-i)^{1-j} U_0}_{\EE^{\d}}.
\]
Then we check that for $t \geq 1$ and $s \in [t_0,t]$ we have 
\[
\frac {\partial}{\partial s} \left(e^{-i(t-s)\Ac} U_{j,\m,R}^{\Nder_0,\Nder_1,\Nder_2}(s) \right) = - \m e^{-i(t-s)\Ac} U_{j,\m,R}^{\Nder_0,\Nder_1,\Nder_2}(s) + i e^{-i(t-s)\Ac} (\Ac-(\t+i\m)) U_{j,\m,R}^{\Nder_0,\Nder_1,\Nder_2}(s).
\]
As above we can check that 
\[
\int_0^t \nr{\frac {\partial}{\partial s} \left(e^{-i(t-s)\Ac} U_{j,\m,R}^{\Nder_0,\Nder_1,\Nder_2}(s) \right)}_{\EE^{-\d}}^2 \, dt \lesssim \nr{(\Ac-i)^{1-j} U_0}_{\EE^{\d}}^2 .
\]
Since for $t \geq 1$ we have 
\[
U_{j,\m,R}^{\Nder_0,\Nder_1,\Nder_2}(t) = e^{-i(t-t_0)\Ac} U_{j,\m,R}^{\Nder_0,\Nder_1,\Nder_2}(t_0) + \int_{t_0}^t \frac {\partial}{\partial s} \left( e^{-i(t-s)\Ac} U_{j,\m,R}^{\Nder_0,\Nder_1,\Nder_2}(s) \right) \, ds,
\]
we obtain 
\[
\nr{U_{j,\m,R}^{\Nder_0,\Nder_1,\Nder_2}(t)}_{\EE^{-\d}} \lesssim \pppg t^{\frac 12} \nr{(\Ac-i)^{1-j} U_0}_{\EE^{\d}}.
\]
This proves that for $\Nder \in \N$, $\d > \Nder + \frac 12$ and $R \geq 1$ we have 
\begin{equation} \label{estim-time-decay-perte}
\nr{(1-\Xc_1) U_{j,\m,R}(t)}_{\EE^{-\d}} \lesssim \pppg t^{\frac 12 - \Nder} \nr{(\Ac-i)^{1-j} U_0}_{\EE^{\d}}.
\end{equation}
Taking the limit $R \to \infty$ gives the first estimate when $\g \in \N + \frac 12$ and $\d > \g +1$. The case $\g \geq \frac 12$ follows by interpolation. Up to now, everything holds with $\EE$ replaced by $\HH$, so we have proved the last statement of the proposition for $(1-\Xc_1) U_{j,\m}$. In the (weighted) energy space(s), we obtain the estimate with $\g \geq 0$ and $\d > 0$ by interpolation between \eqref{estim-time-decay-perte} (applied with $\Nder$ large and $\d_\Nder \in \big] \Nder + \frac 12, \Nder + 1\big[$) and the trivial bound $\nr{U_{j,\m}(t)}_\EE \leq \nr{U_0}_{\EE}$.
The estimates on $U_{j,\m,\infty,R}$ are proved similarly, except that with the cut-off $\th_\infty$ we do not have to worry about low frequencies. Moreover we do not use any weight (see the second statement of Theorem \ref{th-high-freq}) so we have polynomial decay at any order.
\end{proof}

After Proposition \ref{prop-td-high-freq}, it remains to estimate $\Xc_1 U_{j,\m,0}(t)$. In fact we estimate $U_{j,\m,0}(t)$. For this we estimates separately the contributions of the different terms in the developpement of $\tRaz$ given by Theorem \ref{th-low-freq-bis}. Using Proposition \ref{prop-chaleur-intro}, we first estimate the terms involving the heat resolvent.

\begin{proposition} \label{prop-time-decay-heat}
Let $j,k \in \N$, $\b_x \in \N^\dd$ with $\abs {\b_x} \leq 1$ and $\d > \frac \dd 2 + j$. Then there exists $C \geq 0$ such that for $\m > 0$ and $t > 0$ we have 
\[
\nr{\int_\R e^{-it\t} \th_0(\t) z^{j+k} \pppg x^{-\d} \partial^{\b_x} \big(\LD - ia\Ups z \big)^{-1-j} \pppg x^{-\d} \, d\t  }_{\Lc(L^{2}(\O)))} \leq C t^{-\frac \dd 2 - k - \abs {\b_x}},
\]
where $z$ stands for $\t + i \m$.
\end{proposition}

\begin{remark}
There exists a constant $\tilde C$ which does not depend on $a$ and $\Ups$ and such that the constant $C$ of the Proposition is of the form $C = \tilde C (a\Ups)^{\frac {\dd} 2 + k - 1 + \abs {\b_x}}$. This confirms the observation that the decay is slow when the absorption is strong.
\end{remark}

\begin{proof}[Proof of Proposition \ref{prop-time-decay-heat}]
Let $\m > 0$. For $t > 0$ we denote by $I_\m(t) \in \Lc(L^2(\R^\dd))$ the integral which appears in the statement of the proposition. For $z \in \C \setminus (-i\R_+)$ we set 
\[
F(z) =  \th_0 \big( \Re(z) \big)  z^{j+k} \pppg x^{-\d} \partial^{\b_x} \big(\LD-ia\Ups z\big)^{-1-j} \pppg x^{-\d}.
\]
This defines a function on $\C \setminus (-i\R_+)$ which vanishes outside $\big(]-3,3[ + i\R \big) \setminus (-i\R_+)$. Moreover $F$ is holomophic on $\big(]-2,2[ + i\R \big) \setminus (-i\R_+)$, so for $\e \in ]0,1[$ we have 
\[
e^{t\m} I_\m(t) = \int_{\G_{\m,\e}} e^{-itz}  F(z) \, dz,
\]
where $\G_{\m,\e}$ is the contour described by Figure \ref{fig-contour}. In particular, for $\abs{\Re(z)} > \e$ the curve is parametrized by a function $\z : s \mapsto s + i \vf(s)$, where $\vf \in C^\infty(\R)$ is equal to -1 on $[-1 , 1]$ and equal to $\m$ on $\R \setminus [-2,2]$. For $l \in \N$ we have
\begin{align*}
(it)^l \int_{s=1}^3 e^{-it\z(s)} F\big(\z(s)\big) \z'(s) \, ds
& = \sum _{q = 1}^l (it)^{l-q} e^{-it\z(1)} \restr{\left( \frac 1 {\z'(s)} \frac d{ds} \right)^{q-1} F(\z(s)) }{s = 1} \\
& \quad + \int_{s=1}^3 e^{-it \z(s)} \frac d {ds} \left(\frac 1 {\z'(s)} \frac d {ds}  \right)^{l-1}  F(\z(s)) \, ds.
\end{align*}
Since $\z(1) = 1 - i$, the sum decays exponentially in time. The integral on the right is bounded uniformly in $\m > 0$, so we obtain polynomial decay at any order and uniformly in $\m > 0$ for the integral of the left-hand side. We estimate similarly the contribution of $s \in[-3, -1]$. On the other hand we have
\[
\nr{\int_{\abs s = \e}^1 e^{-it\z(s)} F\big(\z(s)\big) \z'(s) \, ds} = O(e^{-t}),
\]
uniformly in $\m > 0$ (in fact this part does not depend on $\m$) and $\e \in ]0,1]$ (we can use Proposition \ref{prop-chaleur-intro}). It remains to consider the part of $\G_{\m,\g,\e}$ in $\singl{\Re(z) \leq \e}$. By the second statement in Proposition \ref{prop-chaleur-intro}, $F(z)$ is of size $o(\abs z\inv)$ in a neighborhood of 0 in $\C \setminus (-i\R_+^*)$, so the integral over the half circle of radius $\e$ goes to 0 as $\e$ goes to $\infty$. It remains to estimate 
\[
\int_{\s = 0}^1 e^{-t\s} \lim_{\e \to 0} \nr{F(\e -i\s) - F(-\e -i\s)}_{\Lc(L^2(\R^\dd))} \, d\s.
\]
By Proposition \ref{prop-chaleur-intro} we have 
\[
\lim_{\e \to 0} \nr{F(\e -i\s) - F(-\e -i\s)}_{\Lc(L^2(\R^\dd))} \lesssim (a \Ups \s)^{\frac \dd 2 + k - 1 + \abs {\b_x}},
\]
so the conclusion follows after integration.
\end{proof}

\begin{figure}
\includegraphics[width = 0.5 \linewidth]{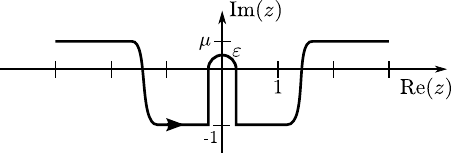}
\caption{Contour of integration for low frequencies.}
\label{fig-contour} 
\end{figure}

\begin{remark} \label{rem-heat-kernel}
When $j=k = 0$ we are dealing with the resolvent of the heat equation, and the proposition gives a decay at rate $O \big(t^{- \frac \dd 2 - \abs{\b_x}}\big)$. We recall that the kernel for the heat equation \eqref{heat} is given by 
\[
K_\heat(t,x) = \left(\frac {a\Ups} {4\pi t}\right)^{\frac \dd 2} e^{-\frac {a \Ups \abs {x}^2}{4t}}.
\]
We can check that even with a compactly supported weight $\h(x)$, the operator 
\[
\h(x) e^{-\frac {t \LD}{a\Ups}} \h(x)
\]
decays as $t^{-\frac \dd 2}$ (up to a multiplicative constant) in $\Lc(L^2(\R^\dd))$. Moreover
\[
\nabla K_\heat(t,x) = \frac {(a\Ups)^{\frac \dd 2 + 1}} {(4\pi t)^{\frac \dd 2}} \frac {x}{2t} e^{-\frac {a \Ups \abs {x}^2}{4t}},
\]
so the size of $\nabla e^{-\frac {t \LD}{a\Ups}}$ decays as $t^{-\frac {\dd + 1}2}$, but $\h(x) \nabla e^{-\frac {t \LD}{a\Ups}} \h(x)$ decays as $t^{-\frac \dd 2 -1}$. Thus, at least for $j=k=0$, the result of Proposition \ref{prop-time-decay-heat} is sharp. This implies in particular that the estimate of Theorem \ref{th-loc-dec} is sharp.
\end{remark}

Now we estimate the contribution of the rest $\Rest(z)$ given in Theorem \ref{th-low-freq-bis}.

\begin{proposition} \label{prop-time-decay-rest}
Let $\Ndev \in \N$, $\b_x \in \N^\dd$ with $\abs {\b_x} \leq 1$ and $\b_t \in \{0,1\}$. Let $\e > 0$ and $\d > \frac {\dd-\e} 2$. Then there exists $C \geq 0$ such that for $\m > 0$ and $t \geq 0$ we have 
\[
\nr{\int_\R e^{-it\t} \th_0(\t) z^{\b_t} \pppg x^{-\d} \partial^{\b_x} \Rest(z) \pppg x^{-\d} \, d\t  }_{\Lc(L^{2}(\O)))} \leq C t^{-\left(\Ndev +1+ \frac \dd 2 + \frac {\abs {\b_x}}2 + \b_t - \e \right) },
\]
where $z$ stands for $\t + i \m$.
\end{proposition}

\begin{proof}
We write $\dd + \abs {\b_x} = 2(\s - \th)$ where $\s \in \N^*$ and $\th \in \big\{ 0, \frac 12 \big\}$.
Let $\n = \Ndev + \b_t + \s$. We denote by $\tilde I_\m(t)$ the integral which appears in the statement of the proposition. After partial integrations as in the proof of Proposition \ref{prop-td-high-freq} we obtain 
\[
e^{t\m} (it)^\n \tilde I_\m(t) = \int_\R e^{-itz} f_{\n}(z)\,d\t
\]
where 
\[
f_\n(z) = \frac {d^\n}{d\t^\n} \left( \th_0(\t) z^{\b_t} \pppg x^{-\d} \partial^{\b_x} \Rest(z) \pppg x^{-\d} \right).
\]
As usual, $z$ stands for $\t + i\m$. By Theorem \ref{th-low-freq-bis} applied with $s =  \frac {\dd-\e}2$ we have 
\[
\nr{f_\m(\t)}_{\Lc(L^2(\R^\dd))} \lesssim \abs \t^{-\th - \frac \e 2} \qandq \nr{f_\n'(\t)}_{\Lc(L^2(\R^\dd))} \lesssim \abs{\t} ^{-1-\th - \frac \e 2}.
\]
By interpolation (see for instance Lemma 6.3 in \cite{khenissir}) we obtain
\[
e^{t\m} t^\n \nr{\tilde I_\m(t)}_{\Lc(L^2(\R^\dd))} \lesssim t^{\th + \e - 1}, 
\]
which concludes the proof.
\end{proof}

Now we can finish the proofs of Theorem \ref{th-loc-dec} and \ref{th-heat}.

\begin{proof}[End of the proof of Theorem \ref{th-heat}]
For the proof of Theorem \ref{th-heat} we estimate $U_{0,\m}(t)$. Since the weight is as strong as we wish, the contribution of high frequencies decays polynomially at any order and can be considered as a rest. We have to estimate $U_{0,\m,0}$. By Proposition \ref{prop-time-decay-rest}, the contribution of $\Rest$ for low frequencies is also a rest. Moreover, for the time derivative, the term $\Id$ which appears in the lower left coefficient of \eqref{eq-res-Ac-tRaz} is holomorphic so its contribution also decays polynomially at any order. It remains the first terms in the developpement given by Theorem \ref{th-low-freq-bis}. By Proposition \ref{prop-time-decay-heat}, these contributions satisfy the properties of the functions $u_{\heat,k}$ as given in Theorem \ref{th-heat}. We only focus on the first term
\[
\tilde u_{\heat,0} = \int_\R e^{itz} \th_0(\t) \big( \LD -ia\Ups z \big)\inv P_\o (i \Th_a u_0 + u_1).
\]
As in the proof of Proposition \ref{prop-td-high-freq}, we can check that 
\begin{equation} \label{tilde-u0-heat}
\tilde u_{\heat,0} = \int_\R e^{itz}  \big( \LD -ia\Ups z \big)\inv P_\o (i \Th_a u_0 + u_1) + O(t^{-\infty}).
\end{equation}
By Lemma \ref{lem-Po} we have $P_\o \Th_a = a \Ups P_{\partial \o}$, so the first term of \eqref{tilde-u0-heat} is the solution of \eqref{heat}, as given in \eqref{def-u-heat}. This concludes the proof of the theorem.
\end{proof}

In Theorem \ref{th-heat} we do not worry about the weight which defines the local energy, and we consider the solution $u$ itself and not only its derivatives. This is not the case in Theorem \ref{th-loc-dec} where we prove an estimate in the energy space and with a sharp weight. In \cite{art-dld-energy-space} we proved a result in the spirit of the Hardy inequality, which we now generalize for our wave guide.

\begin{lemma} \label{lem-hardy-gen}
Let $\d > \frac 12$ and $\s < \d - 1$. Then there exists $C \geq 0$ such that for $u \in C_0^\infty(\bar \O)$ we have 
\[
\nr{\pppg x^\s u}_{L^2(\O)} \leq C \nr{\pppg x^\d \nabla u}_{L^2(\O)}.
\]
\end{lemma}

The interest of this result is that the norm on the right is controlled by the weighted energy. This has a cost in terms of the weight, but we will use this result for the contributions of terms which have a better weight than needed.

\begin{proof}
We first observe that Lemma 4.1 in \cite{art-dld-energy-space} was proved for $\dd \geq 3$ and $\d \geq 0$, but the same result holds with the same proof if $\dd \geq 1$ and $\d > \frac 12$. Now let $u \in C_0^\infty(\bar \O)$. For $y \in \o$ we have
\[
\nr{\pppg x^\s u(\cdot,y)}_{L^2(\R^\dd)}^2 \lesssim \nr{\pppg x^\d \nabla_x u (\cdot,y)}_{L^2(\R^\dd)}^2 \leq \nr{\pppg x^\d \nabla u(\cdot,y)}_{L^2(\R^\dd)}^2.
\]
The result follows after integration over $y \in \o$.
\end{proof}

\begin{proof} [End of the proof of Theorem \ref{th-loc-dec}]
For the proof of Theorem \ref{th-loc-dec} we estimate $U_{1,\m}$. The contribution of high frequencies is given by Proposition \ref{prop-td-high-freq} applied with $\g = \frac \dd 2 + 1$. Let $\d_1 \in \big] \frac \dd 2 , \d-1\big[$. For the contribution of low frequencies, we apply Theorem \ref{th-low-freq-bis} and Propositions \ref{prop-time-decay-heat} and \ref{prop-time-decay-rest} with $\Ndev = 0$ and $\d_1$ instead of $\d$. Since we only estimate the derivatives of the solution, this gives a term whose derivatives with respect to $x$ and $t$ decay as $t^{-\frac \dd 2 - 1}$ and a rest which decays faster. For the derivatives with respect to $y$, we proceed similarly with $\Ndev = 1$. We have $\nabla_y \Pc_{0,0} = 0$ and $\nabla_y \Pc_{1,1} = 0$. The term corresponding to $(k,j) = (1,0)$ decays as $t^{-\frac \dd 2 - 1}$ and the rest decays faster. In the end we have an estimate of the form 
\[
\nr{U_{0,\m,0}}_{\EE^{-\d_1}} \lesssim \pppg t^{-\frac \dd 2 -1} \nr{U_0}_{\HH^{\d_1}}.
\]
We finally use Lemma \ref{lem-hardy-gen} to obtain 
\[
\nr{U_{0,\m,0}}_{\EE^{-\d_1}} \lesssim  \pppg t^{-\frac \dd 2 -1} \nr{U_0}_{\EE^{\d}}.
\]
This concludes the proof.
\end{proof}

The rest of the paper is devoted to the proofs of all the resolvent estimates which have been used in this section.

\section{Separation of the spectrum with respect to the transverse operator} \label{sec-separation}

\newcommand{\Rjslash}{R_n^\sslash(z)}
\newcommand{\Rjslashh}{R_{\h,n}^\sslash(z)}
\newcommand{\Roslash}{R_0^\sslash(z)}
\newcommand{\Roslashh}{R_{\h,0}^\sslash(z)}
\newcommand{\SSGO} {\mathfrak S _\Gc^\O}
\newcommand{\SSGo} {\mathfrak S _\Gc^\o}

In this section we begin our spectral analysis by studying the spectrum and the resolvent estimates for the operator $\Ha$ defined by \eqref{def-Ha}-\eqref{dom-Ha}. In \eqref{eq-Ha-Ta-L} we have written $\Ha$ as the sum of the usual selfadjoint Laplace operator $\LD$ on $\R^d$ and the dissipative operator $\Ta$ on the compact section $\o$. We could use abstract results (see for instance \S XIII.9 in \cite{rs4}) to show that the spectrum of $\Ha$ is 
\begin{equation} \label{spec-Ha}
\s(\Ha) = \s(\LD) + \s(\Ta) = \sum_{k\in \N} \l_k(\a) + \R_+.
\end{equation}
For instance when $\a > 0$ we obtain a sequence of half-lines in the lower half-plane.\\

However this does not give enough information on the resolvent outside the spectrum. Our purpose here is to show that for $\z$ outside $\s(\Ha)$ we can in some sense neglect the contributions of the transverse eigenvalues for which $\Re(\l_k(\a)) \gg \z$ (those for which $d(\z,\l_k(\a) + \R_+) \gg 1$). The idea is to control globally these contributions even if we do not control their number and the lack of self-adjointness. Then it will be possible to write a sum which looks like \eqref{eq-base} but with only a finite number of terms. With such an expression available, it will be easy to deduce precise properties for the resolvent. The problem is that for $\Re(\z) \gg 1$ there will be more and more terms in the sum, so this idea will be mostly used for intermediate and low frequencies. The main result of this section will be Proposition \ref{prop-res-utotal}.\\

Let $R_1,R_2>0$ and 
\begin{equation} \label{def-Gc}
\Gc = \singl{\z \in \C \st \Re(\z) < R_1, \abs{\Im(\z)} < R_2}.
\end{equation}
We assume that $\partial \Gc \cap \s(\Ta) = \emptyset$, which is the case for $R_1$ outside a countable subset of $\R$ and $R_2$ large enough. Let
\begin{equation} \label{def-mathfrak}
\SSGo(\a) = \s(\Ta) \cap \Gc \qqandqq \SSGO(\a) = \SSGo(\a) + \R_+.
\end{equation}
To simplify the notation, we will not always write explicitely the dependance on $\a$ for the quantities which appear in this section.\\

Since $\Ta$ has discrete spectrum it is possible to define a spectral localization on $\Gc$ by means of a Cauchy integral. We define
\begin{equation} \label{def-PGc}
P_\Gc = - \frac 1 {2i\pi} \int_{\partial \Gc}  (\Ta-\s)\inv \, d\s \quad \in \Lc(L^2(\o)),
\end{equation}
and $F_\Gc = \Ran(P_\Gc)$.

\begin{proposition}
The operator $P_\Gc$ is well defined and satisfies the following properties.
\begin{enumerate}[(i)]
\item $P_\Gc$ is a projection on $F_\Gc$.
\item $F_\Gc$ is invariant by $\Ta$.
\item The spectrum of $\restr{\Ta}{F_\Gc}$ is $\SSGo$.
\item $F_\Gc$ is of finite dimension.
\item $P_\Gc$ extends to a bounded operator from $\Huop$ to $\Huo$.
\end{enumerate}
\end{proposition}

\begin{proof}
Let $\g > 0$ be given by Proposition \ref{prop-Ta}. For $R > \g$ we set $\Gc_R = \Gc \cap \singl{\Re(z) \geq -R}$ and define $P_\Gc^R$ as $P_\Gc$ with $\partial \Gc$ replaced by $\partial \Gc_R$. Then we set $F_\Gc^R = \Ran(P_\Gc^R)$. We apply Theorem III.6.17 in \cite{kato}. We obtain properties analogous to (i)-(iii) for $P_\Gc^R$. Moreover, since $\SSGo$ only contains a finite number of eigenvalues of finite multiplicities for $\Ta$, $F_\Gc^R$ is of finite dimension. And finally $P_\Gc^R$ extends to a bounded operator in $\Lc(\HuOp,\HuO)$ by Lemma \ref{lem-res-Ta-H1}.

It only remains to see that since $\SSGo$ is contained in the sector \eqref{sector}, the projection $P_\Gc^R$ does not depend on $R$ and goes to $P_\Gc$ in $\Lc(L^2(\O))$. For this last point, we use the resolvent identity
\[
\big(\tTa - (s + iR_2)\big)\inv - \big(\tTa - (s - iR_2)\big)\inv = 2iR_2 \big(\tTa - (s + iR_2)\big)\inv\big(\tTa - (s - iR_2)\big)\inv.
\]
By Lemma \ref{lem-res-Ta-H1} this is of size $O(s^{-2})$ in $\Lc(L^2(\O))$ when $s \to -\infty$. This concludes the proof.
\end{proof}

Since $F_\Gc$ is of finite dimension, it is quite easy to study the resolvent of $\Ta$ on $F_\Gc$. There exist $\l_1,\dots, \l_N \in \SSGo$ and a basis 
\[
\Bc_\Gc = \big(\f_{j,k} \big) _{\substack{1 \leq j \leq N \\ 0 \leq k \leq \n_j }}
\]
of $F_\Gc$ (with $N \in \N$ and $\n_j \in \N$ for all $j\in\Ii 1 N$) such that the matrix of $\restr{\Ta}{F_\Gc}$ reads $\diag \big( J_{\n_1}(\l_1),\dots, J_{\n_N}(\l_N) \big)$ where for $j \in \Ii 1 N$ the matrix $J_{\n_j}(\l_j)$ is a Jordan bloc of size $(\n_j+1)$ and associated to the eigenvalue $\l_j$.
Thus for $j \in \Ii 1 N$ we have 
\[
\big(\Ta -\l_j \big) \f_{j,0} = 0,
\]
and 
\[
\forall k \in \Ii 1 {\n_j}, \quad \big( \Ta -\l_j \big) \f_{j,k} = \f_{j,k-1}.
\]

Now we extend the operator $P_\Gc \in \Lc(L^2(\o))$ as an operator on $\Lc(L^2(\O))$ as we did for $P_\o$: given $u \in L^2(\O)$, we denote by $P_\Gc u \in L^2(\O)$ the function which satisfies $(P_\Gc u)(x,\cdot) = P_\Gc \big(u(x,\cdot)\big)$ for almost all $x \in \R^\dd$.

\begin{lemma} \label{lem-base-riesz}
Let $u \in L^2(\O)$. Then there exist unique functions $u_{j,k} \in L^2(\R^{\dd})$ for $j \in \Ii 1 N$ and $k \in \Ii 0 {\n_j}$ such that 
\[
P_\Gc  u = \sum_{j=1}^N \sum_{k=0}^{\n_j} u_{j,k} \otimes  \f_{j,k}.
\]
Moreover there exists a constant $C_\Gc$ which does not depend on $u$ such that 
\[
C_\Gc \inv \sum_{j=1}^N \sum_{k=0}^{\n_j} \nr{u_{j,k}}^2_{L^2(\R^{\dd})} \leq \nr{P_\Gc u}^2_{L^2(\O)} \leq C_\Gc \sum_{j=1}^N \sum_{k=0}^{\n_j} \nr{u_{j,k}}^2_{L^2(\R^d)}.
\]
\end{lemma}

This statement can be seen as a partial Riesz basis property. This is in fact trivial since we are on a finite dimensional space. Our main purpose will then be to show that, as long as we are interested in low or intermediate frequencies, it is indeed enough to consider the projection on this finite dimensional space $F_\Gc$.

\begin{proof}[Proof of Lemma \ref{lem-base-riesz}]
For almost all $x \in \R^{\dd}$ we have $u(x,\cdot) \in L^2(\o)$. For such an $x$, $P_\Gc u(x,\cdot)$ belongs to $F_\Gc$ an can be decomposed with respect to the basis $\Bc_\Gc$, which defines almost everywhere on $\R^{\dd}$ the functions $u_{j,k}$ for $j \in \Ii 1 N$ and $k \in \Ii 0 {\n_j}$. Since $F_\Gc$ is of finite dimension, we can find a constant $C_\Gc \geq 1$ which does not depend on $u$ or $x$ and such that
\[
C_\Gc \inv \sum_{j=1}^N \sum_{k=0}^{\n_j} \abs{u_{j,k}(x)}^2 \leq \nr{P_\Gc u(x,\cdot)}^2_{L^2(\o)} \leq C_\Gc \sum_{j=1}^N \sum_{k=0}^{\n_j} \abs{u_{j,k}(x)}^2.
\]
The result follows after integration over $x \in \R^{\dd}$.
\end{proof}

For $\z \in \Gc \setminus \SSGO$ and $u \in L^2(\O)$ we set
\begin{equation} \label{eq-res-Ha-low-freq}
R_\Gc(\z) u = \sum_{j=1}^N \sum_{k=0}^{\n_j} \sum_{l=0}^{k} (-1)^l \big( \LD -\z + \l_j\big)^{-1 -k+l}  u_{j,k} \otimes \f_{j,l}.
\end{equation}

\begin{proposition} \label{prop-res-uGamma}
For $\z \in \Gc \setminus \SSGO$ we have 
\[
(\Ha -\z) R_\Gc(\z) = P_\Gc .
\]
Moreover $R_\Gc(\z)$ extends to an operator in $\Lc(\HuOp,\HuO)$ and if $K$ is a compact subset of $\Gc \setminus \SSGO$ then there exists $C \geq 0$ such that for $\z \in K$ we have 
\[
\nr{R_\Gc(\z)}_{\Lc(\HuOp,\HuO)}  \leq C.
\]
\end{proposition}

\begin{proof}
For $j \in \Ii 1 N$ we can write 
\[
(\Ha -z) = \big( \LD - \z + \l_j\big) + \big( \TaO - \l_j \big).
\]
Then the first statement follows from a straightforward computation. Then we use Lemma \ref{lem-base-riesz}, standard estimates for the self-adjoint operator $\LD$ and the fact that $R_\Gc(\z) = P_\Gc R_\Gc(\z) P_\Gc$ to obtain the required estimate.
\end{proof}

      \detail 
      {
      \begin{remark}
      \begin{equation} \label{eq-propagateur}
      e^{-it\Ha} P_\Gc u = \sum_{j=0}^N \sum_{k=0}^{\n_j} \sum_{l=0}^{k} \frac {t^{k-l}}{(k-l)!} e^{-it(-\D_x + \l_j(\aaa))} u_{j,k} \otimes \f_{j,l}(\aaa).
      \end{equation}
      \end{remark}
      }

The following lemma is quite standard and can be proved by using the spectral measure for the selfadjoint operator $\LD$:

\begin{lemma} \label{lem-contour-laplacien}
Let $\G$ be the boundary of a domain of the form
\[
\tilde \Gc = \singl{z \in \C \st \Re (z) > -R_0, -R_- < \Im(z) < R_+},
\]
with $R_0,R_-,R_+ > 0$. Then for $u \in L^2(\R^\dd)$ we have 
\begin{equation} \label{eq-contour-LD}
-\frac 1 {2i\pi} \int_{\G} (\LD -\z)\inv u \, d\z = u.
\end{equation}
\end{lemma}

In order to convert the properties of the integrals of $\Ta$ and $\LD$ on some suitable contours into properties for the resolvent of the full operator $\Ha$, we will use the following resolvent identity:
\begin{equation} \label{res-identity}
(\Ha-\z)\inv (\TaO-\s)\inv =  (\TaO-\s)\inv (\LD -\z + \s)\inv  - (\Ha-\z)\inv (\LD -\z + \s)\inv.
\end{equation}
This equality relies on the fact that the operators $\Ha$, $\LD$ and $\Ta$ (all seen as operators on $L^2(\O)$) commute. We have already studied the integral over $\s \in \partial \Gc$ of the first and last terms. For the first term of the right-hand side we define for $\z \in \Gc$
\begin{equation} \label{def-BcGc}
\Bc_\Gc (\z) =  \frac 1 {2i\pi} \int_{\partial \Gc}  (\TaO-\s)\inv (\LD -\z + \s)\inv  \, d\s .
\end{equation}

\begin{proposition} \label{prop-BcGc}
The map $\z \mapsto \Bc_\Gc(\z) \in \Lc(\HuOp,\HuO)$ is well defined and holomorphic on $\Gc$. Moreover if $K$ is a compact subset of $\Gc$ then there exists $C_K$ such that for $\z \in K$ we have 
\[
\nr{\Bc_\Gc(\z)}_{\Lc(\HuOp,\HuO)} \leq C_K.
\]
\end{proposition}

\begin{proof}
It is clear that the contribution of the vertical segment in the integral \eqref{def-BcGc} satisfies the conclusion of the proposition. The contribution of the two horizontal half-lines can be written as follows:
\begin{multline*}
\int_{s = -\infty}^{R_1} 2iR_2 \big(\Ta-(s+iR_2) \big)\inv \big(\Ta-(s-iR_2) \big)\inv  \big(\LD-\z+(s+iR_2) \big)\inv \, ds \\
+
\int_{s = -\infty}^{R_1} 2iR_2 \big(\Ta-(s-iR_2) \big)\inv \big(\LD-\z+(s+iR_2) \big)\inv  \big(\LD-\z+(s-iR_2) \big)\inv \, ds.
\end{multline*}
With Lemma \ref{lem-res-Ta-H1} and the standard analogous estimates for $\LD$ we see that these integrals are well defined as operators in $\Lc(\HuOp,\HuO)$ and are uniformly bounded as long as $\z$ stays in a compact subset of $\Gc$.
\end{proof}

With all the results of this section we finally obtain the following proposition.

\begin{proposition} \label{prop-res-utotal}
We have 
\[
\s(\Ha) \cap \Gc = \SSGO \cap \Gc
\]
and for $\z \in \Gc \setminus \s(\Ha)$ we have 
\begin{equation} \label{decomp-res-Ha}
(\Ha-\z)\inv u = R_\Gc(\z) u + \Bc_\Gc(\z) u,
\end{equation}
where $R_\Gc$ and $\Bc_\Gc$ are defined by \eqref{eq-res-Ha-low-freq} and \eqref{def-BcGc}.
\end{proposition}

Thus on $\Gc$ we have written the resolvent of $\Ha$ as the sum of the resolvent on a finite-dimensional subspace (with respect to $y$) and a holomorphic function (both depend on $\Gc$).

On the other hand, we notice that the first statement holds for $\Gc$ as large as we wish, so we have recovered \eqref{spec-Ha}.

\begin{proof}
Let $\z \in  \Gc \setminus \s(\Ha) \subset \Gc \setminus \SSGO$ and $\s \in \partial \Gc$. We have in particular $\s \notin \s(\TaO)$ and $\z - \s \notin \R_+ = \s(\LD)$. By Proposition \ref{prop-res-uGamma}, the resolvent identity \eqref{res-identity} and Lemma \ref{lem-contour-laplacien} we have
\begin{align*}
R_\Gc(\z) u
&  = (\Ha-\z)\inv P_\Gc u  = - \frac 1 {2i\pi} \int_{\partial \Gc} (\Ha-\z)\inv (\TaO-\s)\inv u \, d\s \\
& = \frac 1 {2i\pi} (\Ha-\z)\inv  \int_{\partial \Gc} (\LD -\z + \s)\inv u \, d\s  - \frac 1 {2i\pi} \int_{\partial \Gc}  (\TaO-\s)\inv (\LD -\z + \s)\inv u \, d\s \\
& = (\Ha-\z)\inv u - \Bc_\Gc(\z) u.
\end{align*}
This gives the second statement. Since the right-hand side of \eqref{decomp-res-Ha} is holomorphic on $\Gc \setminus \SSGO$, the left-hand side extends to a holomorphic function on $\Gc \setminus \SSGO$. This implies that $\Gc \setminus \SSGO \subset \Gc \setminus \s(\Ha) $, and concludes the proof.
\end{proof}

The family of operators $\a \mapsto \Ha$ is holomorphic of type B in the sense of Kato \cite{kato}. By continuity of the resolvent $(\Ha - \z)\inv$ with respect to $\a$ we obtain the following conclusion.

\begin{corollary} \label{cor-res-Ha-inter-freq}
Let $K_1$ and $K_2$ be compact subsets of $\C$ such that $K_2 \subset \Gc \setminus \SSGO(\a)$ for all $\a \in K_1$. Then there exists $C \geq 0$ such that for $\a \in K_1$ and $\z \in K_2$ we have 
\[
\nr{(\Ha-\z)\inv}_{\Lc(\HuOp,\HuO)} \leq C.
\]
\end{corollary}

\section{Contribution of intermediate frequencies} \label{sec-inter-freq}

In this section we prove Theorem \ref{th-inter-freq}. This is now a simple consequence of the preliminary work of Sections \ref{sec-general-properties} and \ref{sec-separation}.

\begin{proof} [Proof of Theorem \ref{th-inter-freq}]
Let $\t \in \R \setminus \singl 0$. For $\m \in ]0,1]$, $z = \t +i\m$ and $U = (u,v) \in \HH$ we have by \eqref{eq-res-Ac-tRaz}
\begin{equation*} 
(\Ac-z)\inv U = 
\begin{pmatrix}
\tRaz (i\Th_a + z) u + \tRaz v \\
u + \tRaz (iz\Th_a + z^2) u + z \tRaz v
\end{pmatrix},
\end{equation*}
and hence 
\begin{equation} \label{nr-res-Ac-U}
\begin{aligned}
\nr{(\Ac-z)\inv U}_{\EE}
& \leq \nr{\nabla \tRaz (i\Th_a + z) u}_{L^2(\O)} + \nr{\nabla \tRaz v}_{L^2(\O)}\\
& \quad  + \nr{u + \tRaz (iz\Th_a + z^2) u}_{L^2(\O)} + \nr{z \tRaz v}_{L^2(\O)}.
\end{aligned}
\end{equation}
By Corollary \ref{cor-res-Ha-inter-freq} there exists $C \geq 0$ which depends on $\t$ but not on $\m\in ]0,1]$ or $U \in \HH$ such that 
\[
\nr{\nabla \tRaz v}_{L^2(\O)} + \nr{z \tRaz v}_{L^2(\O)} \leq C \nr{v}_{L^2(\O)}.
\]
For the first term in \eqref{nr-res-Ac-U} we write 
\begin{align*}
\nabla \tRaz (i\Th_a + z) u = \frac 1 z \nabla  u - \frac 1 z \nabla \tRaz \tilde \D u
\end{align*}
(we recall that $\tilde \D$ was defined after \eqref{def-qa}). Then by Corollary \ref{cor-res-Ha-inter-freq}
\[
\nr{\nabla \tRaz (i\Th_a + z) u}_{L^2(\O)} \lesssim \nr{\nabla u} + \nr{\tRaz}_{\Lc(\HuOp,\HuO)} \nr{\nabla u} \lesssim \nr{\nabla u}.
\]
Similarly 
\[
\nr{u + \tRaz (iz\Th_a + z^2) u}_{L^2(\O)} = \nr{\tRaz \tilde \D u}_{L^2(\O)} \lesssim \nr{\nabla u}_{L^2(\O)},
\]
and finally there exists $C \geq 0$ which does not depend on $\m \in ]0,1]$ or $U \in \HH$ and such that
\[
\nr{(\Ac-z)\inv U}_{\EE} \leq C \nr{U}_{\EE}.
\]
Since $\HH$ is dense in $\EE$, this proves that 
\[
\nr{(\Ac-z)\inv}_{\Lc(\EE)} \leq C.
\]
But the size of the resolvent blows up near the spectrum, so $\t$ belongs to the resolvent set of $\Ac$, which means that the resolvent $(\Ac-\t)\inv$ is well defined in $\Lc(\EE)$. It only remains to check as above that this resolvent also defines a bounded operator on $\HH$.
\end{proof}

\begin{remark} \label{rem-res-Ac-U}
The computation of the proof holds for $z$ replaced by $\t$, so for $\t \in \R \setminus \singl 0$ and $U = (u,v) \in \HH$ we have 
\begin{equation} \label{res-Ac-U-bis}
(\Ac-\t)\inv U = 
\begin{pmatrix}
\frac 1 \t u - \frac 1 \t \tRat \tilde \D u + \tRat v \\
- \tRat \tilde \D u + \t \tRat v
\end{pmatrix}.
\end{equation}
\end{remark}

\section{Contribution of low frequencies} \label{sec-low-freq}

We now consider the contribution of low frequencies. For this we have to study the first eigenvalue of the transverse operator.

\begin{proposition} \label{prop-lambda0}
There exist a neighborhood $\Vc$ of 0 in $\C$ and $r > 0$ such that for all $\a \in \Vc$ the set $\Gc$ defined as in \eqref{def-Gc} with $R_1 = R_2 = r$ contains exactly one eigenvalue $\l_0(\a)$ of $\Ta$. Moreover this eigenvalue is algebraically simple, depends holomorphically on $\a \in \Vc$, and we have
\[
{\frac {d\l_0} {d\a}(0)} = -i \Ups.
\]
\end{proposition}

We recall that $\Ups$ was defined in \eqref{def-Ups}.

\begin{proof}
The first eigenvalue of $\To$ is 0 and this eigenvalue is algebraically simple, the eigenvectors being the non-zero constant functions. In particular there exists $r > 0$ such that 0 is the only eigenvalue of $\To$ in $\Gc$ defined as in \eqref{def-Gc} with $R_1 = R_2 = r$.
The family of operators $a \mapsto \Ta$ is a holomorphic family of operators of type B in the sense of \cite[\S VII.4.2]{kato}, so according to the perturabation results in \cite[\S VII.1.3]{kato}, there exist a neighborhood $\Vc$ of 0 and a holomorphic function $\l_0 : \Vc \to \Gc$ such that for all $\a \in \Vc$ the operator $\Ta$ has a unique eigenvalue $\l_0(\a)$ in $\Gc$ and this eigenvalue is simple. 
Moreover the application $\a \mapsto P_\Gc(\a)$ (see \eqref{def-PGc}) is holomorphic and is the projection on the line spanned by the eigenvectors corresponding to this eigenvalue. We denote by $\f_0$ the constant function equal to $\abs \o^{-1/2}$ everywhere on $\o$. Then $T_0 \f_0 = 0$ and $\nr{\f_0}_{L^2(\o)}=1$. Then, choosing $\Vc$ smaller if necessary, $\f_\a : = P_\Gc(\a) \f_0$ is not zero, depends holomorphically on $\a$ and satisfies $T_\a \f_\a = \l_0(\a) \f_\a$ for all $\a \in \Vc$. Thus for all $\a \in \Vc$ we have 
\[
\nr{\nabla_y \f_\a}^2_{L^2(\o)}  - i \a \int_{\partial \o} \abs{\f_\a}^2 = \l_0(\a) \nr{\f_\a}^2_{L^2(\o)}.
\]
We take the derivative of this equality with respect to $\a \in \R$ at point $\a = 0$. Since $\l_0(0) = 0$, $\nr{\f_0} = 1$, $\nabla_y \f_0 = 0$ and $\abs{\f_0}^2 =\abs \o\inv$ everywhere on $\o$ and hence on $\partial \o$, we obtain the expected value for $\l_0'(0)$.
\end{proof}

Let $\Vc$, $r$ and $\Gc$ be given by Proposition \ref{prop-lambda0}. Let $\Uc$ be a neighborhood of 0 such that $a z \in \Vc$ for all $z \in \Uc$. we denote by $P_z$ the projection defined as in \eqref{def-PGc} with $\Ta$ replaced by $\Taz$. We similarly denote by $\Bc(z)$ the operator defined as in \eqref{def-BcGc}. Choosing $\Uc$ smaller if necessary, we can assume that $\abs {\l_0(az)} \leq \frac r 2$ for all $z \in \Uc$. Then $P_z$ can also be written as 
\begin{equation} \label{def-Pz}
P_z = -\frac 1 {2i\pi} \int_{\abs \s = r} (\Taz - \s)\inv \, d\s.
\end{equation}
The application $z \mapsto P_z$ is holomorphic with values in $\Lc(\Huop,\Huo)$. We denote by $P^{(\Nder)}_0 \in \Lc(\Huop,\Huo)$, $\Nder \in \N$, the derivatives of $z \mapsto P_z$ at point 0.\\

By proposition \ref{prop-res-utotal} we have on $L^2(\O)$
\begin{equation} \label{eq-low-freq-Raz}
\Raz  = \big( \LD  + \l_0(az) - z^2 \big)\inv P_{z} + \Bc(z).
\end{equation}
We set
\[
\y(z) = -\frac{\l_0(az) + i a \Ups z - z^2}{z^2}.
\]
By Proposition \ref{prop-lambda0}, $\y$ extends to a holomorphic function on $\Uc$. Using the resolvent identity between $\big( \LD - ia\Ups z - z^2 \y(z)\big)\inv$ and $\big( \LD - ia\Ups z \big)\inv$ we can check by induction on $\Ndev \in \N$ that 
\begin{align} \label{id-res-low-freq}
\big( \LD  + \l_0(az) - z^2 \big)\inv
& = \sum_{k=0}^\Ndev z^{2k} \y(z)^k \big( \LD - ia\Ups z \big)^{-1-k}\\
\nonumber
& \quad  + z^{2(\Ndev+1)} \y(z)^{\Ndev+1} \big( \LD - ia\Ups z \big)^{-1-\Ndev} \big( \LD - ia\Ups z - z^2 \y(z)\big)\inv.
\end{align}
For $k \in \Ii 0 \Ndev$ we can write $\y(z)^{k} = \sum_{l=0}^{\Ndev -k} \y_{k,l}z^l + z^{\Ndev-k+1} \tilde \y_k(z)$ where $\y_{k,0},\dots,\y_{k,\Ndev-k}$ are complex numbers and $\tilde \y_k$ is holomorphic. We also have $P_z = \sum_{l=0}^\Ndev P_0^{(l)} z^l / l! + z^{\Ndev+1} \tilde P_\Ndev(z)$ where $\tilde P_\Ndev : \Uc \to \Lc(\Huop,\Huo)$ is holomorphic. Thus we obtain \eqref{dev-res-low-freq} where $\Rest(z)$ is the sum of the holomorphic function $\Bc(z)$ and a linear combination of terms of the form 
\[
z^l \big(\LD - ia\Upsilon z\big)^{-k_1} \big(\LD - ia\Upsilon z - z^2 \y(z) \big)^{-k_2} \tilde P(z),
\]
where $\tilde P : \Uc \to \Lc(\Huop,\Huo)$ is holomorphic and $l,k_1,k_2 \in \N$ are such that $k_1 + k_2 \geq 1$ and $l-k_1-k_2 \geq \Ndev$. Moreover for all $k \in \Ii 0 \Ndev$ we have 
\[
\Pc_{k,k} = \y_{k,0} P_0 = \y(0)^k P_0,
\]
so the statement about $\Pc_{k,k}$ in Theorem \ref{th-low-freq-bis} holds with $\s = \y(0)$.\\

The estimate of $\Rest(z)$ in Theorem \ref{th-low-freq-bis} is a consequence of the following proposition.

\begin{proposition} \label{prop-dilatation}
Let $k_1,k_2 \in \N$ with $k_1 + k_2 \geq 1$, $s \in \big[ 0 , \frac {\dd}2 \big[$, $\d > s$ and $\b_x \in \N^\dd$ be such that $\abs{\b_l} \leq 1$. For $z \in \C_+$ we set
\[
\Tc(z) = \pppg x^{-\d} \partial^{\b_x} \big(\LD - ia\Upsilon z\big)^{-k_1} \big(\LD - ia\Upsilon z - z^2 \y(z) \big)^{-k_2}  \pppg x^{-\d} \quad \in \Lc(\R^\dd).
\]
Then for $\Nder \in \N$ there exists $C \geq 0$ such that for $z \in \C_+ \cap \Uc$ we have
\[
\nr{\Tc^{(\Nder)}(z)}_{\Lc(\R^\dd)} \leq C \left( 1 + \abs z^{-k_1 -k_2 -\Nder + s + \frac {\abs{\b_x}} 2} \right).
\]
\end{proposition}

\begin{proof}
The derivative $\Tc^{(\Nder)}(z)$ can be written as a sum of terms of the form 
\begin{equation} \label{term-der-Tc}
h(z) \pppg x^{-\d} \partial^{\b_x} \big(\LD - ia\Upsilon z\big)^{-k_1-\Nder_1} \big(\LD - ia\Upsilon z - z^2 \y(z) \big)^{-k_2-\Nder_2}  \pppg x^{-\d}
\end{equation}
where $\Nder_1,\Nder_2 \in \N$ are such that $\Nder_1 + \Nder_2 \leq \Nder$ and $h$ is a holomophic function.
We use the same scaling argument as in \cite{boucletr14, art-dld-energy-space} (in a much simpler version). For $z \in \C_+$ and a function $u$ on $\R^{\dd}$ we define $\Phi_z u$ by
\[
(\Phi_z u) (x) = \abs{z}^{\frac {\dd} 4} u\big( \abs z^{\frac 12} x \big).
\]
The dilation $\Phi_z$ is unitary as an operator on $L^2(\R^\dd)$, but for $p \in [1,+\infty]$ we have on $L^p(\R^\dd)$
\begin{equation} \label{dil-Lp}
\nr{\Phi_z}_{\Lc(L^p(\R^\dd))} = \abs{z}^{\frac \dd 4 - \frac \dd {2p} }.
\end{equation}
Let 
\[
\n = k_1 + \Nder_1 + k_2 + \Nder_2 - \frac {\abs{\b_x}}2 \quad \text{and} \quad \s = \min(s,\n).
\]
We have 
\begin{align*}
 \big(\LD - ia \Ups z \big)\inv  = \abs z \inv \Phi_z\big(\LD - ia \Ups \hat z  \big)\inv \Phi_z\inv
\end{align*}
(where $\hat z$ stands for $z / \abs z$) and 
\begin{align*}
\big(\LD - ia \Ups z - z^2 \y(z) \big)\inv= \abs z \inv  \Phi_z \big(\LD - ia \Ups \hat z - z\hat z \y (z)  \big)\inv \Phi_z\inv .
\end{align*}
For any $\th \in\R$ the two resolvents on the right are in $\Lc(H^{\th-1},H^{\th+1})$ uniformly for $z \in \C_+ \cap \Uc$ (we can choose $\Uc$ smaller if necessary). On the other hand we have
\[
\partial^{\b_x} = \abs{z}^{\frac {\abs {\b_x}} 2} \Phi_z \partial^{\b_x} \Phi_z\inv ,
\]
so \eqref{term-der-Tc} is equal to 
\[
z^{-\n} h(z) \pppg x^{-\d}\Phi_z \partial^{\b_x} \big(\LD - ia\Upsilon \hat z\big)^{-k_1- \Nder_1} \big(\LD - ia\Upsilon \hat z - z\hat z \y(z) \big)^{-k_2 - \Nder_2} \Phi_z \inv \pppg x^{-\d}.
\]
We have the Sobolev embeddings $L^{p_r} \subset H^{-\s}$ and $H^{\s} \subset L^{p_l}$ where $p_l = \frac {2\dd}{\dd-2\s}$ and $p_r = \frac {2\dd}{\dd + 2\s}$. Moreover $\pppg x^{-\d} \in \Lc(L^{p_l},L^2) \cap \Lc(L^2, L^{p_r})$, so with \eqref{dil-Lp} we get 
\[
\nr{\Tc(z)}_{\Lc(L^2(\R^\dd))} \lesssim \abs{z}^{\s-\n}.
\]
It only remains to recall that $\s$ is equal to $\n$ or $s$ to conclude.
\end{proof}

Now we estimate the terms which only contain powers of the heat resolvent. We first remark that the second statement of Proposition \ref{prop-chaleur-intro} is a consequence of Proposition \ref{prop-dilatation}. For the first estimate we use the the explicit kernel of the heat equation.

\begin{proof} [Proof of Proposition \ref{prop-chaleur-intro}.\eqref{estim-chaleur-diff}]
For $\ell \in \N^*$ and $\z \in \C_+$ we denote by $K_\ell(\z)$ the kernel of ${(\LD - \z^2)^{-\ell}}$:
\[
K_\ell(\z;x) = \frac 1 {(2\pi)^\dd} \int_{\R^\dd} \frac {e^{i \innp x \x}}{(\abs \x^2 - \z^2)^\ell}  d\x.
\]
Let $\k_0 > 0$, $\k \in ]0,\k_0]$ and $x \in \R^\dd$. By \cite[\S 1.5]{melrose} we have for $r > 0$ small enough 
\begin{align*}
\tilde K_\ell(\k;x)
& := \lim_{\z \to \k} K_\ell(\z;x) - \lim_{\z \to -\k} K_\ell(\z;x) \\
& = \frac 1 {(2\pi)^\dd} \int_{\th \in S^{\dd -1}} \int_{\abs{\s - \k} = r } e^{i \s  \innp x \th}  \frac {\s^{\dd-1}} {(\s^2 - \k^2)^\ell} \, d\s \, d\th,
\end{align*}
where $S^{\dd-1}$ is the unit sphere in $\R^\dd$. For $\s$ in a neighborhood of $\k$, $\th \in S^{\dd-1}$ and $x \in \R^\dd$ we set $f(\s) = \frac {\s^{\dd-1}} {(\s + \k)^\ell}$ and $F(\s,\th,x) =  e^{i \s  \innp x \th} f(\s)$. Then by the residue theorem we obtain
\[
\tilde K_\ell(\k;x) = \frac {2i\pi} {(2\pi)^\dd} \int_{\th \in S^{\dd -1}} \frac {\partial_\s^{\ell-1} F(\k,\th,x)} {(\ell-1)!}  \, d\th.
\]
We have
\[
\abs{\partial_\s^{\ell-1} F(\k,\th,x)} \lesssim \abs x^{\ell-1} \k^{d-2\ell},
\]
and hence for $\d > \frac \dd 2 + \ell -1$
\begin{equation} \label{estim-noyau-1}
\nr{\pppg x^{-\d} \left(\big(\L - (\k^2+i0)\big)^{-\ell} - \big(\LD - (\k^2-i0)\big)^{-\ell}  \right) \pppg x^{-\d}}_{\Lc(L^2(\R^d))} \lesssim \k^{{\dd - 2\ell}}.
\end{equation}
Now let $j \in \Ii 1 \dd$. We can check that the derivative $\partial _{x_j} \partial_\s^{\ell-1} F(\k,\th,x)$ is a linear combination of terms of the form
\begin{equation*}
T_{j,\n}(\k,\th,x) := \th_j \innp x \th ^{\n-1} e^{i \k \innp x \th} f^{(\ell-1-\n)}(\k), \quad \text{for } \n \in \Ii 1 {\ell-1},
\end{equation*}
or
\begin{equation*}
\tilde T_{j,\tilde \n}(\k,\th,x) := \k \th_j \innp x \th ^{\tilde \n} e^{i \k \innp x \th} f^{(\ell-1-\tilde \n)}(\k), \quad \text{for } \tilde \n \in \Ii 0 {\ell-1}.
\end{equation*}
It is not difficult to see that for $\n \in \Ii 2 {\ell-1}$ and $\tilde \n \in \Ii 1 {\ell-1}$ we have
\[
\abs{T_{j,\n}(\k,\th,x)} + \abs{\tilde T_{j,\tilde \n}(\k,\th,x)}  \lesssim \abs x^{\ell-1} \k^{d-2\ell + 2}.
\]
For $\th \in S^{\dd-1}$ we set $\hat \th_j = (\th_1,\dots, \th_{j-1}, - \th_j, \th_{j+1}, \dots, \th_\dd) \in S^{\dd-1}$. We have
\[
\abs{\int_{\th \in S^{\dd-1}} T_{j,1}(\k,\th,x)\, d\th} \leq \frac 12 \int_{\th \in S^{\dd-1}}\abs{ T_{j,1}(\k,\th,x) + T_{j,1}(\k,\hat \th_j,x) }  \, d\th \lesssim \abs x \k^{d-2\ell+2}.
\]
We have a similar estimate for $\tilde T _{j,0}$, so finally
\begin{equation} \label{estim-noyau-2}
\nr{\pppg x^{-\d} \partial_{x_j} \Big(\big(\L - (\k^2+i0)\big)^{-\ell} - \big(\LD - (\k^2-i0)\big)^{-\ell}  \Big) \pppg x^{-\d}} \lesssim \k^{d-2\ell+2}.
\end{equation}
It only remains to apply \eqref{estim-noyau-1} and \eqref{estim-noyau-2} with $\ell = j + 1$ and $\k = \sqrt s$ to conclude the proof.
\end{proof}

We finish this section by checking that there is no problem with low frequency if we localize away from low frequencies with respect to the first $\dd$ variables. More precisely we prove Proposition \ref{prop-high-freq-low-freq}, which was used for the proof of Theorem \ref{th-high-freq-loc-decay}.

\begin{proof} [Proof of Proposition \ref{prop-high-freq-low-freq}]
Let $r > 0$ be such that $\h_1 = 1$ on $[0,r]$.
For $v \in L^2$ the result of Lemma \ref{lem-contour-laplacien} holds with $u = (1-\h_1)(\LD) v$ and $\tilde \Gc$ of the form 
\[
\tilde \Gc = \singl{z \in \C \st \Re(z) > r , \abs{\Im(z)} < r}.
\]
Thus we can apply Proposition \ref{prop-res-utotal} with a domain $\Gc$ of the form 
\[
\Gc = \singl{z \in \C \st \Re(z) < - r , \abs{\Im(z)} < r}.
\]
But $\SSGo (az) = \emptyset$ for $z \in \C$ small enough, so $z \mapsto \tRaz (1-\h_1)(\LD)$ is holomorphic on a neighborhood of 0. With Proposition \ref{prop-Ac-diss} this proves that $z \mapsto (\Ac-z)\inv \in \Lc(\HH,\EE) \subset \Lc(\HH,\HH)$ extends to a holomorphic function on a neighborhood of 0 (notice that $\h_1(\LD)$ commutes with $\Raz$ and $\Th_a$).

Let $\tilde \h_1 \in C_0^\infty(\R,[0,1])$ be equal to 1 on a neighborhood of 0 and such that $\h_1 = 1$ on a neighborhood of $\supp (\tilde \h_1)$. Then we define $\tilde \Xc_1$ as we did for $\Xc_1$ in \eqref{def-Xc}. Since $\tilde \Xc_1$ commutes with $\Ac$ we have for all $z \in \C_+$
\[
(1-\Xc_1) (\Ac-z)\inv = (1-\Xc_1) (\Ac-z)\inv \big(1-\tilde \Xc_1 \big).
\]
Since $ \big(1-\tilde \Xc_1 \big)$ belongs to $\Lc(\EE,\HH)$, this concludes the proof.
\end{proof}

\section{Contribution of high frequencies} \label{sec-high-freq}

\newcommand{\tum}{\widetilde {u_m}}
\newcommand{\tfh}{\widetilde {f_m}}

In this section we prove the high frequency resolvent estimates of Theorem \ref{th-high-freq}. By \eqref{spec-Ha}, if $\t^2$ is close to the spectrum of $\Hat$ there exists $\l \in \s(\Tat)$ and $r \geq 0$ such that $\t^2$ is close to $\l + r$. We deal separately with the contributions of the different pairs $(\l,r)$. Those for which $r$ is small compared to $\t^2$, and those for which $r$ is large itself.

\subsection{Contribution of large transverse eigenvalues}

If $\t^2$ is large and $r$ is small, then $\l$ has to be large. The good properties for the resolvent in this case come from the fact that the eigenvalues of $\Tat$ close to $\t^2$ are far from the real axis and, even if $\Tat$ is not self-adjoint, we have the expected corresponding estimate for the resolvent. The following result is a direct consequence of Theorem \ref{th-gap-Tah}:

\begin{proposition} \label{prop-gap-Taz}
There exist $\t_0 \geq 1$, $\g > 0$ and $c \geq 0$ such that for $\t \geq \t_0$ and $\z \in \C$ which satisfy
\[
\abs{\Re(\z-\t^2)} \leq \g \t^2 \qandq \Im(\z) \geq -  \g \t
\]
the resolvent $(\Tat-\z)\inv$ is well defined and we have
\[
\nr{(\Tat -\z)\inv}_{\Lc(L^2(\o))} \leq \frac c \t.
\]

\end{proposition}

As already explained, we cannot use the results of Section \ref{sec-separation} to obtain uniform estimates for high frequencies. However we use the same kind of idea in the proof of the following proposition.

\begin{proposition} \label{prop-small-long-freq}
Let $\t_0$ and $\g$ be given by Proposition \ref{prop-gap-Taz}. If $\h_1$ is supported in $]-\g,\g[$ then there exists $c \geq 0$ such that for $\t \geq \t_0$ we have 
\[
\nr{\h_{\t} (\LD) \Rat}_{\Lc(L^2(\O))} \leq \frac c {\t}.
\]
\end{proposition}

We recall that $\h_\t$ was defined by $\h_1(\cdot / \t^2)$.

\begin{proof}
For $\t \geq \t_0$ we set
\[
\Gc_\t = \singl{\z \in \C \st \abs{\Re(\z) - \t^2} \leq \g \t^2 , \abs{\Im(\z)} \leq \g \t}.
\]
The proof is based on the resolvent identity \eqref{res-identity} applied with $\a = a \t$ and $\z = \t^2$, and integrated over $\s \in \partial \Gc_\t$. According to Proposition \ref{prop-gap-Taz} we have $\Gc_\t \cap \s(\Tat) = \emptyset$ so
\[
\h_{\t}(\LD) \Rat   \int_{\partial \Gc_\t} (\Tat - \s)\inv \, d\s = 0.
\]
In the spirit of Lemma \ref{lem-contour-laplacien} we can check that
\[
- \Rat  \frac 1 {2i\pi} \int_{\partial \Gc_\t} \h_{\t}(\LD) \big( \LD - \t^2 + \s \big) \inv  \, d\s = \Rat  \h_{\t}(\LD)  .
\]
Now let 
\[
\Rc(\t) = \int_{\partial \Gc_\t} (\Tat - \s)\inv \h_{\t} (\LD) \big( \LD - \t^2 + \s \big) \inv \,d\s.
\]
Let $E$ be the spectral measure associated to $\LD$. We have 
\begin{align*}
\Rc(\t)
& = \int_{\partial \Gc_\t} (\Tat - \s)\inv \left(\int_0^{+\infty} \frac {\h_{\t} (\Xi)}{\Xi - \t^2 + \s} \, dE(\Xi) \right) \,d\s\\
& = \int_0^{\g \t^2} \h_{\t} (\Xi)\left(\int_{\partial \Gc_\t}   \frac {(\Tat - \s)\inv}{\Xi - \t^2 + \s} \, d\s \right)  \, dE(\Xi).
\end{align*}
For $\Xi \in [0,\g \t^2]$ we set 
\[
\Gc_{\t,\Xi} = \singl{\z \in \Gc_\t \st \abs{\Re(\z) - \t^2 + \Xi} \leq \g \t}.
\]
Since the function $\s \mapsto \frac {(\Tat - \s)\inv}{\Xi - \t^2 + \s}$ is holomophic on $\Gc_{\t} \setminus \Gc_{\t,\Xi}$ and $\partial \Gc_{\t,\Xi}$ is of length $8 \g \t$ we have by Proposition \ref{prop-gap-Taz} 
\[
\nr{\int_{\partial \Gc_\t}   \frac {(\Tat - \s)\inv}{\Xi - \t^2 + \s} \, d\s}_{\Lc(L^2(\O))} = \nr{\int_{\partial \Gc_{\t,\Xi}}   \frac {(\Tat - \s)\inv}{\Xi - \t^2 + \s}\,d\s} _{\Lc(L^2(\O))} \lesssim \frac 1 {\t}.
\]
Therefore
\[
\nr{\Rc(\t)}_{\Lc(L^2(\O))} \lesssim \frac 1 {\t},
\]
and we conclude with \eqref{res-identity}.
\end{proof}

\subsection{Contribution of high longitudinal frequencies} \label{sec-high-freq-longitudinal}

If the section $\o$ is of dimension 1, we can prove that the first eigenvalues of $\Tat$ go back to the real axis when the absorption coefficient $a\t$ goes to infinity (see Appendix \ref{sec-sec-dim1}). In other words
\[
\sup _{\substack{\l \in \s(\Tat)\\ \Re(\l) \leq \t^2}} \Im(\l) \limt \t {+\infty} 0,
\]
and hence
\[
d \big(\t^2, \s(\Hat) \big) \limt \t{+\infty} 0.
\]
Thus we cannot expect a uniform bound for $R_a(\t)$ on $\Lc(L^2(\O))$ when $\t \gg 1$. This is only proved when $\dim(\o) = 1$ but we expect that the same phenomenon occurs when $\dim(\o) \geq 2$.

However, if $\l \in \s(\Tat)$ is such that $\Re(\l) < \t^2$ and $\abs{\Im (\l)} \ll 1$ then according to Proposition \ref{prop-gap-Taz} we have ${\t^2 - \Re(\l)} \gg 1$. By usual semiclassical technics we can prove estimates for the resolvent $(\LD - (\t^2 - \l))\inv$ in this case. We use the same kind of ideas for the following result.

\begin{proposition} \label{prop-high-long-freq}
Let $\t_0$ be given by Proposition \ref{prop-gap-Taz}. Let $\d > \frac 12$. Then there exists $c \geq 0$ such that for $\t \geq \t_0$ we have 
\[
\nr{\pppg x^{-\d} (1-\h_\t) (\LD) R_a(\t) \pppg x^{-\d}}_{\Lc(L^2(\O))} \leq \frac c \t.
\]
\end{proposition}

For the proof of this and the following propositions it is convenient to rewrite the problem in the semiclassical setting. We have defined $\Tah$ in \eqref{dom-Tah}. For $h \in ]0,1]$ we set $\LDh = h^2 \LD$, $\Hah = h^2 H_{a/h}$ and $\Rh = \big(\Hah -1 \big)\inv \in \Lc(L^2(\O))$. We also denote by $\Rh$ the operator ${(-h^2 \D - ih\Th_a -1)\inv} \in \Lc(\HuOp,\HuO)$. Then for $\t \geq 1$ and $h = \t \inv$ we have 
\begin{equation} \label{eq-res-semiclass}
\tRat = h^2 \Rh.
\end{equation}
For a suitable symbol $q$ on $\R^{2\dd}$, $h \in ]0,1]$ and $u \in L^2(\O)$ we define
\[
\Opwx(q) u (x,y) = \frac 1 {(2\pi h)^\dd} \int_{\tilde x \in \R^{\dd}} \int_{\x \in \R^{\dd}} e^{\frac ih \innp{x-\tilde x}{\x}} q\left( \frac {x + \tilde x}2 , \x \right) u(\tilde x,y) \, d\x \, d\tilde x.
\]
This is a pseudo-differential operator only in the $x$-directions, so there is no difficulty with the fact that $\O$ is bounded in the $y$-directions.

\begin{lemma} \label{lem-traces}
For $h \in ]0,1]$, $f \in C_0^\infty(\O)$ and $u = \Rh f$ we have
\[
a \int_{\partial \OO} \abs{u}^2  + \frac 1 a \int_{\partial \OO} \abs{h \partial_\n u}^2 \leq \frac 2 {h} \nr u_{L^{2,-\d}(\O)} \nr f _{L^{2,\d}(\O)}.
\]
\end{lemma}

\begin{proof}
We have
\[
\abs{\Im \int_{\partial \OO}  h\partial_\nu  {u} \bar{u} \, dx} = -\frac 1 {h} \abs{\Im \innp{ (\Hah-1) u}{u}_{L^2(\O)}} \leq \frac 1 h \nr u_{L^{2,-\d}(\O)} \nr f _{L^{2,\d}(\O)}.
\]
Since $h \partial_\n u = i a u$ on $\partial \OO$ we have on the other hand
\[
 \int_{\partial \OO}  h\partial_\nu  {u} \bar{u}  =  i a \int_{\partial \OO}  \abs {u}^2  = \frac i {a} \int_{\partial \OO}  \abs {h \partial _\n u}^2.
\]
The conclusion follows.
\end{proof}

For the proof of Proposition \ref{prop-high-long-freq} we use an escape function as in \cite{jecko04,art-nondiss}:

\begin{proof} [Proof of Proposition \ref{prop-high-long-freq}]
For $(x,\x) \in \R^{2\dd}$ we set 
\[
g(x,\x) = (1-\h_1)^2 \big( \abs \x^2 \big) \int_0^{+\infty}  \pppg {x-2 \th \x}^{-2\d} \, d\th 
\]
(we recall that $(1-\h_1)$ vanishes on a neighborhood of 0). The symbol $g$ and all its derivatives are bounded on $\R^{2\dd}$. Moreover for $(x,\x) \in \R^{2\dd}$ we have 
\[
\big\{ \abs \x^2, g \big\} (x,\x) = \restr{\frac d {ds} g(x+2s\x,\x)}{s = 0} = (1-\h_1)^2 (\abs \x^2) \pppg x^{-2\d},
\]
where $\{p,q\}$ is the Poisson bracket $\nabla _\x p \cdot \nabla_x q - \nabla _x p \cdot \nabla_\x q$. Let $f \in C_0^\infty(\O)$ and $u_h = \Rh f$. We recall that $[\LDh,\Opwx(g)] = -2ih \Opwx(\x \cdot \partial_xg)$ (there is no rest) so
\begin{eqnarray*}
\lefteqn{\innp{\Opwx\big( \big\{ \abs \x^2 , g \big\} \big) u_h}{u_h}_{\O} = \frac ih \innp{[\LDh -1 , \Opwx(g)] u_h}{u_h}_{\O} }\\
&& = - \frac 2 {h}  \Im \innp{\Opwx(g) u_h}{(\LDh - 1)  u_h}_{\O}  \\
&& =  \frac 2 h \Im \innp{\Opwx(g) u_h} {- h^2 \D_y  u_h}_{\O}  + O \big( h\inv \nr f_{L^{2,\d}(\O)} \nr {u_h}_{L^{2,-\d}(\O)} \big)
\end{eqnarray*}
(we have used the fact that $\Opwx(g)$ defines a bounded operator on $L^{2,-\d}(\O)$). But 
\begin{align*}
\innp{\Opwx(g) u_h} {- h^2 \D_y  u_h}_{\O}
& = - h^2 \int_{\partial \O} \Opwx (g) u_h \partial_\nu \bar{u_h} +   h^2 \int_{\partial \O} \Opwx (g) \partial_\nu u_h  \bar{u_h} \\
& \quad + \innp{- h^2 \D_y u_h}{\Opwx(g) u_h}_{\O},
\end{align*}
so according to Lemma \ref{lem-traces}
\[
 \frac 2 h \Im \innp{\Opwx(g) u_h} {- h^2 \D_y  u_h}_{\O} =O \big( h\inv \nr f_{L^{2,\d}(\O)} \nr {u_h}_{L^{2,-\d}(\O)} \big).
\]
By Proposition \ref{prop-small-long-freq} we have 
\begin{align*}
\nr {u_h}_{L^{2,-\d}}^2
\lesssim \nr {(1-\h_1) (\LDh) u_h}_{L^{2,-\d}}^2 + \frac{\nr{f}_{L^2}^2} {h^2}
\lesssim \frac {\nr{f}_{L^{2,\d}} \nr {u_h} _{L^{2,-\d}}} h + h \nr{u_h}_{L^{2,-\d}}^2 + \frac{\nr{f}_{L^2}^2} {h^2},
\end{align*}
and the conclusion follows.
\end{proof}

With Propositions \ref{prop-small-long-freq} and \ref{prop-high-long-freq} we obtain the following result:

\begin{proposition} \label{prop-high-freq}
Let $\t_0$ be given by Proposition \ref{prop-gap-Taz}. Let $\d > \frac 12$. Then there exists $c \geq 0$ such that for $\t \geq \t_0$ we have 
\[
\nr{\pppg x^{-\d}  R_a(\t) \pppg x^{-\d}}_{\Lc(L^2(\O))} \leq \frac c \t.
\]
\end{proposition}

\subsection{Estimates for the derivatives of the resolvent}

We have proved uniform estimates for the resolvent $\Rat$ on $L^2(\O)$. Now we have to deduce estimates for its derivatives. In order to prove high frequency estimates for the powers of the resolvent of a Schr\"odinger operator, we can use estimates in the incoming and outgoing region (see \cite{Isozaki-Ki-85-,jensen85}). Here we have to check that this strategy works on our wave guide if we consider incoming and outgoing region with respect to the first $\dd$ variables. More important, we will have to take into account the inserted factors $\Th_a$. We will see that if we insert an obstract operator $\Th \in \Lc(\HuO,\HuOp)$ (or even in $\Lc(H^s(\O), H^s(\O)')$ for some $s \in \big] \frac 12, 1\big[$), we obtain estimates which are not good enough to conclude. In order to prove sharp estimates, we will use the fact that the inserted operator $\Th_a$ is exactly (up to the factor $\t$) the dissipative part in the resolvent $\Rat$.\\

For $R \geq 0$, $d\geq 0$ and $\s \in ]-1,1[$ we denote by
\[
\zoneS_\pm(R,\n,\s)  = \singl{ (x,\x)\in \R^d \times \R^d \st \abs x \geq R,   \abs \x  \geq \n  \text{ and } \pm \innp x \x \geq \pm \s \abs x \abs \x} 
\]
the incoming and outgoing regions in $\R^{2d} \simeq T^*\R^d$. Then we denote by $\symb_\pm(R,\n,\s)$ the set of symbols $b \in C^\infty(\R^{2d})$ which are supported in $\zoneS_\pm(R,\n,\s)$ and such that 
\[
\abs{\partial_x^{\b_x} \partial_\x^{\b_\x} b(x,\x)} \lesssim \pppg x^{-\abs {\b_x}}.
\]

\begin{definition} \label{def-Rcha}
Let $p \in \N^*$, $k_1,\dots,k_p \in \N^*$ and $k = k_1+ \dots + k_p$. For $h \in ]0,1]$ we set 
\begin{equation} \label{def-Th-zeta}
\Psi_{h} = \Rh^{k_1}  \Th_a \Rh^{k_2} \Th_a \dots \Th_a \Rh^{k_p}.
\end{equation}
We say that the family $(\Rcha)_{h \in ]0,1]}$ of operators in $\Lc(L^2(\O))$ belongs to $\mathfrak R^{k_1,\dots,k_p}$ if it satisfies one of the following properties.
\begin{enumerate}[(i)]
\begin{subequations} \label{f-Rcha}

\item There exists $\tilde \h_1 \in C_0^\infty(\R,[0,1])$ supported in $]-\g,\g[$ ($\g$ being given by Proposition \ref{prop-gap-Taz}) such that
\begin{equation} \label{f-Rcha-0}
\Rcha = \tilde \h_1(\LDh) \Psi_{h}.
\end{equation}

\item There exists $\d > k - \frac 12$ such that 
\begin{equation} \label{f-Rcha-1}
\Rcha = \pppg x^{-\d}  \Psi_{h}  \pppg x^{-\d}.
\end{equation}

\item There exist $\d > k - \frac 12$, $\rho > 0$, $R > 0$, $\n > 0$, $\s_-\in]-1,1[$ and $b_- \in \symb_-(R,\n,\s_-)$ such that
\begin{equation} \label{f-Rcha-2}
\Rcha = \pppg x^{\d-k- \rho} \Opwx(b_-) \Psi_{h}  \pppg x^{-\d}.
\end{equation}

\item There exist $\d > k - \frac 12$, $\rho > 0$, $R > 0$, $\n > 0$, $\s_+\in]-1,1[$ and $b_+ \in \symb_+(R,\n,\s_+)$ such that
\begin{equation} \label{f-Rcha-3}
\Rcha  = \pppg x^{-\d}  \Psi_{h}   \Opwx(b_+)\pppg x^{\d-k-\rho}.
\end{equation}

\item There exist $\d_-,\d_+ \in \R$, $R > 0$, $\n > 0$, $\s_\pm\in]-1,1[$ and $b_\pm \in \symb_\pm(R,\n,\s_\pm)$ such that $\s_-<\s_+$ and
\begin{equation} \label{f-Rcha-4}
\Rcha =  \pppg x^{\d_-} \Opwx(b_-) \Psi_{h}  \Opwx(b_+) \pppg x^{\d_+}.
\end{equation}

\end{subequations}
\end{enumerate}
\end{definition}

\begin{proposition} \label{prop-incoming-outgoing}
Let $(\Rcha) \in \mathfrak R^1$. Then there exist $h_0 > 0$ and $c \geq 0$ such that for $h \in ]0,h_0]$ and $\b_1,\b_2 \in \{0,1\}$ we have
\[
\nr{\Rcha}_{\Lc(H^{\b_1}(\O)',H^{\b_2}(\O))} \leq \frac c {h^{1+\b_1+\b_2}}.
\]
\end{proposition}

\begin{proof}
\stepp We begin with the estimates in $\Lc(L^2(\O))$. If $(\Rcha)$ is of the form \eqref{f-Rcha-0} or \eqref{f-Rcha-1}, then this is just Proposition \ref{prop-small-long-freq} or \ref{prop-high-freq} rewritten with semiclassical notation. We consider the case \eqref{f-Rcha-2}. Let $\z \in \C_+$. The operator $\Tah$ commutes with $\LDh$ and any pseudo-differential operator with respect to the $x$ variable so we can write 
\begin{eqnarray} \label{estim-entrante}
\lefteqn{\nr{\pppg x^{\d-1-\rho} \Opwx(b_-) (\Hah - \z)\inv \pppg x^{-\d}}_{\Lc(L^2(\O))}}\\
\nonumber
&& \leq \nr{ \frac i h \int_{0}^{+\infty} \pppg x^{\d-1-\rho} \Opwx(b_-) e^{-\frac {it} h (\Hah - \z)} \pppg x^{-\d}\, dt}_{\Lc(L^2(\O))}\\
\nonumber
&& \leq \frac 1 h \int_{0}^{+\infty} \nr{\pppg x^{\d-1-\rho} \Opwx(b_-) e^{- \frac {it} h \LDh} \pppg x^{-\d}}_{\Lc(L^2(\R^\dd))} dt.
\end{eqnarray}
By Proposition 3.2 in \cite{wang88} we have 
\[
\nr{\pppg x^{\d-1-\rho} \Opwx(b_-) e^{- \frac {it} h \LDh} \pppg x^{-\d}}_{\Lc(L^2(\R^\dd))} \lesssim \pppg t^{-1-\rho}.
\]
It only remains to take the limit $\z \to 1$ to conclude after integration over $t \geq 0$. The proof for the cases \eqref{f-Rcha-3} and \eqref{f-Rcha-4} follow the same lines, using the second estimate of Proposition 3.2 and Proposition 3.5 in \cite{wang88}.

\stepp Now we consider the estimates in $\Lc(L^2(\O),H^1(\O))$. The domain $\Dom(\Hah)$ is invariant by pseudo-differential operators in the $x$-variable with bounded symbols, so for $\f \in L^2(\O)$ we have $\Rcha \f \in \Dom(\Hah)$ and hence 
\begin{align} \label{eq-nabla-Re}
\nr{\nabla \Rcha \f}_{L^2(\O)}^2 = \frac 1 {h^2} \Re \innp{\Hah \Rcha\f}{\Rcha \f}_{L^2(\O)}.
\end{align}
We consider the case \eqref{f-Rcha-1}. Then we have 
\begin{eqnarray*}
\lefteqn{\innp{\Hah \Rcha\f}{\Rcha \f}_{L^2(\O)}}\\
&& = \innp{[\LDh,\pppg x^{-\d}]\Rh \pppg x^{-\d} \f}{\Rcha \f} + \innp{\pppg x^{-\d} (\Hah-1) \Rh \pppg x^{-\d} \f}{\Rcha \f} + \nr{\Rcha \f}_{L^2(\O)}^2\\
&& \lesssim \frac {\nr{\f}_{L^2(\O)}^2}{h^2} + h \nr{\Rcha \f}_{H^1(\O)} \nr{\f}_{L^2(\O)}.
\end{eqnarray*}
For $h$ small enough we obtain
\begin{align*} 
\nr{\nabla \Rcha \f}^2 \lesssim \frac {\nr{\f}^2}{h^4}.
\end{align*}
We proceed similarly for the other cases. We only have to be careful with the commutators of the form $[\LDh,\Opwx(b_-)]$. For instance for the case \eqref{f-Rcha-2}, the commutator $[\LDh,\Opwx(b_-)]$ is a pseudo-differential operator whose symbol is supported in an incoming region and decays at least like $\pppg x\inv$. Thus we can use the case \eqref{f-Rcha-1} if $\d - 1 - \rho < \frac 12$. Then we can prove by induction on $N \in \N$ the estimate for the case \eqref{f-Rcha-2} when $\d - 1 - \rho < \frac 12 + N$. 

\stepp All the estimates which we have proved have analogs if we replace $\Rh$ by its adjoint and if we change the roles of the symbols $b_-$ and $b_+$. We also have to consider negative times in \eqref{estim-entrante} and write
\[
\big(\Hah^* - \bar \z \big)\inv = -\frac ih \int_0^{+\infty} e^{\frac {i\th}h (\Hah^* - \bar \z)} \, d\th.
\]
This gives for instance for $b_- \in \symb_-(R,\n,\s_-)$
\[
\nr{\pppg x^{-\d} \Rh^* \Opwx(b_-) \pppg x^{\d-1-\rho} }_{\Lc(L^2(\O))} \leq \frac c h.
\]
We also have estimates for $\Rcha^*$ in $\Lc(L^2(\O),\HuO)$. Taking the adjoints gives the required estimates for $\Rcha$ in $\Lc(\HuOp,L^2(\O))$. Finally for the estimates in $\Lc(\HuOp,\HuO)$ we proceed as above, estimating $\f$ in $\HuOp$.
\end{proof}

\begin{proposition} \label{prop-Rcha}
Let $p \in \N^*$, $k_1,\dots,k_p \in \N^*$ and $k = k_1+ \dots + k_p$. Let $(\Rcha)$ in $\mathfrak R^{k_1,\dots,k_p}$. Let $\b_1,\b_2 \in \{0,1\}$. Then there exist $h_0 > 0$ and $c \geq 0$ such that for $h \in ]0,h_0]$ we have
\begin{equation} \label{estim-Rcha}
\nr{\Rcha}_{\Lc(H^{\b_1}(\O)',H^{\b_2}(\O)))} \leq \frac c {h^{k + \b_1 + \b_2}}
\end{equation}
and for all $\f \in L^2(\O)$:
\begin{equation} \label{estim-Rcha-qa}
q_a \big( \Rcha \f \big) \leq \frac {c \nr \f^2}{h^{2k}}.
\end{equation}
\end{proposition}

\begin{proof}
\stepp We begin with the case $p=1$, which means that $\Psi_{h} = \Rh^k$. We first consider the estimates in $\Lc(L^2(\O))$. If $\Rcha$ is of the form \eqref{f-Rcha-0}, then we write $\tilde \h_1 = \tilde \h_1 \tilde \h_2 \dots \tilde \h_k$ where $\tilde \h_j \in C_0^ \infty(\R,[0,1])$ is supported in $]-\g,\g[$ and equal to 1 on a neighborhood of $\supp(\tilde \h_1)$ for all $j \in \Ii 2 k$. The operator $\tilde \h_j(\LDh)$ commutes with $\Rh$ for all $j \in \Ii 1 k$ so by Proposition \ref{prop-small-long-freq}
\[
\nr{\h_1(\LDh) \Rh^k}_{\Lc(L^2(\O))} \leq \prod _{j=1}^k \nr{\tilde \h_j (\LDh) \Rh}_{\Lc(L^2(\O))} \lesssim \frac 1 {h^k}.
\]
The cases \eqref{f-Rcha-1}-\eqref{f-Rcha-4} are proved by induction on $k$. The strategy is quite standard. We recall the idea, which will also be used to get the general result. By proposition \ref{prop-incoming-outgoing}, we already have the result when $k = 1$, so we assume that $k \geq 2$. Let $\h_0 \in C_0^\infty(\R^\dd)$ be equal to 1 on a neighborhood of 0. Let $\tilde \h_0 \in C_0^\infty(\R,[0,1])$ be equal to 1 on a neighborhood of 0. Let $\h_+ \in C_0^\infty([-1,1],[0,1])$ be equal to 0 on a neighborhood of -1 and equal to 1 on a neighborhood of 1. Let $\h_- = 1 - \h_+$ and, for $(x,\x) \in \R^{2\dd}$:
\[
\b_\pm (x,\x) = (1-\h_0)(x) (1-\tilde \h_0)(\x^2)  \h_\pm \left( \frac {\innp x \x}{\abs x \abs \x} \right).
\]
Then $\b_\pm$ belongs to $\Sc_\pm (R,\n,\s_\pm)$ for some $R > 0$, $\n > 0$ and $\s_\pm \in ]-1,1[$ and we have 
\begin{equation} \label{dec-0+-}
(1-\tilde \h_0)(\LDh) = \Opwx \big(\h_0(x) (1-\tilde \h_0(\x^2))\big) + \Opwx (\b_+) + \Opwx(\b_-).
\end{equation}
Let $\rho \in \big] 0 , \d - k + \frac 12 \big[$. We have
\begin{eqnarray*}
\lefteqn{\nr{\pppg x^{-\d} \Rh^k \pppg x^{-\d}}_{\Lc(L^2(\O))} \lesssim \nr{\pppg x^{-\d} \tilde \h_0(\LDh) \Rh^k \pppg x^{-\d}}}\\
&& \quad + \nr{\pppg x^{-\d} \Rh \pppg x^{-\d}} \nr{\pppg x^{-\d} \Rh^{k-1} \pppg x^{-\d}}\\
&& \quad + \nr{\pppg x^{-\d} \Rh \Opwx(\b_+) \pppg x^{\d-1-\rho}} \nr{\pppg x^{1+ \rho -\d} \Rh^{k-1} \pppg x^{-\d}}\\
&& \quad + \nr{\pppg x^{-\d} \Rh \pppg x^{-\d+k-1+\rho}} \nr{\pppg x^{\d-k+1-\rho} \Opwx(\b_-) \Rh^{k-1} \pppg x^{-\d}}.
\end{eqnarray*}
The last three terms are given by the product of the norm of an operator in $\mathfrak{R}^1$ and the norm of an operator in $\mathfrak R^{k-1}$, so by the case \eqref{f-Rcha-0} and the inductive assumption we get 
\[
\nr{\pppg x^{-\d} \Rh^k \pppg x^{-\d}}_{\Lc(L^2(\O))} \lesssim \frac 1 {h^k}.
\]
We prove the estimate in the other cases similarly. For instance for \eqref{f-Rcha-2} we write
\begin{eqnarray*}
\lefteqn{\nr{\pppg x^{\d-k-\rho} \Opwx(b_-) \Rh^k \pppg x^{-\d}}\lesssim \nr{\pppg x^{\d-k-\rho} \Opwx(b_-) \tilde \h_0 (\LDh)\Rh^k \pppg x^{-\d}}}\\
&& + \nr{\pppg x^{\d-k-\rho} \Opwx(b_-) \Rh \pppg x^{-\d}} \nr{\pppg x^{-\d} \Rh^{k-1} \pppg x^{-\d}}\\
&& + \nr{\pppg x^{\d-k-\rho} \Opwx(b_-) \Rh \Opwx(\b_+) \pppg x^{\d}} \nr{\pppg x^{-\d} \Rh^{k-1} \pppg x^{-\d}}\\
&& + \nr{\pppg x^{\d-k -\rho} \Opwx(b_-) \Rh \pppg x^{k-1-\d+\frac \rho 2}} \nr{\pppg x^{\d-k+1- \frac \rho 2} \Opwx(\b_-) \Rh^{k-1} \pppg x^{-\d}}.
\end{eqnarray*}
For the first term we observe that if $\tilde \h_0$ is supported close enough to 0 then $\Opwx(b_-) \tilde \h_0(\LDh)$ is a pseudo-differential operator whose symbol decays like any power of $h$ and any power of $\pppg x\inv$. If $\b_+$ was suitably chosen then we can conclude again by induction for the last three terms. We proceed similarly for \eqref{f-Rcha-3} and \eqref{f-Rcha-4}, which gives the estimates in $\Lc(L^2(\O))$. For the general estimates in $\Lc(H^{\b_1}(\O)',H^{\b_2}(\O))$ we proceed as in the proof of Proposition \ref{prop-incoming-outgoing}.

\stepp Now we prove \eqref{estim-Rcha-qa} for $p= 1$. Let $\f \in L^2(\O)$. As in \eqref{eq-nabla-Re} we write 
\begin{align} \label{eq-qa-Im}
q_a (\Rcha \f) = -\frac 1 h \Im \innp{(\Hah-1) \Rcha\f}{\Rcha \f}.
\end{align}
Then we proceed as in the proof of Proposition \ref{prop-incoming-outgoing}. For instance in the case \eqref{f-Rcha-1} we obtain
\begin{align} \label{estim-qa}
q_a (\Rcha \f) \lesssim \frac {\nr{\f}}{h^{k+1}} \left( \nr{\big[\LDh ,\pppg x^{-\d} \big] \Rh^k \pppg x^{-\d} \f} + \nr{\pppg x^{-\d} \Rh^{k-1} \pppg x^{-\d} \f} \right) \lesssim \frac {\nr{\f}^2}{h^{2k}}.
\end{align}
The other cases are similar, and this concludes the proof of the proposition for $p = 1$.

\stepp Then we proceed by induction on $p$. So let $p \geq 2$ and $k^\sharp = k_2 + \dots + k_p$. We consider the estimate in $\Lc(L^2(\O))$ for the case \eqref{f-Rcha-1}. We set $\Rh^\sharp = \Rh^{k_2}\Th_a \dots \Th_a \Rh^{k_p}$. We define $\tilde \h_2$ as at the beginning of the proof ($\tilde \h_0 \tilde \h_2 = \tilde \h_0$). Since $\tilde \h_0(\LDh)$ and the three operators in the right-hand side of \eqref{dec-0+-} commute with $\Th_a$ we can write for $\f,\p \in L^2(\O)$ and $\rho > 0$ small enough:
\begin{eqnarray*}
\lefteqn{\innp{\pppg x^{-\d} \Rh^{k_1} \Th_a \Rh^\sharp \pppg x^{-\d} \f} {\p} = \innp{   \Th_a \tilde \h_2(\LDh) \Rh^\sharp \pppg x^{-\d}\f} {\tilde \h_0(\LDh)(\Rh^*)^{k_1}\pppg x^{-\d}\p} }\\
&& + \innp{ \Th_a \pppg x^\d \Opwx \big(\h_0(x) (1-\tilde \h_0)(\abs\x^2) \big) \Rh^\sharp \pppg x^{-\d}\f} { \pppg x^{-\d} (\Rh^*) ^{k_1}\pppg x^{-\d}\p}\\
&& + \innp{\Th_a \pppg x^{k_1 + \rho -\d}\Rh^\sharp \pppg x^{-\d} \f} {\pppg x^{\d-k_1-\rho}  \Opwx(\b_+)  (\Rh^*)^{k_1}\pppg x^{-\d}\p}\\
&& + \innp{ \Th_a \pppg x^{\d-k^\sharp-\rho} \Opwx(\b_-) \Rh^\sharp \pppg x^{-\d} \f} { \pppg x^{k^\sharp + \rho -\d} (\Rh^*)^{k_1}\pppg x^{-\d}\p}.
\end{eqnarray*}
Since the form $q_a$ is non-negative we can apply the Cauchy-Schwarz inequality in each term. 
If $\rho$ is small enough, then \eqref{estim-Rcha-qa} applied to $\Rh^{k_1}$ and $\Rh^\sharp$ (and their adjoints) gives \eqref{estim-Rcha}. Again, the other cases are proved similarly.

\stepp Now we prove the estimates in $\Lc(L^2(\O),H^1(\O))$ as we did in the proof of Proposition \ref{prop-incoming-outgoing}. We first assume that $k_1 = 1$ and consider the case \eqref{f-Rcha-1}. We start from \eqref{eq-nabla-Re}. For $\f \in L^2(\O)$ we obtain 
\begin{align*}
\nr{\nabla \Rcha \f}^2
& \leq \frac 1 {h^2} \nr{\Rcha \f}^2 + \frac 1 {h^2} \nr{\big[\LDh, \pppg{x}^{-\d} \big] \Rh \Th_a \Rh^\sharp \pppg{x}^{-\d}\f} \nr{\Rcha\f}\\
& \quad + \frac 1 {h^2}\abs{\innp{ \Th_a \pppg{x}^{-\d} \Rh^\sharp \pppg{x}^{-\d}\f}{\Rcha \f}}. 
\end{align*}
By the Cauchy-Schwarz inequality and the already available estimates we get
\begin{align*}
\nr{\nabla \Rcha \f}^2 
& \lesssim \frac {\nr \f^2} {h^{2k+2}} + \frac {\nr{\Rcha \f}_{H^1(\O)} \nr{\f}}{h^{k+1}} + \frac 1 {h^2} q_a \left(\pppg{x}^{-\d}(\Rh^\sharp)^* \pppg{x}^{-\d}\f \right)^{\frac 12} q_a \left(\Rcha \f \right)^{\frac 12}\\
& \lesssim \frac {\nr \f^2} {h^{2k+2}} + \frac {\nr{\nabla \Rcha \f} \nr{\f}}{h^{k+1}} + \frac {\nr\f} {h^{k + 1}} q_a \left(\Rcha \f \right)^{\frac 12}.
\end{align*}
On the other hand, starting from \eqref{eq-qa-Im}, we similarly obtain
\[
q_a(\Rcha \f) \lesssim \frac {\nr{\Rcha \f}_{H^1(\O)} \nr \f}{h^{k-1}} + \frac {\nr\f} {h^{k}} q_a \left(\Rcha \f \right)^{\frac 12}.
\]
Together, these two inequalities yield 
\[
\nr{\nabla\Rcha \f} \lesssim \frac {\nr \f}{h^{k+1}} \qandq q_a(\Rcha \f) \lesssim \frac{\nr{\f}^2}{h^{2k}}.
\]
Then we finally obtain the estimates in $\Lc(\HuOp,L^2(\O))$ and $\Lc(\HuOp,\HuO)$ as we did in the proof of Proposition \ref{prop-incoming-outgoing}. This concludes the proof when $k_1 = 1$. Then we proceed by induction on $k_1$, following the same idea. Notice that for $k_1 \geq 2$ we no longer have to prove the estimate on $\nr{\nabla \Rcha \f}$ and $q_a(\Rcha \f)$ simultaneously.
\end{proof}

Now we can finish the proof of Theorem \ref{th-high-freq}:

\begin{proof} [Proof of Theorem \ref{th-high-freq}]
Let $\b_1,\b_2 \in \{0,1\}$. By Proposition \ref{prop-der-tRaz}, \eqref{eq-res-semiclass} and Proposition \ref{prop-Rcha} we have for any $\Nder \in \N$ and $\d > \Nder + \frac 12$
\begin{equation} \label{estim-der-tRat}
\nr{\h_\t(\LD)\tilde R_a^{(\Nder)}(\t)}_{\Lc(H^{\b_1}(\O)',H^{\b_2}(\O))} +  \nr{\pppg x^{-\d} \tilde R_a^{(\Nder)}(\t) \pppg x^{-\d}}_{\Lc(H^{\b_1}(\O)',H^{\b_2}(\O))} \lesssim  \t ^{\b_1 + \b_2 -1}.
\end{equation}
Let $\Nder \in \N$ and $\d > \Nder + \frac 12$. We take the derivative of order $\Nder$ in \eqref{res-Ac-U-bis}. With \eqref{estim-der-tRat} we obtain for $\abs \t \geq 1$ and $U = (u,v) \in \HH$
\[
\nr{(\Ac-\t)^{-\Nder-1}U}_{\EE^{-\d}} \lesssim \nr{\pppg x^{-\d} \nabla u} + \nr{\pppg x^\d \tilde \D u}_{\HuOp} + \nr{\pppg x^\d v}_{L^2(\O)}.
\]
But for $v \in C_0^\infty(\bar \O)$ we have 
\begin{align*}
\innp{\pppg x^\d \tilde \D u}{v}
& = \innp{ \nabla u}{\nabla \pppg x^\d v} \lesssim  \left(\innp{ \nabla u}{\pppg x^\d \nabla v} + \innp{ \nabla u}{\pppg x^{\d-1 } v} \right)\\
& \lesssim \nr{\pppg x^\d \nabla u}_{L^2(\O)} \nr{v}_{\HuO},
\end{align*}
so 
\[
\nr{\pppg x^\d \tilde \D u}_{\HuOp} \lesssim \nr{\pppg x^\d \nabla u}_{L^2(\O)}.
\]
This proves that 
\[
\nr{(\Ac-\t)^{-\Nder-1}}_{\EE^{-\d}} \lesssim \nr{\pppg x^{\d} \nabla u} + \nr{\pppg x^\d v}_{L^2(\O)} = \nr{U}_{\EE^\d},
\]
which gives the first estimate of Theorem \ref{th-high-freq}. The other estimates are proved similarly.
\end{proof}

\appendix

\section{Spectral gap for the transverse operator} \label{sec-gap-Taz}

In this appendix we give a proof of Theorem \ref{th-gap-Tah}. For this we will use semiclassical technics and in particular the contradiction argument of \cite{lebeau96}. Notice that in this section we only consider functions on $\o$ or $\R^\nn$, so without ambiguity we can simply denote by $\D$ the Laplacian with respect to the variable $y$.\\
\\

By unique continuation, it is not difficult to see that for $\a \in \R \setminus \singl 0$ and $h > 0$ the operator $\Tah$ has no real eigenvalue. Then, if we can prove that the resolvent $(\Tah-\l)\inv$ for $\l \in \R$ close to 1 is of size $O(h\inv)$, the standard perturbation argument proves that there is a spectral gap of size $O(h)$ and the resolvent is of size $O(h\inv)$ for $\l$ in this region. Thus it is enough to prove Theorem \ref{th-gap-Tah} for $\l$ real. It is also enough to prove the result for $\a$ real, but this is less clear:

\begin{lemma} \label{lem-perturb-Th}
Assume that there exist $h_0 \in ]0,1]$, $\g \in ]0,1[$ and $c \geq 0$ such that for $h \in ]0,h_0]$ and $\a ,\l \in ]1-\g , 1 + \g[$ we have 
\[
\nr{(\Tah -\l)\inv}_{\Lc(L^2(\o))} \leq \frac c h.
\]
Then the statement of Theorem \ref{th-gap-Tah} holds (maybe with different constants $h_0$, $\g$ and $c$).
\end{lemma}

\begin{proof}
As in the proof of Lemma \ref{lem-res-Ta-H1}, we can check that for $\a ,\l \in ]1-\g,1+\g[$ the resolvent $(\Tah - \l)\inv$ extends to an operator $(-h^2 \D -ih \Th_\a - \l)\inv \in \Lc(\Huop,\Huo)$ and for $\b_1,\b_2 \in \{ 0,1\}$ we have 
\begin{equation} \label{estim-res-Tah}
\nr{(-h^2 \D -ih \Th_\a - \l)\inv}_{\Lc(H^{\b_1}(\o)',H^{\b_2}(\o))} \lesssim \frac 1 {h^{1+\b_1+\b_2}}.
\end{equation}
Let $\a , \l \in ]1-\g,1+\g[$ and $s \in [0,\a]$. In $\Lc(\Huo,\Huop)$ we have
\[
(-h^2 \D -ih \Th_{\a-ihs} - \l) = (-h^2 \D -ih \Th_\a - \l) \left( 1 - h (-h^2 \D -ih \Th_\a - \l)\inv \Th_{hs} \right).
\]
For $v \in \Huo$ and $\f \in \Huop$ we have 
\[
\abs{\innp{(-h^2 \D -ih \Th_\a - \l)\inv \Th_{hs} v}{\f}} \leq q^\o_{hs} (v)^{\frac 12} q^\o_{hs} \big( (-h^2 \D +ih\Th_\a - \l)\inv \f \big)^{\frac 12},
\]
where the form $q^\o$ is defined as $q$ (see \eqref{def-qa}) with $\O$ replaced by $\o$ (we recall that $\Th_\a$ can be viewed as an operator in $\Lc(\Huo,\Huop)$). Since $s \leq \a$ we have $q^\o_{hs} \leq q^\o_{h\a}$. By \eqref{estim-res-Tah} and an equality analogous to \eqref{eq-qa-Im} we obtain 
\begin{eqnarray*}
\lefteqn{\abs{\innp{(-h^2 \D -ih \Th_\a - \l)\inv \Th_{hs} v}{\f}}}\\
&& \leq \sqrt {hs} \nr{v}_{\Huo}  \nr{(-h^2 \D + ih\Th_\a - \l)\inv \f}_{\Huo}^{\frac 12} \nr{\f}_{\Huop}^{\frac 12}\\
&& \lesssim \frac {\sqrt s} h \nr{v}_{\Huo} \nr{\f}_{\Huop}.
\end{eqnarray*}
This proves that for $s\geq 0$ small enough we have 
\[
\nr{h (-h^2 \D -ih \Th_\a - \l)\inv \Th_{hs}}_{\Lc(\Huo)} \leq \frac 12.
\]
Then $(-h^2 \D -ih \Th_{\a-ihs} - \l)$ has an inverse in $\Lc(\Huop,\Huo)$ and 
\[
\nr{(-h^2 \D -ih \Th_{\a-ihs} - \l)\inv}_{\Lc(\Huo)} \leq \frac {2c}h.
\]
We can similarly add an imaginary part of size $O(h)$ to the spectral parameter $\l$.
\end{proof}

By Lemma \ref{lem-perturb-Th} and by density of $C_0^\infty(\o)$ in $L^2(\o)$, it is enough to prove that there exists $\g > 0$, $h_0 \in ]0,1]$ and $c \geq 0$ such that for $h \in ]0,h_0]$, $\a,\l \in ]1-\g,1+\g[$ and $f \in C_0^\infty(\o)$ we have
\begin{equation} \label{estim-gap-Tah}
\nr{(\Tah -\l)\inv f}_{L^2(\o)} \leq  \frac c h \nr{f}_{L^2(\o)} .
\end{equation}
We prove \eqref{estim-gap-Tah} by contradiction. If the statement is wrong, then we can find sequences $\seq h m \in ]0,1]^\N$, $\seq \a m \in \R^\N$, $\seq \l m \in \R^\N$ and $\seq f m \in C_0^\infty(\o)^\N$ such that $h_m \to 0$, $\a_m \to 1$, $\l_m \to 1$ and, if we set $u_m = (\Tahm - \l_m)\inv f_m$, then $\nr{u_m}_{L^2(\o)} = 1$ and $\nr{f_m}_{L^2(\o)} = o(h_m)$. We first notice that by elliptic regularity we have $u \in C^\infty(\bar \o)$ for all $m \in \N$ (but we have no other uniform estimate on $u_m$ than the one in $L^2(\o)$).\\

For $m \in \N$ we consider the function $\tum \in L^2(\R^\nn)$ equal to $u_m$ on $\o$ and equal to 0 outside $\o$. We have $\nr{\tum}_{L^2(\R^\nn)} = 1$ for all $m$. We consider a semiclassical measure for this family: after extracting a subsequence if necessary, there exists a Radon measure $\m$ on $\R^{2\nn} \simeq T^* \R^\nn$ such that for all $q \in C_0^\infty(\R^{2\nn})$ we have 
\begin{equation} \label{lim-semiclass-measure}
\innp{\Opwm(q) \tum}{\tum}_{L^2(\R^\nn)} \limt m \infty \int_{\R^{2\nn}} q \, d\m.
\end{equation}
In order to obtain a contradiction and conclude the proof of Theorem \ref{th-gap-Tah}, we prove that $\m \neq 0$ and $\m = 0$ (see Propositions \ref{prop-mu-not-zero} and \ref{prop-mu-zero}). We first observe that since $\tum = 0$ outside $\o$, the measure $\m$ is supported in $\bar \o \times \R^\nn$.\\

\begin{lemma} \label{lem-traces-o}
We have 
\[
 \int_{\partial \o} \abs{u_m}^2  +  \int_{\partial \o} \abs{h_m \partial_\n u_m}^2 \limt m \infty 0.
\]
Moreover there exists $C \geq 0$ such that for all $m \in \N$
\[
\nr{h_m \nabla u_m}_{L^2(\o)} \leq C.
\]
\end{lemma}

\begin{proof}
Since $h_m \partial_\n u_m = i \a_m u_m$ on $\partial \o$ we have 
\[
\nr{h_m \nabla u_m}_{L^2(\o)}^2 - i \a_m \int_{\partial \o} \abs{u_m}^2 - \l_m \nr{u_m}_{L^2(\o)}^2 = \innp{f_m}{u_m} \limt m \infty 0.
\]
Taking the real and imaginary parts gives the two statements of the proposition.
\end{proof}

\begin{lemma} \label{lem-loc-energy}
Let $\h \in C_0^\infty(\R,[0,1])$ be equal to 1 on a neighborhood of 1. Then we have 
\[
\innp{(1-\h)(-h_m^2 \D) \tum}{\tum}_{L^2(\R^\nn)} \limt m \infty 0.
\]
\end{lemma}

\begin{proof}
For $m \in \N$ large enough we can set 
\[
v_m = (1-\h)(-h_m^2\D) (-h_m^2 \D - {\l_m})\inv \tum \quad \in L^2(\R^n).
\]
Then for $\th \in \{0,1,2\}$ there exists $C_\th \geq 0$ such that
\begin{equation} \label{estim-vm}
h_m^\th \nr{v_m}_{H^\th(\R^\nn)} \leq C_\th.
\end{equation}
We have  
\begin{eqnarray*}
\lefteqn{\innp{(1-\h)(-h_m^2 \D) \tum}{\tum}_{L^2(\R^\nn)}  = \innp{\big(-h_m^2 \D - {\l_m} \big) v_m}{\tum}_{L^2(\R^\nn)}}\\
&& = \innp{\big(-h_m^2 \D - {\l_m} \big) v_m}{u_m}_{L^2(\o)}\\
&& = - h_m^2 \int_{\partial \o}\partial_\nu {v_m}  \,  \bar{u_m} + h_m^2 \int_{\partial \o} v_m   \,\partial_\nu \bar {u_m} + \innp{v_m}{f_m}_{L^2(\o)},
\end{eqnarray*}
so by the trace theorems
\begin{multline*}
\abs{\innp{(1-\h)(-h_m^2 \D) \tum}{\tum}_{\R^\nn}}\\
\leq \nr{u_m}_{L^2(\partial \o)} h_m^2 \nr{v_m}_{H^{2}(\R^\nn)} + \nr{h_m \partial_\n u_m}_{L^2(\partial \o)}  h_m \nr{v_m}_{H^{1}(\R^\nn)} + \nr{f_m}_{L^2(\o)} \nr{v_m}_{L^{2}(\R^\nn)}.
\end{multline*}
We conclude with \eqref{estim-vm} and Lemma \ref{lem-traces-o}.
\end{proof}

\begin{proposition} \label{prop-mu-not-zero}
We have $\m \neq 0$.
\end{proposition}

\begin{proof}
Let $\h \in C_0^\infty(\R,[0,1])$ be equal to 1 on a neighborhood of 1. Let $\h_\o \in C_0^\infty(\R^\nn,[0,1])$ be equal to 1 on a neighborhood of $\bar \o$. For $(y,\y) \in \R^{2\nn}$ we set 
\[
\tilde \h(y,\y) = \h_\o(y) \h \big(\abs \y^2 \big).
\]
By compactness of the suppport of $\tilde \h$ we have 
\begin{align*}
\int_{\R^{2\nn}} \tilde \h \, d\m
& = \lim_{m \to \infty} \innp{\Opwm(\tilde \h) \tum}{\tum}_{\R^\nn}  = \lim_{m \to \infty} \innp{\h(-h_m^2 \D) \tum}{\tum}.
\end{align*}
By Lemma \ref{lem-loc-energy} this last limit is equal to 1. This implies in particular that $\m \neq 0$.
\end{proof}

The main difficulty for the proof of Theorem \ref{th-gap-Tah} is the propagation of the measure $\m$. As already mentioned, this question is simplified by the fact that in our setting the damping is effective everywhere on the boundary. This explains why we do not have to consider generalized bicharacteristics on $T^* \o$. Here we simply have invariance of the measure by the flow on $T^* \R^{\nn}$.

\begin{proposition} \label{prop-prop-mesure}
Let $q \in C_0^\infty(\R^{2\nn})$ and $t \in \R$. Then we have 
\[
\int_{\R^{2\nn}} q(y,\y) \, d \m = \int_{\R^{2\nn}} q (y - 2 t \y , \y) \, d\m.
\]
\end{proposition}

Many arguments used in the proof of this proposition are inspired by \cite{miller00}.

\begin{proof}
\stepp By differentiation under the integral sign we have 
\begin{equation*}
\frac d {ds} \int_{\R^{2\nn}} q (y - 2 s \y , \y) \, d\m = -2 \int_{\R^{2\nn}}  \y \cdot \nabla_y q (y - 2 s \y , \y) \, d\m.
\end{equation*}
So it is enough to prove that for all $q \in C_0^\infty(\R^{2\nn})$ we have 
\begin{equation} \label{eq-crochet-poisson}
\int_{\R^{2\nn}} \big\{\y^2,q \big\} \, d\m  = 0.
\end{equation}

\stepp This is clear if $q$ and hence $\y \cdot \nabla_y q$ are supported outside $\bar \o \times \R^\nn$. Now let $q$ be supported in $\o \times \R^\nn$. Let $\h \in C_0^\infty(\o)$ be such that $\big(\supp(1-\h) \times \R^n\big) \cap \supp (q) = \emptyset$. We can write 
\begin{align*}
\lim_{m \to \infty} \innp{\Opwm\big(\{ \y^2 , q\} \big)  \tum } {\tum}_{L^2(\R^\nn)}
& = \lim_{m \to \infty}  \frac i {h_m} \innp{ [-h_m^2 \D,\Opwm(q)]  \tum} { \tum}_{L^2(\R^\nn)} \\
& \hspace{-1cm}= - \lim_{m \to \infty}  \frac 2 {h_m} \Im \innp{\Opwm(q)  \tum} {(-h_m^2 \D - \l_m)  \tum}_{L^2(\R^\nn)}  \\
& \hspace{-1cm}= - \lim_{m \to \infty}  \frac 2 {h_m} \Im \innp{\Opwm(q)  \tum} {\h (-h_m^2 \D - \l_m)  \tum}_{L^2(\R^\nn)}  \\
& \hspace{-1cm}= - \lim_{m \to \infty}  \frac 2 {h_m} \Im \innp{\Opwm(q)  \tum} {\h f_m}_{L^2(\R^\nn)}  \\
& \hspace{-1cm}= 0.
\end{align*}
This proves \eqref{eq-crochet-poisson} for $q$ supported in $\o \times \R^\nn$. By linearity it remains to prove that for any $y \in \partial \o$ there exists a neighborhood $\Uc_y$ of $y$ in $\R^\nn$ such that \eqref{eq-crochet-poisson} holds for $q$ supported in $\Uc_y \times \R^\nn$.

\stepp So let $y_0 \in \partial \o$. We first make a change of variables to reduce to the case where $\o$ looks like the half space $\R^\nn_+$ around $y_0$. Notice that this is already the case if $\nn = 1$, so for this part of the proof we can assume that $\nn \geq 2$. For $r > 0$ we denote by $B'(r)$ the open ball of radius $r$ in $\R^{\nn-1}$. Since $\partial \o$ is a smooth manifold of dimension $\nn-1$, there exist a neighborhood $\Wc_\partial$ of $y_0$ in $\partial \o$, $\rho > 0$ and a diffeomorphism $\vf_\partial : \Wc_\partial \to B'(2\rho)$ such that $\vf_\partial (y_0) = 0$. For $y \in \partial \o$ we denote by $\n(y)$ the outward normal vector of $\o$ at $y$. Then by the tubular neighborhood theorem (see for instance Paragraph 2.7 in \cite{berger-gostiaux}), taking $\Wc_\partial$ and $\rho$ smaller if necessary, the map 
\[
\tilde \vf : \fonc{B'(2\rho) \times ]-2\rho , 2 \rho[} {\R^\nn} {(y',s)} {\vf_\partial \inv(y') - s \n \big(\vf_\partial \inv(y')\big)}
\]
defines a diffemorphism from $\Vc_{2\rho} := B'(2\rho) \times ]-2\rho , 2 \rho[$ to its image $\Wc_{2\rho} := \vf (\Vc_{2\rho})$. Thus $\vf = \tilde \vf \inv$ defines a diffeomorphism from a neighborhood $\Wc_{2\rho}$ of $y$ in $\R^n$ to $\Vc_{2\rho}$ such that $\vf(y) = 0$ and $\vf(\Wc_{2\rho} \cap \o) = \Vc_{2\rho}^+ := B'(2\rho) \times ]0,2\rho[$. We write $\vf = (\vf_1,\dots,\vf_\nn)$ where $\vf_j \in C^\infty(\Wc_{2\rho},\R)$ for all $j \in \Ii 1 \nn$. We set $\Vc_{\rho} = B'(\rho) \times ]-\rho,\rho[$, $\Vc_{\rho}^+ = B'(\rho)\times ]0,\rho[$, $\Wc_\rho = \vf\inv(\Vc_\rho)$ and consider $\h \in C_0^\infty(\R^\nn,[0,1])$ supported in $\Wc_{2\rho}$ and equal to 1 on a neighborhood of $\bar {\Wc_{\rho}}$. We prove \eqref{eq-crochet-poisson} for $q$ supported in $\Wc_\rho \times \R^\nn$.

\stepp For $m \in \N$ and $v \in C^\infty(\Vc_{2\rho})$ we have 
\[
 \big( -h_m^2 \D (v \circ \vf)\big)\circ \vf\inv = P_m v 
\]
where $P_m$ is of the form 
\begin{equation} \label{def-Pm}
P_m = A(y) \Dddp^2  + B (y,\Dp) \Ddd  + C(y,\Dp)  + h_m \tilde b(y) \Ddd  + h_m \tilde C(y,\Dp).
\end{equation}
Here $\Ddd$ stands for $-ih_m \partial_{y_\nn}$ and the operators $B(y,\Dp),C(y,\Dp), \tilde C(y,\Dp)$ are differential operators (of orders 1,2 and 1, respectively) in the first $(\nn-1)$ variables with smooth coefficients on $\Vc_{2\rho}$. We denote by $b,c,\tilde c \in C^\infty(\Vc_{2\rho} \times \R^{\nn -1})$ their symbols. We can check that with this choice for the diffeomorphism $\vf$ we have on $\Vc_{2\rho}$
\begin{equation} \label{A-equal-1}
A(y) = \nr{\nabla \vf_\nn(y)}^2 = 1.
\end{equation}
On the other hand the operator $P_m$ is symmetric on $L^2(\Vc_{2\rho})$. Thus the formal adjoint 
\[
P_m^* =  \Dddp^2  +  \Ddd B (y,\Dp)^*  + C(y,\Dp)^*  + h_m \Ddd \bar{\tilde b(y)}  + h_m \tilde C(y,\Dp)^* 
\]
satisfies $P_m^* = P_m$ for all $m \in \N$. We denote by $p$ the principal symbol of $P_m$:
\[
p(y,\y) = \y_\nn^2 + b(y,\y') \y_\nn + c(y,\y').
\]
Here we write $\y = (\y',\y_n)$ with $\y' \in \R^{\nn-1}$ and $\y_n \in \R$.
For $m \in \N$ we set $v_m = (\h u_m) \circ \vf\inv$. This defines a smooth function on $\Vc_{2\rho}^+$ which can be extended by 0 to a smooth function on $\R_\nn^+$. We denote by $\tvm$ its extension by 0 on $\R^\nn$. The choice of $\vf$ ensures that on $\partial \R_\nn^+\cap \Vc_\rho$ we have 
\[
h_m \partial_\nu v_m = i \a_m v_m.
\]
If we choose $\h$ such that $\partial_\nu \h = 0$ on $\partial \o$, then these equalities hold in fact everywhere on $\partial \R^\nn$. With Lemma \ref{lem-traces-o} we can check that 
\begin{equation} \label{estim-traces-vm}
\nr{v_m}_{L^2(\partial \R_\nn^+)} + \nr{\Ddd v_m} _{L^2(\partial \R_\nn^+)}   \limt m \infty 0,
\end{equation}
and
\begin{equation} \label{estim-nabla-vm}
\nr{h_m \nabla v_m} \lesssim 1.
\end{equation}

      \detail 
      {
      $u = v \circ \Psi$
      \[
      \partial_{y_j} u (y) = \sum_{k} \partial_{w_k} v (\Psi(y)) \partial_{y_j} \Psi_k (y)
      \]
      \[
      \partial_{y_j}^2 u (y) = \sum_{k} \partial_{w_k} v (\Psi(y)) \partial^2_{y_j} \Psi_k(y) + \sum_{k,l} \partial_{w_k,w_l} v (\Psi(y)) \partial_{y_j} \Psi_k(y) \partial_{y_j} \Psi_l(y),
      \]
      }

We set $g_m = (P_m-\l_m) v_m$. Then
\begin{equation} \label{expr-gm}
g_m = \h f_m \circ \vf \inv - h_m^2\big(2  \nabla \h \cdot \nabla u_m + u_m \D \h \big) \circ \vf\inv
\end{equation}
and
\begin{equation} \label{eq-Ddd2-vm}
\begin{aligned}
\Ddd^2 v_m
&  =  g_m +  \l_m v_m - B (y,\Dp) \Ddd v_m  - C(y,\Dp)v_m  \\
& \quad  -  h_m \tilde b(y) \Ddd v_m - h_m \tilde C(y,\Dp) v_m .
\end{aligned}
\end{equation}
With \eqref{estim-traces-vm} we obtain that if $\p \in C_0^\infty(\Vc_\rho \times \R^{\nn-1})$ and $\Psi_m = \Opwm (\p)$ then 
\begin{equation} \label{trace-Ddd-2}
\nr{\Psi_m \Ddd^2 v_m}_{L^2(\partial \R^\nn_+)} \limt m \infty 0.
\end{equation}

\stepp 
Given $\tilde q \in C_0^\infty(\Wc_\rho \times \R^n)$ there exists $q \in C_0^\infty(\Vc_\rho \times \R^n)$ such that 
\[
\frac i {h_m} \innp{ [-h_m^2 \D , \Opwm(\tilde q)] \tum}{\tum} = \frac i {h_m} \innp{ [P_m , \Opwm(q)] \tvm}{\tvm} + \bigo m \infty (h_m).
\]
See for instance Theorem 9.3 in \cite{zworski}. We deduce in particular that \eqref{lim-semiclass-measure} holds with $\tum$ and $\m$ replaced by $\tvm$ and some measure $\n$ on $\R^{2\nn}$, respectively. Now we have to prove that for all $q \in C_0^\infty(\Vc_\rho \times \R^n)$
\begin{equation} \label{eq-nu}
\int_{\R^{2\nn}} \{ p,q\} \, d\n = 0.
\end{equation}
As for Lemma \ref{lem-loc-energy} we can first check that $\n$ is supported in $p\inv(\singl 1)$.

\stepp 
Now assume that \eqref{eq-nu} holds if $q$ is replaced by any function $\f$ of the form $\f : (y,\y) \mapsto  \y_\nn \f_1(y,\y') + \f_0(y,\y')$ where $\f_0$ and $\f_1$ belong to $C_0^\infty(\Vc_\rho \times \R^{\nn-1})$. 
Let $q \in C_0^\infty(\Vc_\rho \times \R^n)$. Let $R' > 0$ be such that $\supp(q) \subset \Vc_{\rho} \times B'(0,R') \times \R$. We set $K' = \bar {\Vc_{2\rho}\times B'(0,2R')}$. Let $R_\nn > 0$ be such that $(K' \times \R) \cap p\inv (\singl 1) \subset K' \times ]-R_\nn,R_\nn[$. We set $K = K' \times [-R_\nn,R_\nn]$.

According to the Weierstrass density theorem, there exist sequences of polynomials $(q_j)$, $(b_j)$ and $(c_j)$ on $K$ which approach $q$, $b$ and $c$ in $C^1(K')$. Then in $C^1(K)$ we have
\[
p_j := \y_\nn^2 + b_j \y_\nn + c_j \limt j \infty p.
\]
Then for $j \in \N$ there exist polynomials $\tilde q_j$, $\f_{0,j}$ and $\f_{1,j}$ such that 
\[
q_j = (p_j-1) \tilde q_j + \f_{1,j} \y_\nn + \f_{0,j}.
\]
Let $\th \in C_0^\infty(\R^{2\nn-1})$ be supported in $\Vc_\rho \times \R^{\nn-1}$ and such that $q\th = q$. Then we have 
\begin{align*}
\lim_{j \to \infty } \int_{\R^{2\nn}}\{ p,  \th q_j \} \, d\nu 
& = \lim_{j \to \infty } \int_{\R^{2\nn}} \th \tilde q_j \{ p, (p_j-1) \} \, d\nu + \lim_{j \to \infty } \int_{\R^{2\nn}}  (p_j-1)  \{ p, \th \tilde q_j \} \, d\nu\\
& \quad   +\lim_{j \to \infty } \int_{\R^{2\nn}} \{ p,  \th(\y_\nn \f_{1,j} + \f_{0,j}) \} \, d\nu .
\end{align*}
Since $\{ p, p_j-1 \}$ goes to $\{p,p-1\} =0$ on $K$, $(p_j-1)$ goes to $(p-1)$, $p=1$ on the support of $\nu$ and according to the fact that \eqref{eq-crochet-poisson} holds for $ \th(\y_\nn \f_{1,j} + \f_{0,j})$ we obtain 
\begin{align*}
\int_{\R^{2\nn}}\{ p, \th  q_j \} \, d\nu \limt j \infty 0.
\end{align*}
On the other hand we have 
\[
\int_{\R^{2\nn}}\{ p, \th  q_j \} \, d\nu \limt j \infty \int_{\R^{2\nn}}\{ p, \th  q \} \, d\nu = \int_{\R^{2\nn}}  \{ p, q \} \, d\nu,
\]
so $q$ satisfies \eqref{eq-nu}. Thus it remains to prove \eqref{eq-nu} for a symbol like $\f$.

\stepp For the rest of the proof we fix two functions $\f_0,\f_1 \in C_0^\infty(\Vc_{\rho} \times \R^{\nn-1})$ and define $\f$ as above. For $m \in \N$ and $j\in\{0,1\}$ we set $\Phi_{j,m} = \Opwm(\f_j)$. This defines bounded operators on $L^2(\R^\nn)$. Since there symbols do not depend on $\y_\nn$ they can be seen as operators on $L^2(\R_+^\nn)$. Then we set $\Phi_m = \Phi_{1,m} \Ddd + \Phi_{0,m}$. We have 
\begin{equation} \label{Op-yphi}
\begin{aligned}
\Opwm(\y_\nn \f_1) = \Phi_{1,m} \Dddp - \frac {ih_m}2 \Phi'_{1,m} = \Dddp \Phi_{1,m} + \frac {ih_m}2 \Phi'_{1,m},
\end{aligned}
\end{equation}
where $\Phi'_{1,m} = \Opwm(\partial_{x_\nn} \f_1)$, and in particular $\Phi_m = \Opwm(\f) + O(h_m)$.
We consider $\th_1 \in C_0^\infty(\R,[0,1])$ equal to 1 on [-1,1]. For $r > 1$ we set $\th_r : \x \mapsto \th_1(\x/r)$. Since $\nu$ is supported in $p\inv(\singl 1)$ and $\abs{p(y,\y',\y_\nn)}$ goes to infinity when $\abs {\y_\nn}$ goes to infinity uniformly in $(y,\y')$ in the support of $\f_0$ or $\f_1$ we have for $r$ large enough
\begin{align*}
\int_{\R^{2\nn}} \{p,\f\} \, d\n
& = \int_{\R^{2\nn}} \th_r(\y_\nn) \{p,\f\} \, d\n\\
& = \lim_{m \to \infty} \frac i{h_m} \innp{\th_r(\Ddd) [P_m,\Opwm(\f)] \tvm}{\tvm}_{L^2(\R^\nn)}\\
& = \lim_{m \to \infty} \frac i{h_m} \innp{\th_r(\Ddd) [P_m,\Phi_m] \tvm}{\tvm}_{L^2(\R^\nn)}.
\end{align*}
Thus \eqref{eq-nu} will be a consequence of
\begin{equation} \label{eq-comm-Rplus}
\limsup_{m \to \infty} \nr{\frac i{h_m} \left( \th_r(\Ddd) [P_m,\Phi_m] \tvm -  \widetilde E \big([P_m,\Phi_m] v_m \big) \right)}_{L^2(\Vc_{2\rho})} \limt r \infty 0
\end{equation}
(where $\widetilde E \big([P_m,\Phi_m] v_m \big)$ is the extension by 0 of $[P_m,\Phi_m] v_m$ on $\R^\nn$) and
\begin{equation} \label{eq-comm-Rplus-suite}
\frac i{h_m} \innp{ [P_m,\Phi_m] v_m}{v_m}_{L^2(\Vc_{2\rho}^+)} \limt m \infty 0.
\end{equation}

\stepp
We begin with the proof of \eqref{eq-comm-Rplus}. We can write
\[
\frac i {h_m} [P_m, \Phi_m] = \sum_{j=0}^2 \Psi_{j,m} \Dddp^j
\]
where for $j \in \{0,1,2\}$ we have $\Psi_{j,m} = \Opwm(\p_{j,m})$ with $\p_{j,m} \in C_0^\infty(\Vc_{\rho} \times \R^{\nn-1})$ uniformly in $m$. In fact there is a term $\Psi_{3,m} \Dddp^3$ with $\Psi_{3,m} = h_m\inv [A(y),\Phi_{1,m}]$, but this term disappears by \eqref{A-equal-1}. This will be important to have terms of order at most 2 with respect to the last variable.

\stepp 
For $\k \in \N$ and $v \in C^\infty(\Vc_{2\rho}^+)$ we denote by $\widetilde {\Dddp^\k v}$ the function equal to $\Dddp^\k v$ on $\Vc_{2\rho}^+$ and equal to 0 on $\Vc_{2\rho} \setminus \Vc_{2\rho}^+$.
Let $\p \in C_0^\infty(\Vc_{2\rho} \times \R^{\nn-1})$. We set $\Psi_m = \Opwm(\p)$ and $\Psi'_m = \Opwm(\partial_{y_n} \p)$.
Let $s \in \big] 0 ,\frac 12 \big[$. For $k \in \{ 0, 1 \}$ and $\vf \in C_0^\infty(\Vc_{2\rho})$ we have by \eqref{estim-traces-vm}
\begin{eqnarray*}
\lefteqn{\abs{\innp{\pppg {\Ddd}^{s-1} \Psi_{m} \left( \Ddd  \widetilde{\Dddp^k v_m} -  \widetilde{\Dddp^{k+1} v_m} \right)} {\vf}_{L^2(\Vc_{2\rho})}}} \\
&& = \abs{\innp{ \widetilde{\Dddp^k v_m}} {\Dddp \Psi_{m}^* \pppg {\Ddd}^{s-1} \vf}_{L^2(\Vc_{2\rho})} - \innp{\widetilde{\Dddp^{k+1} v_m}} {\Psi_{m}^* \pppg {\Ddd}^{s-1} \vf}_{L^2(\Vc_{2\rho})} } \\
&& = \abs{\innp{\Dddp^k v_m} {\Dddp \Psi_{m}^* \pppg {\Ddd}^{s-1} \vf}_{L^2(\Vc_{2\rho}^+)} - \innp{\Dddp^{k+1} v_m} {\Psi_{m}^* \pppg {\Ddd}^{s-1} \vf}_{L^2(\Vc_{2\rho}^+)} } \\
&& \leq h_m\nr{\Dddp^k v_m}_{L^2(\partial \R_+^n)}\nr{\Psi_{m}^*\pppg {\Ddd}^{s-1} \vf}_{L^2(\partial \R_+^n)}\\
&& \lesssim h_m \nr{\Psi_{m}^*\pppg {\Ddd}^{s-1} \vf}_{H^{1-s}(\R^n)}\\
\end{eqnarray*}
For $\th \in \N$ (and hence for any $\th \geq 0$ by interpolation) we have 
\[
h^\th \nr{\pppg{\Ddd}^{-\th} \vf}_{H^\th(\R^\nn)} \lesssim \nr{\vf}_{L^2(\R^\nn)}.
\]
Applied with $\th = 1 - s$ we obtain
\begin{equation} \label{estim-hs}
\nr{\pppg {\Ddd}^{s-1} \left( \Psi_{m} \Ddd \widetilde{\Dddp^k v_m} - \Psi_{m} \widetilde{\Dddp^{k+1} v_m} \right)}_{L^2(\Vc_{2\rho})} = \bigo m \infty(h_m^s).
\end{equation}
For $j \in \Ii 0 2$ and $r > 0$ this yields in particular 
\begin{equation} \label{estim-Ddd}
\limsup _{m \to \infty} \nr{\th_r(\Ddd)\Psi_{j,m} \left(  \Dddp^j \widetilde{v_m} -  \widetilde{\Dddp^{j} v_m} \right)}_{L^2(\Vc_{2\rho})} = 0.
\end{equation}

\stepp 
For $j=\{0,1\}$ we use \eqref{Op-yphi} and \eqref{estim-hs} to write 
\begin{eqnarray*}
\lefteqn{\nr{\pppg {\Ddd}^s \Psi_{j,m} \widetilde{\Dddp^j v_m}}_{L^2(\R^\nn)}  \leq \nr{\pppg {\Ddd}^{s-2}  \Psi_{j,m}  \widetilde{\Dddp^j v_m}}_{L^2(\R^\nn)} + \nr{\pppg {\Ddd}^{s-2} \Dddp^2  \Psi_{j,m}   \widetilde{\Dddp^j v_m}}_{L^2(\R^\nn)}} \\
&& \lesssim \nr{\Dddp^j v_m}_{L^2(\Vc_{2\rho}^+)} + \nr{\pppg {\Ddd}^{s-1} \Psi_{j,m} \Ddd \widetilde{\Dddp^j v_m}}_{L^2(\R^\nn)} +  {h_m} \nr{\pppg {\Ddd}^{s-1} \Psi'_{j,m}  \widetilde{\Dddp^j v_m}}_{L^2(\R^\nn)}\\
&& \lesssim 1 + \nr{\pppg {\Ddd}^{s-1} \Psi_{j,m} \widetilde{\Dddp^{j+1} v_m}} + O(h_m).
\end{eqnarray*}
If $j = 0$ then with \eqref{estim-nabla-vm} this proves that $\pppg {\Ddd}^s \Psi_{0,m} \tvm$ is uniformly bounded in $L^2(\R^\nn)$. For $j = 1$ we also have to use \eqref{expr-gm}-\eqref{eq-Ddd2-vm} to conclude that $\pppg {\Ddd}^s \Psi_{1,m} \widetilde{\Ddd v_m}$ is uniformly bounded.
Thus for $j \in \{ 0,1\}$ we have by functional calculus
\begin{eqnarray} \label{estim-thr-nr}
\lefteqn{\limsup_{m \to \infty} \nr{\big( \th_r(\Ddd) - 1\big) \Psi_{j,m} \widetilde {\Dddp^j v_m}}}\\
\nonumber
&& \lesssim \limsup_{m \to \infty} \nr{\big( \th_r(\Ddd) - 1\big) \pppg {\Ddd}^{-s}}
\lesssim \sup_{\t \in \R} \abs{\frac {\th_r(\t)-1}{\t^s}} \limt r {+\infty} 0.
\end{eqnarray}
With \eqref{estim-Ddd} we deduce
\begin{equation} \label{mil00-lem12}
\begin{aligned}
\limsup_{m \to \infty} \nr{\th_r (\Ddd) \Psi_{j,m} \Dddp^j \tvm - \Psi_{j,m} \widetilde{ \Dddp^j v_m}}_{L^2(\R^\nn)} \limt r \infty 0.
\end{aligned}
\end{equation}
Now assume that $j = 2$. In order to prove \eqref{estim-thr-nr} we first apply \eqref{expr-gm}-\eqref{eq-Ddd2-vm} and then we use the cases $j=0$ and $j= 1$. Then \eqref{mil00-lem12} follows from \eqref{estim-Ddd} as before. Thus we have proved \eqref{mil00-lem12} for all $j \in \Ii 0 2$, and \eqref{eq-comm-Rplus} follows.

\stepp Now we turn to the proof of \eqref{eq-comm-Rplus-suite}. 
Assume that 
\begin{equation} \label{IPP-P}
- \frac i {h_m} \left( \innp{P_m \Phi_m v_m}{v_m} -  \innp{\Phi_m v_m}{P_m^* v_m} \right) \limt m \infty 0.
\end{equation} 
Then, since $P_m$ is formally self-adjoint, we have 
\begin{align*}
\limsup_{m\to \infty} \frac i{h_m} \innp{[P_m,\Phi_m] v_m}{v_m}_{\Vc_{2\rho}^+}
& = \limsup_{m\to \infty} \frac i{h_m} \left( \innp{\Phi v_m}{P_m^* v_m}_{\Vc_{2\rho}^+} -  \innp{\Phi P_m v_m}{ v_m}_{\Vc_{2\rho}^+}\right)\\
& = \limsup_{m\to \infty} \frac i{h_m} \left( \innp{\Phi v_m}{g_m}_{\Vc_{2\rho}^+} -  \innp{\Phi g_m}{ v_m}_{\Vc_{2\rho}^+}\right).
\end{align*}
We recall that $v_m$ is smooth on $\R_+^\nn$ so that $\Phi_m P_m v_m$ is well defined for all $m \in \N$.  Since $\f_0$ and $\f_1$ are supported in $\Vc_\rho$ and the derivatives of $\h$ are supported outside $\Vc_\rho$, we obtain \eqref{eq-comm-Rplus-suite} with \eqref{expr-gm}.
Thus it remains to prove \eqref{IPP-P}. We first observe that if $\seq {\tilde w} m$ and $\seq w m$ are sequences in $H^1(\Vc_{2\rho}^+)$ which go to 0 in $L^2(\partial \Vc_{2\rho}^+)$ then we have 
\begin{equation} \label{IPP-1}
-\frac i {h_m} \left( \innp{\Ddd \tilde w_m}{w_m}_{\Vc_{2\rho}^+} - \innp{ \tilde w_m}{\Ddd w_m}_{\Vc_{2\rho}^+} \right) =  \int_{\partial \Vc_{2\rho}^+} \tilde w_m \bar{w_m} \limt m \infty 0.
\end{equation}
If furthermore $\tilde w_m$ and $w_m$ are in $H^2(\Vc_{2\rho}^+)$ and are such that $\Ddd \tilde w_m$ and $\Ddd w_m$ go to 0 in $L^2(\partial \Vc_{2\rho}^+)$, then
\begin{equation*} 
-\frac i {h_m} \left( \innp{\Ddd^2 \tilde w_m}{w_m}_{\Vc_{2\rho}^+} - \innp{ \tilde w_m}{\Ddd^2 w_m}_{\Vc_{2\rho}^+} \right)
\limt m \infty 0.
\end{equation*}
With \eqref{estim-traces-vm} we directly obtain
\[
- \frac i {h_m} \left( \innp{P_m \Phi_{0,m} v_m}{v_m}_\OO -  \innp{\Phi_{0,m}v_m}{(P_m)^* v_m}_\OO \right) \limt m \infty 0 .
\]
By \eqref{Op-yphi}, \eqref{IPP-1}, \eqref{estim-traces-vm} and \eqref{trace-Ddd-2} we have 
\begin{eqnarray*}
\lefteqn{ -\lim_{m \to \infty} \frac i {h_m} \left(\innp{\Dddp^2 \Phi_{1,m} \Dddp v_m}{v_m} - \innp{\Phi_{1,m} \Dddp v_m}{ \Dddp^2  v_m} \right)}\\
&& =   - \lim_{m \to \infty} \left(\innp{  \Ddd \Phi'_{1,m} \Dddp v_m}{v_m} - \innp{\Phi'_{1,m} \Dddp v_m}{ \Ddd  v_m} \right) \\
&& \quad  - \lim_{m \to \infty}\frac i {h_m} \left(\innp{\Ddd \Phi_{1,m} \Dddp^2 v_m}{v_m} - \innp{\Phi_{1,m} \Dddp^2 v_m}{ \Ddd  v_m} \right)\\
&& = 0.
\end{eqnarray*}
This gives \eqref{IPP-P} with $P_m$ replaced by $\Dddp^2$. We proceed similarly for the other terms in $P_m$ (partial integrations with differential operators with respect to the first $\nn-1$ variables do not raise any problem). This concludes the proof of \eqref{IPP-P} and hence the proof of Proposition \ref{prop-prop-mesure}.
\end{proof}

Now we can conclude the proof of Theorem \ref{th-gap-Tah}:

\begin{proposition} \label{prop-mu-zero}
We have $\m = 0$.
\end{proposition}

\begin{proof}
This follows from the facts that $\m$ vanishes on a neighborhood of $\{ \y = 0 \}$ (see Lemma \ref{lem-loc-energy}), is invariant by the classical flow (see Proposition \ref{prop-prop-mesure}) and vanishes outside $\bar \o \times \R^\nn$.
\end{proof}

\section{The case of a one-dimensional section} \label{sec-sec-dim1}

\newcommand{\oo}{\ell}
\newcommand{\srev}{\theta}

In this appendix we give more precise information about the spectrum of $\Taz$ (see \eqref{def-Ta}-\eqref{dom-Ta}) in the case where the section $\o$ is of dimension 1. This continues the analysis of \cite[Section 3]{art-diss-schrodinger-guide}.\\

We assume that $\o = ]0,\oo[ \subset \R$ for some $\oo > 0$, and we set $\n = \pi / \oo$. In this case the operator $\Ta$ is given by the second derivative $- \frac {d^2} {dy^2}$ with domain
\[
\Dom(\Ta) = \singl{u \in H^2(0,\oo) \st u'(0) = - i \a u(0), u'(\oo) = i \a u(\oo)}.
\]
We recall from Proposition 3.1 in \cite{art-diss-schrodinger-guide} that for $\t \geq 0$ the spectrum of $\Tat$ is given by a sequence of simple eigenvalues $\l_n(a \t) = \srev_n(a\t)^2$ where the functions $\srev_n$, $n \in \N$, satisfy the following properties:
\begin{enumerate}[(i)]
\item For all $n \in \N$ we have $\srev_n(0) = n\n$.
\item For $n \in \N$ and $\a \geq 0$ we have 
\begin{equation} \label{eq-z}
\big( \a - \srev_n(\a) \big)^2 e^{2i \oo \srev_n(\a)} = \big( \a  + \srev_n(\a) \big)^2.
\end{equation}
\item For $n \in \N$ and $\a > 0$ we have $\Re(\srev_n(\a)) \in ]n\n ,(n+1) \n[$.
\item For $n \in \N$ there exists $C_n >0$ such that for all $\a > 0$ we have $\Im(\srev_n(\a)) \in [-C_n,0[$.
\item For all $n \in \N^*$ the map $\a \mapsto \srev_n(\a)$ depends analytically on $\a \geq 0$ (for $n = 0$, it is continuous on $\R_+$ and analytic on $\R_+^*$).
\end{enumerate}

In the following proposition we describe more precisely the behavior of the eigenvalues $\srev_n(\a)$ when $\a$ goes to $+\infty$. In particular, \eqref{lim-im-srev} shows that the spectrum of $\Ta$ approches the real axis for high frequencies. This is why it was only possible to give uniform estimates for $\Raz$ and hence $(\Ac-\t)$ in weighted spaces (see Section \ref{sec-high-freq-longitudinal}). The other properties of the proposition were not used in the paper. They are given for their own interests.

\begin{proposition} \label{prop-srev-agrand}
\begin{enumerate} [(i)]
\item Let $n\in\N$. Then the map $\a \mapsto \Re (\srev_n(\a))$ is increasing from $n\n$ to $(n+1)\n$ when $\a$ goes from 0 to $+\infty$.
\item For all $n \in \N$ we have
\begin{equation} \label{lim-im-srev}
\Im (\srev_n(\a)) \limt \a {+\infty} 0.
\end{equation}
\item We have
\[
\sup_{\a \in \R_+} \abs{\Im (\srev_n(\a))} = \bigo n \infty \big(\ln (n) \big).
\]
\item \label{estim-alog}
For $\b \in \R$ we have
\[
\Re \big(\srev_n \left( n\n + \b \ln (n) \right) \big) - n\n  \limt n \infty \n \left( 1 - \frac {\arg\left(\b + \frac i {\oo}\right)}\pi \right)
\]
(where $\arg\big( \b + \frac i \oo \big)$ belongs to $]0,\pi[$) and
\[
-\Im\big(\srev_n \left(n\n + \b \ln(n) \right) \big) \simm n \infty \frac {\ln (n)} {\oo}.
\]
\item \label{estim-an}
Let $\g \in \R_+^* \setminus \singl 1$. We have
\[
\Im \left( \srev_n \left(\g n \n \right) \right) \limt n \infty \frac 1{\oo} \ln \abs{\frac {1+\g}{1-\g}} \quad \text{and} \quad  \Re   \srev_n \left(\g n \n \right) - n\n  \limt n \infty \begin{cases} 0 & \text{if } \g <1, \\ \n & \text{if } \g > 1 .\end{cases}
\]

\item \label{estim-anrho}
Let $\rho \in ]0,1[$ and $s \in \R \setminus \singl 0$. Then 
\[
- \Im \left( \srev_n \left(n \n + s n^\rho \right) \right) \simm n \infty \frac {1-\rho} {\oo} \ln (n) \quad \text{and} \quad   \Re   \srev_n \left(n \n + s n^\rho \right) - n\n  \limt n \infty \begin{cases} 0 & \text{if } s < 0, \\ \n & \text{if } s > 0 .\end{cases} 
\]
\end{enumerate}

\end{proposition}

\begin{figure}[h]
\begin{minipage}[c]{0.45 \textwidth}
\includegraphics[width = \linewidth]{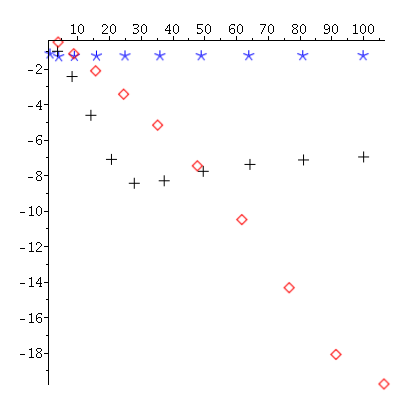}
\end{minipage}
\hfill 
\begin{minipage}[c]{0.45 \textwidth}
\includegraphics[width = \linewidth]{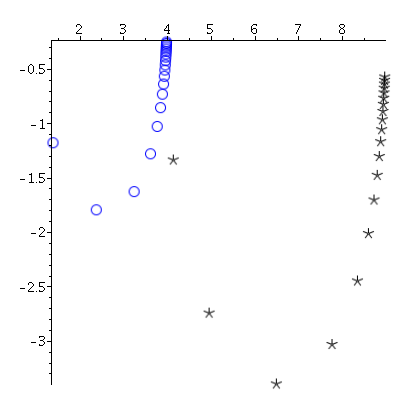}
\end{minipage}

\emph{On the left: the 20 first eigenvalues for $a = 1$ (asterisks), $a = 5$ (crosses) and $a = 10$ (diamonds). On the right: the first (circles) and second (asterisks) eigenvalues for $a$ going from 1 to 20 (from left to right).}

\caption{The eigenvalues of $\Ta$ when $\o = [0,1]$.} \label{fig-eigenvalues}

\end{figure}

\begin{proof}
\stepp Let $n \in \N$. Taking the derivative with respect to $\a$ in \eqref{eq-z} gives
\begin{align*}
2i\oo \srev '_n(\a) e^{2i \oo \srev _n(\a)}
= \frac d {d\a} \left(\frac {\a+\srev _n(\a)}{\a-\srev _n(\a)}\right)^2
= 4  \left(\frac {\a+\srev _n(\a)}{\a-\srev _n(\a)}\right)^2 \frac {-\srev _n(\a) + \a\srev '_n(\a)}{(a+\srev _n(\a))(\a-\srev _\n(a))},
\end{align*}
      \detail 
      {
      \begin{align*}
      2i\oo \srev '_n(a) e^{2i \srev _n(a)\oo}
      & = \frac d {da} \left(\frac {a+\srev _n(a)}{a-\srev _n(a)}\right)^2\\
      & = 2  \frac {a+\srev _n(a)}{a-\srev _n(a)} \frac {(1+\srev '_n(a)) (a-\srev _n(a)) - (1-\srev '_n(a)) (a+\srev _n(a))}{(a-\srev _n(a))^2} \\
      & = 2  \left(\frac {a+\srev _n(a)}{a-\srev _n(a)}\right)^2 \frac {(1+\srev '_n(a)) (a-\srev _n(a)) - (1-\srev '_n(a)) (a+\srev _n(a))}{(a+\srev _n(a))(a-\srev _n(a))} \\
      & 
      = 4  \left(\frac {a+\srev _n(a)}{a-\srev _n(a)}\right)^2 \frac {-\srev _n(a) + a\srev '_n(a)}{(a+\srev _n(a))(a-\srev _n(a))},
      \end{align*}
      }
and hence
\begin{equation*} 
\srev '_n(\a) =  \frac {2\srev _n(\a)}{2\a - i\oo(\a^2 - \srev _n(\a)^2)}  = \frac {2\srev _n(\a) \left(2\a+i \oo \big(\a^2 - \bar{\srev _n(\a)} ^2\big)\right)}{\abs{2\a-i\oo(\a^2-\srev _n(\a) ^2)}^2}.
\end{equation*}
In particular for $\a > 0$:
\[
\Re (\srev'_n(\a)) = \frac {4 \a \Re (\srev _n(\a)) - 2 \oo \big( \a^2 + \abs {\srev _n(\a)}^2 \big) \Im(\srev _n(\a))} {\abs{2\a-i\oo(\a^2-\srev ^2)}^2} > 0.
\]

\stepp We have 
\[
\abs{ \frac {\a + \srev _n(\a)}{\a - \srev _n(\a)} }^2 = e^{- 2 \oo \Im(\srev _n(\a))} .
\]
Assume by contradiction that we can find sequences $\seq n m \in \N^\N$ and $\seq \a m \in (\R_+^*)^\N$ such that if we set $\srev_m = \srev_{n_m}(\a_m)$ we have 
\[
\frac{\abs{\Im (\srev_m)}}{\ln(n_m)} \limt m \infty +\infty.
\]
Necessarily, $n_m$ goes to infinity when $m \to \infty$. If for some subsequence we have
\[
\frac {\a_m}{\abs{\srev_m}} \limt m \infty  0  \text{ or } +\infty, 
\]
then 
\[
e^{2 \oo \Im(\srev_m)} \limt m \infty 1.
\]
This gives a contradiction, so there exists $C > 1$ such that for all $m \in \N$ we have
\[
C \inv \leq \frac {\a_m}{\abs{\srev_m}} \leq C.
\]
Since $\Re(\srev_m)$ grows like $n_m \n$, we have in particular $\a_m \lesssim n_m + \abs{\Im(\srev_m)}$. 
Then
\[
e^{2\oo \abs{\Im(\srev_m)}} = \abs{ \frac {\a_m + \srev_m}{\a_m - \srev _m} }^2 \lesssim \frac {n_m^2 + \abs{\Im(\srev_m)}^2} {\abs{\Im(\srev_m)}^2}, 
\]
from which we deduce that $\abs{\Im(\srev_m)}$ cannot grow faster that $O(\ln(n_m))$ and get a contradiction.

\stepp We now turn to the third statement. For $n \in \N^*$ we can write
\[
\srev_n \big( n\n + \b \ln(n) \big) = n\n + R_n - i I_n
\]
with $R_n \in ]0, \n[$ and $I_n \geq 0$. We have 
\[
e^{2 i\oo R_n} e^{2 \oo I_n} = \left( \frac{2n\n + R_n - iI_n + \b \ln (n)} {R_n -i I_n - \b \ln(n)}\right) ^2.
\]
Then
\begin{align*}
I_n 
= \frac 1 {2\oo} \ln \abs {\frac{2n\n + R_n - iI_n + \b \ln (n)} {R_n -i I_n - \b \ln(n)}} ^2 \simm n \infty \frac {\ln (n)} {\oo}.
\end{align*}
On the other hand we have modulo $\pi$
\[
\oo R_n \equiv \frac 12  \arg \left( e^{2i\oo R_n} e^{2\oo I_n} \right) \equiv  \arg \left( \frac {2n\n }{(\b +\frac i {\oo}) \ln(n)} + \littleo n \infty (1) \right) \equiv - \arg \left(\b +\frac i {\oo}\right)  + \littleo n \infty (1).
\]
If we choose $\arg\big(\b + \frac i \oo \big)$ in $]0,\pi[$ we obtain
\[
R_n \limt n \infty \frac 1 {\oo} \left( \pi - \arg \left(\b +\frac i {\oo}  \right) \right).
\]
Since the map $\a \mapsto \Re(\srev_n(\a))$ is increasing on $\R_+$, we obtain in particular for all $n \in \N$
\[
\Re (\srev_n(\a)) \limt \a {+\infty} (n+1)\n.
\]

\stepp Now let $\g \in \R_+^* \setminus \singl 1$. Again we consider $R_n \in ]0, \n [$ and $I_n \geq 0$ such that 
\[
\srev_n \left( \g n \n \right) = n\n + R_n - i I_n.
\]
Then 
\[
e^{2i\oo R_n} e^{2\oo I_n} \limt n \infty \left( \frac {\g+1}{\g-1} \right)^2.
\]
This proves that $I_n \limt n \infty \frac 1{\oo} \ln \abs{\frac {1+\g}{1-\g}} $ and $ d(R_n,\n\N) \limt n \infty 0$. Using the fact that the real part of $\srev_n(a)$ is increasing we see that $R_n$ has to go to 0 for $\g < 1$ and to $\n$ if $\a > 1$. Finally, the results concerning $\a = n \n + s \n^\rho$ are proved similarly.
\end{proof}

\bigskip 

\noindent  
{\bf Acknowledgements: } This work is partially supported by the French ANR Project NOSEVOL (ANR 2011 BS01019 01).

\bibliographystyle{alpha}
\bibliography{bibliotex}

\begin{thebibliography}{{Wak}14}

\bibitem[AIK]{Aloui-Ib-Kh}
L.~Aloui, S.~Ibrahim, and M.~Khenissi.
\newblock Energy decay for linear dissipative wave equations in exterior
  domains.
\newblock Preprint arXiv:1503.0837.

\bibitem[AK02]{alouik02}
L.~Aloui and M.~Khenissi.
\newblock Stabilisation pour l'\'equation des ondes dans un domaine
  ext\'erieur.
\newblock {\em Rev. Math. Iberoamericana}, 18:1--16, 2002.

\bibitem[{Ana}10]{anantharaman10}
N.~{Anantharaman}.
\newblock {Spectral deviations for the damped wave equation.}
\newblock {\em {Geom. Funct. Anal.}}, 20(3):593--626, 2010.

\bibitem[BG88]{berger-gostiaux}
M.~Berger and B.~Gostiaux.
\newblock {\em Differential Geometry: Manifolds, Curves and Surfaces}.
\newblock Graduate Texts in Mathematics. Springer, 1988.

\bibitem[BGH11]{bonnetgh11}
A.-S. {Bonnet-Ben Dhia}, B.~{Goursaud}, and {Ch.} {Hazard}.
\newblock {Mathematical analysis of the junction of two acoustic open
  waveguides.}
\newblock {\em {SIAM J. Appl. Math.}}, 71(6):2048--2071, 2011.

\bibitem[BH12]{bonyh12}
J.-F. Bony and D.~H\"afner.
\newblock {Local Energy Decay for Several Evolution Equations on Asymptotically
  Euclidean Manifolds.}
\newblock {\em Annales Scientifiques de l' \'Ecole Normale Sup\'erieure},
  45(2):311--335, 2012.

\bibitem[BJ]{burq-joly}
N.~Burq and R.~Joly.
\newblock Exponential decay for the damped wave equation in unbounded domains.
\newblock {\em Communications in Contemporary Mathematics}.
\newblock To appear.

\bibitem[BK08]{borisovk08}
D.~Borisov and D.~{Krej\v ci\v r\'\i k}.
\newblock {PT}-symmectric waveguides.
\newblock {\em Integral Equations and Operator Theory}, 68(4):489--515, 2008.

\bibitem[BLR92]{bardoslr92}
C.~{Bardos}, G.~{Lebeau}, and J.~{Rauch}.
\newblock {Sharp sufficient conditions for the observation, control, and
  stabilization of waves from the boundary.}
\newblock {\em {SIAM J. Control Optim.}}, 30(5):1024--1065, 1992.

\bibitem[Bou11]{bouclet11}
J.-M. Bouclet.
\newblock Low frequency estimates and local energy decay for asymptotically
  {E}uclidean laplacians.
\newblock {\em Comm. Part. Diff. Equations}, 36:1239--1286, 2011.

\bibitem[BR14]{boucletr14}
J.-M. Bouclet and J.~Royer.
\newblock Local energy decay for the damped wave equation.
\newblock {\em Jour. Func. Anal.}, 266(2):4538--4615, 2014.

\bibitem[Bur98]{burq98}
N.~Burq.
\newblock {D\'ecroissance de l'\'energie locale de l'\'equation des ondes pour
  le probl\`eme ext\'erieur et absence de r\'esonance au voisinage du r\'eel.}
\newblock {\em Acta Math.}, 180(1):1--29, 1998.

\bibitem[CH04]{Chill-Ha-04}
R.~{Chill} and A.~{Haraux}.
\newblock {An optimal estimate for the time singular limit of an abstract wave
  equation.}
\newblock {\em {Funkc. Ekvacioj, Ser. Int.}}, 47(2):277--290, 2004.

\bibitem[DE95]{duclose95}
P.~{Duclos} and P.~{Exner}.
\newblock {Curvature-induced bound states in quantum waveguides in two and
  three dimensions.}
\newblock {\em {Rev. Math. Phys.}}, 7(1):73--102, 1995.

\bibitem[EN06]{engel2}
K.J. Engel and R.~Nagel.
\newblock {\em A Short Course on Operator Semigroups}.
\newblock Springer, 2006.

\bibitem[Gri85]{grisvard}
P.~Grisvard.
\newblock {\em Elliptic problems in nonsmooth domains}.
\newblock Pitman Advanced Publishing programs, 1985.

\bibitem[HO04]{hosonoo04}
T.~{Hosono} and T.~{Ogawa}.
\newblock {Large time behavior and $L^{p}$-$L^{q}$ estimate of solutions of
  2-dimensional nonlinear damped wave equations.}
\newblock {\em {J. Differ. Equations}}, 203(1):82--118, 2004.

\bibitem[IK85]{Isozaki-Ki-85-}
H.~{Isozaki} and H.~{Kitada}.
\newblock {A remark on the micro-local resolvent estimates for two body
  Schr\"odinger operators.}
\newblock {\em {Publ. Res. Inst. Math. Sci.}}, 21:889--910, 1985.

\bibitem[{Ike}02]{Ikehata-02}
R.~{Ikehata}.
\newblock {Diffusion phenomenon for linear dissipative wave equations in an
  exterior domain.}
\newblock {\em {J. Differ. Equations}}, 186(2):633--651, 2002.

\bibitem[ITY13]{ikehataty13}
R.~Ikehata, G.~Todorova, and B.~Yordanov.
\newblock Optimal decay rate of the energy for wave equations with critical
  potential.
\newblock {\em J. Math. Soc. Japan}, 65(1):183--236, 2013.

\bibitem[Jec04]{jecko04}
{\relax Th}.~Jecko.
\newblock {From classical to semiclassical non-trapping behaviour.}
\newblock {\em C. R., Math., Acad. Sci. Paris}, 338(7):545--548, 2004.

\bibitem[Jen85]{jensen85}
A.~Jensen.
\newblock Propagation estimates for {S}chr\"odinger-type operators.
\newblock {\em Trans. A.M.S.}, 291(1):129--144, 1985.

\bibitem[Kat80]{kato}
T.~Kato.
\newblock {\em Perturbation Theory for linear operators}.
\newblock Classics in Mathematics. Springer, second edition, 1980.

\bibitem[Khe03]{khenissi03}
M.~Khenissi.
\newblock {\'E}quation des ondes amorties dans un domaine ext\'erieur.
\newblock {\em Bull. Soc. Math. France}, 131(2):211--228, 2003.

\bibitem[KK05]{krejcirikk05}
D.~{Krej\v ci\v r\'\i k} and J.~{K\v r\'\i \v z}.
\newblock {On the spectrum of curved planar waveguides.}
\newblock {\em {Publ. Res. Inst. Math. Sci.}}, 41(3):757--791, 2005.

\bibitem[KR]{khenissir}
M.~Khenissi and J.~Royer.
\newblock Local energy decay and smoothing effect for the damped
  {S}chr\"odinger equation.
\newblock Preprint arXiv:1505.07200.

\bibitem[KR14]{krejcirikr14}
D.~{Krej\v{c}i\v{r}\'{\i}k} and N.~{Raymond}.
\newblock {Magnetic effects in curved quantum waveguides.}
\newblock {\em {Ann. Henri Poincar\'e}}, 15(10):1993--2024, 2014.

\bibitem[Leb96]{lebeau96}
G.~Lebeau.
\newblock {\'E}quation des ondes amorties.
\newblock {In : A. Boutet de Monvel and V. Marchenko (editors), \emph{Algebraic
  and geometric methods in mathematical physics}, 73-109. Kluwer Academic
  Publishers}, 1996.

\bibitem[LMP63]{laxmp63}
P.D. Lax, C.S. Morawetz, and R.S. Phillips.
\newblock Exponential decay of solutions of the wave equation in the exterior
  of a star-shaped obstacle.
\newblock {\em Comm. on Pure and Applied Mathematics}, 16:477--486, 1963.

\bibitem[LR97]{lebeaur97}
G.~{Lebeau} and L.~{Robbiano}.
\newblock {Stabilisation de l'\'equation des ondes par le bord.}
\newblock {\em {Duke Math. J.}}, 86(3):465--491, 1997.

\bibitem[Mel79]{melrose79}
R.~Melrose.
\newblock Singularities and energy decay in acoustical scattering.
\newblock {\em Duke Math. Journal}, 46(1):43--59, 1979.

\bibitem[Mel95]{melrose}
R.~Melrose.
\newblock {\em Geometric Scattering theory}.
\newblock Stanford lectures. Cambridge University press, 1995.

\bibitem[Mil00]{miller00}
L.~Miller.
\newblock Refraction of high-frequency waves density by sharp interfaces and
  semiclassical measures at the boundary.
\newblock {\em J. Math. Pures Appl.}, 79(3):227--269, 2000.

\bibitem[MN03]{marcatin03}
P.~{Marcati} and K.~{Nishihara}.
\newblock {The $L^{p}$--$L^{q}$ estimates of solutions to one-dimensional
  damped wave equations and their application to the compressible flow through
  porous media.}
\newblock {\em {J. Differ. Equations}}, 191(2):445--469, 2003.

\bibitem[MRS77]{morawetzrs77}
C.S. Morawetz, J.V. Ralston, and W.A. Strauss.
\newblock Decay of the solution of the wave equation outside non-trapping
  obstacles.
\newblock {\em Comm. on Pure and Applied Mathematics}, 30:447--508, 1977.

\bibitem[{Nar}04]{narazaki04}
T.~{Narazaki}.
\newblock {$L^p$-$L^q$ estimates for damped wave equations and their
  applications to semi-linear problem.}
\newblock {\em {J. Math. Soc. Japan}}, 56(2):585--626, 2004.

\bibitem[Nis]{nishiyama}
H.~Nishiyama.
\newblock Remarks on the asymptotic behavior of the solution to an abstract
  damped wave equation.
\newblock Preprint, arXiv:1505.01794.

\bibitem[{Nis}03]{nishihara03}
K.~{Nishihara}.
\newblock {$L^p$-$L^q$ estimates of solutions to the damped wave equation in
  3-dimensional space and their application.}
\newblock {\em {Math. Z.}}, 244(3):631--649, 2003.

\bibitem[Ral69]{ralston69}
J.~Ralston.
\newblock Solution of the wave equation with localized energy.
\newblock {\em Comm. on Pure and Applied Mathematics}, 22:807--823, 1969.

\bibitem[RCU13]{rabinovichcu13}
V.S. {Rabinovich}, R.~{Castillo-P\'erez}, and F.~{Urbano-Altamirano}.
\newblock {On the essential spectrum of quantum waveguides.}
\newblock {\em {Math. Methods Appl. Sci.}}, 36(7):761--772, 2013.

\bibitem[Roy]{art-dld-energy-space}
J.~Royer.
\newblock Local decay for the damped wave equation in the energy space.
\newblock Preprint, arXiv:1506.00377.

\bibitem[Roy14]{art-nondiss}
J.~Royer.
\newblock Resolvent estimates for a non-dissipative {H}elmholtz equation.
\newblock {\em Bulletin de la S.M.F.}, 142(4):591--633, 2014.

\bibitem[Roy15]{art-diss-schrodinger-guide}
J.~Royer.
\newblock Exponential decay for the {S}chr\"odinger equation on a dissipative
  wave guide.
\newblock {\em Ann. Henri Poincar\'e}, 16(8):1807--1836, 2015.

\bibitem[RS79]{rs4}
M.~Reed and B.~Simon.
\newblock {\em Method of Modern Mathematical Physics}, volume IV, Analysis of
  Operator.
\newblock Academic Press, 1979.

\bibitem[RT74]{raucht74}
J.~Rauch and M.~Taylor.
\newblock Exponential decay of solutions to hyperbolic equations in bounded
  domains.
\newblock {\em Indiana Univ. Math. J.}, 24(1):79--86, 1974.

\bibitem[RTY10]{Radu-To-Yo-11}
P.~{Radu}, G.~{Todorova}, and B.~{Yordanov}.
\newblock {Decay estimates for wave equations with variable coefficients.}
\newblock {\em {Trans. Am. Math. Soc.}}, 362(5):2279--2299, 2010.

\bibitem[RTY16]{Radu-To-Yo-16}
P.~{Radu}, G.~{Todorova}, and B.~{Yordanov}.
\newblock The generalized diffusion phenomenon and applications.
\newblock {\em SIAM J. Math. Anal.}, 48(1):174--203, 2016.

\bibitem[Sj{\"o}00]{sjostrand00}
J.~Sj{\"o}strand.
\newblock Asymptotic distribution of eigenfrequencies for damped wave
  equations.
\newblock {\em Publ. RIMS, Kyoto Univ.}, 36:573--611, 2000.

\bibitem[Str07]{strauss}
W.~A. Strauss.
\newblock {\em Partial Differential Equation, an Introduction}.
\newblock Wiley, 2nd edition, 2007.

\bibitem[{Wak}14]{Wakasugi-14}
Y.~{Wakasugi}.
\newblock {On diffusion phenomena for the linear wave equation with
  space-dependent damping.}
\newblock {\em {J. Hyperbolic Differ. Equ.}}, 11(4):795--819, 2014.

\bibitem[Wan87]{wang87}
X.~P. Wang.
\newblock Time-decay of scattering solutions and classical trajectories.
\newblock {\em Annales de l'I.H.P., section A}, 47(1):25--37, 1987.

\bibitem[Wan88]{wang88}
X.~P. Wang.
\newblock Time-decay of scattering solutions and resolvent estimates for
  semiclassical {S}chr\"odinger operators.
\newblock {\em Jour. Diff. Equations}, 71:348--395, 1988.

\bibitem[Zwo12]{zworski}
M.~Zworski.
\newblock {\em Semiclassical Analysis}, volume 138 of {\em Graduate Studies in
  Mathematics}.
\newblock American Mathematical Society, 2012.

\end{thebibliography}

\end{document}